\documentclass[a4paper,UKenglish,cleveref, autoref, thm-restate]{lipics-v2021}

\usepackage{graphicx}
\usepackage{microtype}
\usepackage{url}
\usepackage{algorithm}
\usepackage{algpseudocode}
\usepackage{comment}

\usepackage[titletoc]{appendix}

\newcommand{\newsentence}[1]{
#1
}
\newcommand{\oldsentence}[1]{
}

\title{An Optimal-Time RLBWT Construction in BWT-runs Bounded Space} 

\titlerunning{An Optimal-Time RLBWT Construction in BWT-runs Bounded Space} 

\author{Takaaki Nishimoto}{RIKEN Center for Advanced Intelligence Project, Japan}{takaaki.nishimoto@riken.jp}{}{}
\author{Shunsuke Kanda}{RIKEN Center for Advanced Intelligence Project, Japan}{shnsk.knd@gmail.com}{}{}
\author{Yasuo Tabei}{RIKEN Center for Advanced Intelligence Project, Japan}{yasuo.tabei@riken.jp}{}{}
\authorrunning{T. Nishimoto, S. Kanda, and Y. Tabei}
\Copyright{Takaaki Nishimoto, Shunsuke Kanda, and Yasuo Tabei}
\ccsdesc[100]{Theory of computation~Data compression}

\keywords{lossless data compression, Burrows--Wheeler transform, highly repetitive text collections}

\category{} 

\nolinenumbers 

\EventEditors{John Q. Open and Joan R. Access}
\EventNoEds{2}
\EventLongTitle{42nd Conference on Very Important Topics (CVIT 2016)}
\EventShortTitle{CVIT 2016}
\EventAcronym{CVIT}
\EventYear{2016}
\EventDate{December 24--27, 2016}
\EventLocation{Little Whinging, United Kingdom}
\EventLogo{}
\SeriesVolume{42}
\ArticleNo{23}

\begin{document}

\maketitle

\begin{abstract}
The compression of highly repetitive strings (i.e., strings with many repetitions) has been a central research topic in string processing, and quite a few compression methods for these strings have been proposed thus far. 
Among them, an efficient compression format gathering increasing attention is the run-length Burrows--Wheeler transform (RLBWT), 
which is a run-length encoded BWT as a reversible permutation of an input string on the lexicographical order of suffixes. 
State-of-the-art construction algorithms of RLBWT have a serious issue with respect to (i) non-optimal computation time or (ii) a working space that is linearly proportional to the length of an input string.
In this paper, we present \emph{r-comp}, the first optimal-time construction algorithm of RLBWT in BWT-runs bounded space. 
That is, the computational complexity of r-comp is $O(n + r \log{r})$ time and $O(r\log{n})$ bits of working space for the length $n$ of an input string and the number $r$ of equal-letter runs in BWT. 
The computation time is optimal (i.e., $O(n)$) for strings with the property $r=O(n/\log{n})$, which holds for most highly repetitive strings. 
Experiments using a real-world dataset of highly repetitive strings show the effectiveness of r-comp with respect to computation time and space.
\end{abstract}

\newcommand{\floor}[1]{\left \lfloor #1 \right \rfloor}
\newcommand{\ceil}[1]{\left \lceil #1 \right \rceil}
\newcommand{\argmax}{\mathop{\rm arg~max}\limits}
\newcommand{\argmin}{\mathop{\rm arg~min}\limits}
\newcommand{\polylog}{\mathop{\rm polylog}\limits}

\newcommand{\PFP}{|\mathsf{PFP}|}

\newcommand{\LF}{\mathsf{LF}}
\newcommand{\occ}{\mathsf{occ}}
\newcommand{\rank}{\mathsf{rank}}
\newcommand{\inspos}{\mathsf{ins}}
\newcommand{\reppos}{\mathsf{rep}}


\newcommand{\graphds}{\mathsf{Grp}}
\newcommand{\update}{\mathsf{update}}
\newcommand{\fastUpdate}{\mathsf{fastUpdate}}
\newcommand{\balance}{\mathsf{balance}}
\newcommand{\splitop}{\mathsf{split}}

\newcommand{\move}{\mathsf{mov}}
\newcommand{\shfn}{\mathsf{sh}}
\newcommand{\eqfn}{\mathsf{eq}}

\newcommand{\deltagrp}{\Delta^{\mathsf{grp}}}
\newcommand{\deltabwt}{\Delta^{\mathsf{bwt}}}

\newcommand{\bmap}{\mathsf{bmap}}

\section{Introduction}\label{intro}
\emph{Highly repetitive strings} (i.e., strings including many repetitions) have become common in research and industry. 
For instance, the 1000 Genomes Project~\cite{1000Genomes} was established for the purpose of building a detailed catalogue of human genetic variation, and 
it has sequenced a large number of human genomes. 
Nowadays, approximately 60 billion pages are said to exist on the Internet, and large sections of those pages (e.g., version-controlled documents) are highly repetitive.
There is therefore a growing demand to develop scalable data compression for efficiently storing, processing, and analyzing a gigantic number of 
highly repetitive strings.

To fulfill this demand, quite a few data compression methods for highly repetitive strings have been developed. 
Examples are LZ77~\cite{LZ77}, grammar compression~\cite{DBLP:conf/dcc/LarssonM99,DBLP:journals/algorithms/FuruyaTNIBK20,DBLP:journals/tcs/Rytter03,DBLP:journals/tcs/Jez16}, block trees~\cite{Belazzougui20}, and many others~\cite{DBLP:conf/dcc/NishimotoT19,DBLP:journals/tit/NavarroOP21,DBLP:conf/esa/DinklageE0KP19}.
Among them, an efficient compression format gathering increased attention is the \emph{run-length Burrows--Wheeler transform} (RLBWT), which is a run-length encoded BWT~\cite{burrows1994block} as a reversible permutation of an input string on the lexicographical order of suffixes. 
Recently, researchers have focused on developing string processing methods such as locate query~\cite{GNPjacm19,DBLP:journals/tcs/BannaiGI20,DBLP:conf/cpm/BelazzouguiCGPR15,Nishimoto21-2}, 
document listing~\cite{DBLP:conf/spire/CobasM020}, and substring enumeration~\cite{Nishimoto21-1} on RLBWT. 
Although several algorithms for constructing the RLBWT from an input string have been proposed thus far, there is no prior work that achieves the computational  complexity of optimal time (i.e., time linearly proportional to the length of the input string) and BWT-runs bounded space (i.e., a working space linearly proportional to the number of equal-letter runs in the BWT and logarithmically proportional to the length of the input string).

\textbf{Contribution.}
We present \emph{r-comp}, the first construction algorithm of RLBWT that achieves optimal time and BWT-runs bounded space.  
R-comp directly constructs the RLBWT of an input string.
It reads one character of an input string at a time from the reversed string and gradually builds the RLBWT corresponding to the suffixes read so far. 
The state-of-the-art online construction methods~\cite{DBLP:journals/algorithmica/PolicritiP18,DBLP:journals/jda/OhnoSTIS18} use inefficient data structures such as dynamic wavelet trees and B-trees for inserting each character into the current RLBWT at an insertion position, which is the most time-consuming part in an RLBWT construction. 
We present a new \emph{divided BWT (DBWT)} representation of BWT and a new bipartite graph representation on DBWT called \emph{LF-interval graph} to speed up the construction of RLBWT. 
The DBWT and LF-interval graph are efficiently built while reading each character one by one, and they enable us to quickly compute 
an appropriate position for inserting each character into the current RLBWT of the string. 
Another remarkable property of r-comp is the ability to extend the RLBWT for a newly added character without rebuilding the data structures used in r-comp from the beginning.

As a result, the computational complexity of r-comp is 
$O(n + r \log r)$ time and $O(r\log{n})$ bits of working space for the length $n$ of an input string and the number $r$ of equal-letter runs in BWT. 
In particular, the computational complexity is optimal (i.e., $O(n)$) for strings with the property $r = O(n / \log n)$, which holds 
for most highly repetitive strings. 
We experimentally tested the ability of r-comp to compress various highly repetitive strings, and we show that 
r-comp performs better than other methods with respect to computation time and space.

\section{Related work}\label{sec:related_works}
\begin{table}[t]
    \footnotesize
    \vspace{-0.5cm}
    \caption{
    Summary of state-of-the-art RLBWT construction algorithms.
    The update time in the rightmost column is the time needed to construct a new RLBWT from the current RLBWT for a character newly added to the string. 
    The update time of r-comp is amortized.
    $T$ is an input string of alphabet size $\sigma$ and length $n$; 
    $\PFP$ is the size of the dictionary and factorization created by the prefix-free parsing of $T$~\cite{DBLP:journals/almob/BoucherGKLMM19}.
    }
    \vspace{-5mm}    
    \label{table:result} 
    \center{
    \scalebox{0.85}{
    \begin{tabular}{r||c|c|c|c}
Method & Type & Running time & Working space~(bits) & Update time \\ \hline
D. Belazzougui+~\cite{DBLP:journals/talg/BelazzouguiCKM20} & indirect & $O(n)$ & $O(n \log \sigma)$ & Unsupported \\ \hline 
J. Munro+~\cite{DBLP:conf/soda/MunroNN17} & indirect & $O(n)$ & $O(n \log \sigma)$ & Unsupported \\ \hline 
D.Kempa~\cite{DBLP:conf/soda/Kempa19} & indirect & $O(n / \log_{\sigma} n + r \log^{7} n)$ & $O(n \log \sigma + r \log^{6} n)$ & Unsupported \\ \hline 
D.Kempa+\cite{DBLP:conf/stoc/KempaK19} & indirect & $O(n \log \sigma / \sqrt{\log n})$ & $O(n \log \sigma)$ & Unsupported \\ \hline 
Big-BWT~\cite{DBLP:journals/almob/BoucherGKLMM19} & indirect & $O(n)$ & $O(\PFP \log n)$ & Unsupported \\ \hline 
KK method~\cite{DBLP:conf/focs/KempaK20,DBLP:journals/dam/NishimotoIIBT20} & direct & $O(n (\log \log n)^2 + r \log^{8} n)$ & $O(r \polylog n)$ & Unsupported \\ \hline 
PP method~\cite{DBLP:journals/algorithmica/PolicritiP18} & direct & $O(n \log r)$ & $O(r \log n)$ & $O(\log r)$ \\ \hline 
Faster-PP method~\cite{DBLP:journals/jda/OhnoSTIS18} & direct & $O(n \log r)$ & $O(r \log n)$ & $O(\log r)$ \\ \hline \hline 
r-comp (this study) &direct & $O(n + r \log r)$ & $O(r \log n)$ & $O(1 + (r \log r) / n )$ 
    \end{tabular} 
    }
    }
\end{table}

There are two types of methods for indirectly or directly constructing the RLBWT of a string (see Table~\ref{table:result} for a summary of state-of-the-art construction algorithms of RLBWT).
In the indirect constructions of RLBWT, the BWT of an input string is first built and then the BWT is encoded into the RLBWT by run-length encoding. 
Several efficient algorithms for constructing the BWT of a given string have been proposed ~\cite{DBLP:journals/talg/BelazzouguiCKM20,DBLP:conf/stoc/KempaK19,DBLP:journals/siamcomp/NavarroN14,DBLP:journals/jda/CrochemoreGKL15,DBLP:conf/stoc/KempaK19,DBLP:conf/soda/MunroNN17}.  
Let $T$ be a string of length $n$ with an alphabet of size $\sigma$, and let $r$ be the number of equal-letter runs in its BWT.
Kempa~\cite{DBLP:conf/soda/Kempa19} proposed a RAM-optimal time construction of the BWT of string $T$ with compression ratio $n / r = \Omega(\polylog n)$. 
The algorithm runs in $O(n / \log_{\sigma} n)$ time with $O(n \log \sigma)$ bits of working space. 
Kempa and Kociumaka also proposed a BWT construction in $O(n \log \sigma)$ bits of working space~\cite{DBLP:conf/stoc/KempaK19}. 
This algorithm runs in $O(n \log \sigma / \sqrt{\log n})$ time, which is bounded by $o(n)$ time for a string with $\log \sigma = o(\sqrt{\log n})$. 
These algorithms are not space efficient for highly repetitive strings in that their working space is linearly proportional to the length of the input string.

Big-BWT~\cite{DBLP:journals/almob/BoucherGKLMM19} is a practical algorithm for constructing the BWT of a huge string using \emph{prefix-free parsing}, which constructs a dictionary of strings and a factorization from string $T$. 
Although Big-BWT runs in optimal time (i.e., $O(n)$) with $O(\PFP \log n)$ bits of working space for the sum $\PFP$ of (i) the lengths of all the strings in the dictionary and (ii) the number of strings in the factorization, Big-BWT is not space efficient for highly repetitive strings in the worst case, 
because $\PFP$ can be $\sqrt{n}$ times larger than $r$, resulting in $\Omega(r\sqrt{n}\log{n})$ bits of working space~(see Appendix~\ref{app:bigbwt} for the proof). 
Even worse, several data structures used in these indirect constructions cannot be updated. Thus, one needs to rebuild the data structures from scratch for a newly added character, which reduces the usability of indirect constructions of RLBWT. 

In the direct constructions of RLBWT, 
Policriti and Prezza~\cite{DBLP:journals/algorithmica/PolicritiP18} proposed an algorithm for the construction of RLBWT, which we call \emph{PP method}. 
The PP method reads an input string in reverse by one character, and it 
gradually builds the RLBWT corresponding to the suffix that was just read, where 
an inefficient dynamic wavelet tree is used for inserting a character into the RLBWT at an appropriate position, limiting the scalability of the PP method in practice. 
Ohno et al.~\cite{DBLP:journals/jda/OhnoSTIS18} proposed a faster method, which we call \emph{Faster-PP method}, by replacing the dynamic wavelet tree used in the PP method by a B-tree. 
Whereas both the PP method and Faster-PP method run with the same time and space complexities---$O(n \log r)$ time and $O(r \log n)$ bits of working space---the time complexity is not the optimal time for most highly repetitive strings.

Kempa and Kociumaka~\cite{DBLP:conf/focs/KempaK20} proposed a conversion algorithm, which is referred to as \emph{KK method}, from the LZ77 parsing~\cite{LZ77} of $T$ to the RLBWT in $O(z \log^{7} n)$ time with $O(z \polylog n)$ bits of space, where $z$ is the number of phrases in the parsing. 
Theoretically, we can compute the RLBWT of an input string by combining the KK method with an algorithm for computing the LZ77 parsing~(e.g., \cite{DBLP:journals/dam/NishimotoIIBT20}), 
and the working space of their conversion is bounded by $O(r \polylog n)$ bits because $z = O(r \log n)$~\cite{DBLP:journals/tit/NavarroOP21}. 
Kempa and Langmead~\cite{DBLP:conf/esa/KempaK21} proposed a practical algorithm for constructing a compressed grammar from an input string in $\Omega(n)$ time using an approximate LZ77 parsing. 
Because these methods use several static data structures that cannot be updated,  
the data structures must be rebuilt from scratch when a new character is added.

Although there are several algorithms for indirectly or directly constructing the RLBWT, no previous work has been able to achieve optimal time (i.e., $O(n)$ time) with BWT-runs bounded space (i.e., $O(r\log{n})$ bits). 
We present \emph{r-comp}, the first direct construction of RLBWT that achieves optimal time with BWT-runs bounded space for most highly repetitive strings.
Details of r-comp are presented in the following sections.

\newsentence\bgroup
This paper is organized as follows. 
Section~\ref{sec:preliminary} introduces basic notions used in this paper, and 
a DBWT representation of BWT is presented in Section~\ref{sec:dbwt}. 
Section~\ref{sec:lfgraph} presents an LF-interval graph representation of DBWT and a fast update operation on LF-interval graphs. 
The r-comp algorithm is presented in Section~\ref{sec:rcomp_algorithm}. 
Section~\ref{sec:exp} presents the experimental results using the r-comp algorithm on benchmark and real-world datasets of highly repetitive strings. 
\egroup
\section{Preliminaries}\label{sec:preliminary}
\begin{figure}[t]
 \begin{center}
		\includegraphics[scale=0.7]{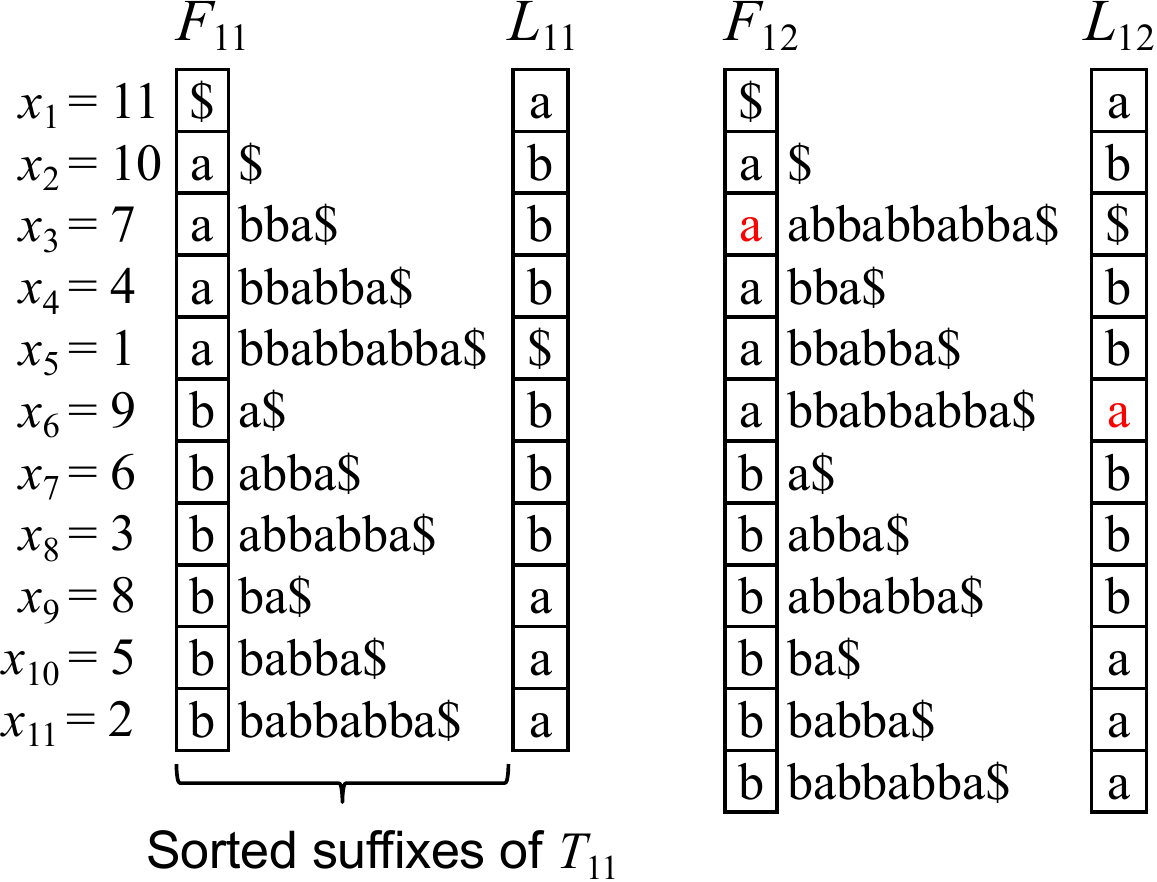}

	  \caption{
	  (Left) Sorted suffixes of $T_{11} = abb abba bba\$$, $F_{11}$, and $L_{11}$. 
	  (Right) Sorted suffixes of $T_{12} = aabb abba bba\$$, $F_{12}$, and $L_{12}$.
	  }
 \label{fig:bwt}
 \end{center}
\end{figure}

\subparagraph{Basic notation.}
An \emph{interval} $[b, e]$ for two integers $b$ and $e$~($b \leq e$) represents the set $\{b, b+1, \ldots, e \}$. 
Let $T$ be a string of length $n$ over an alphabet $\Sigma = \{ 1, 2, \ldots, n^{O(1)} \}$ of size $\sigma$, 
and $|T|$ be the length of $T$~(i.e., $|T| = n$). 
Let $T[i]$ be the $i$-th character of $T$~(i.e., $T = T[1], T[2], \ldots, T[n]$) and 
$T[i..j]$ be the substring of $T$ that begins at position $i$
and ends at position $j$. 
Let $T_{\delta}$ be the suffix of $T$ of length $\delta$ $(1 \leq \delta \leq n)$, i.e., $T_{\delta} = T[(n-\delta+1)..n]$. 
A \emph{rank query} $\rank(T, c, i)$ on a string $T$ returns the number of occurrences of character $c$ in $T[1..i]$, i.e., $\rank(T, c, i) = |\{ j \mid T[j] = c, 1 \leq j \leq i \}|$. 

For a string $P$, $P[i] < P[j]$ means that the $i$-th character of $P$ is smaller than the $j$-th character of $P$. Moreover, $T \prec P$ means that $T$ is lexicographically smaller than $P$. Formally, $T \prec P$ if and only if either of the following two conditions holds: 
(i) there exists an integer $i$ such that $T[1..i-1] = P[1..i-1]$ and $T[i] < P[i]$; (ii) $T$ is a prefix of $P$~(i.e., $T = P[1..|T|]$) and $|T| < |P|$. 
Here, $\occ_{<}(T, c)$ denotes the number of characters smaller than character $c$ in string $T$~(i.e., $\occ_{<}(T, c) = |\{ j \mid j \in \{ 1, 2, \ldots, n \} \mbox{ s.t. } T[j] < c \}|$). 
Special character $\$$ is the smallest character in $\Sigma$. 
Throughout this paper, we assume that special character $\$$ only appears at the end of $T$~(i.e., $T[n] = \$$ and $T[i] \neq \$$ for all $\{ 1, 2, \ldots, n - 1 \}$).

A \emph{run} is defined as the maximal repetition of the same character. 
Formally, a substring $T[i..j]$ of $T$ is a \emph{run} of the same character $c$ if it satisfies the following three conditions: 
(i) $T[i..j]$ is a repetition of the same character $c$ (i.e., $T[i]=T[i+1]= \cdots = T[j]=c$); (ii) $i = 1$ or $T[i-1] \neq c$; (iii) $j = n$ or $T[j+1] \neq c$.

We use base-2 logarithm throughout this paper. 
Our computation model is a unit-cost word RAM with a machine word size of $\Theta(\log n)$ bits. 
We evaluate the space complexity in terms of the number of machine words. 
A bitwise evaluation of space complexity can be obtained with a $\log n$ multiplicative factor. 

\subparagraph{BWT, LF function, and RLBWT.}
The BWT~\cite{burrows1994block} of a suffix $T_{\delta}$ is a permuted string $L_{\delta}$ of $T_{\delta}$,
and it is constructed as follows: all the suffixes of $T_{\delta}$ are sorted in the lexicographical order and the character preceding each suffix is taken. 
Formally, let $x_{1}, x_{2}, \ldots, x_{\delta}$ be 
the starting positions of the sorted suffixes of $T_{\delta}$~(i.e., $x_{1}, x_{2}, \ldots, x_{\delta}$ are a permutation of sequence $1, 2, \ldots, \delta$ such that $T_{\delta}[x_{1}..\delta] \prec T_{\delta}[x_{2}..\delta] \prec \cdots \prec T_{\delta}[x_{n}..\delta]$).
Then, $L_{\delta} = T_{\delta}[x_{1}-1], T_{\delta}[x_{2}-1], \ldots, T_{\delta}[x_{\delta}-1]$), where 
$T_{\delta}[0]$ is defined as the last character of $T_{\delta}$~(i.e., $T_{\delta}[0] = T_{\delta}[\delta] = \$$). 
Similarly, the permuted string $F_{\delta}$ of suffix $T_{\delta}$ consists of the first characters of the sorted suffixes of $T_{\delta}$, i.e., $F_{\delta} = T_{\delta}[x_{1}], T_{\delta}[x_{2}], \ldots, T_{\delta}[x_{\delta}]$. 

\newsentence{
Figure~\ref{fig:bwt} illustrates the sorted suffixes of $T_{11}$ and $T_{12}$ for $T = aabb abba bba\$$. 
Here, $x_{1}, x_{2}, \ldots, x_{11}$ are the starting positions of the sorted suffixes of $T_{11}$. 
Moreover, $L_{11} = abbb \$ bbb aaa$ and $F_{11} = \$ aaaa bbbbbb$. 
The BWT of $T$ is $L_{12} = a b \$ bb a bbb aaa$.
}

There is a one-to-one correspondence between $L_{\delta}$ and $F_{\delta}$ because the two strings are permutations of $T_{\delta}$. 
Formally, for two integers $i, j \in \{ 1, 2, \ldots, \delta \}$, $L_{\delta}[i]$ corresponds to $F_{\delta}[j]$ if and only if 
either of the following two conditions holds: 
(i) $x_{i}-1 = x_{j}$ or (ii) $x_{i} = 1$ and $x_{j} = \delta$. 
\emph{LF function} $\LF_{\delta}$ is a bijective function from $L_{\delta}$ to $F_{\delta}$ such that 
for two integers $i, j \in \{ 1, 2, \ldots, \delta \}$~\cite{DBLP:conf/focs/FerraginaM00}. 
Function $\LF_{\delta}(i) = j$ if and only if $L_{\delta}[i]$ corresponds to $F_{\delta}[j]$. 
\emph{LF formula}~\cite{DBLP:conf/focs/FerraginaM00} is a well-known property of LF function, and it enables us to compute 
the corresponding position in $F_\delta$ from a position in $L_\delta$. 
Namely, $\LF_{\delta}(i)$ is equal to 
the summation of (i) the number of characters in $L_{\delta}$ smaller than $L_{\delta}[i]$ and (ii) the number of $L_{\delta}[i]$ in the prefix $L_{\delta}[1..i]$, 
i.e., $\LF_{\delta}(i) = \occ_{<}(L_{\delta}, L_{\delta}[i]) + \rank(L_{\delta}, i, L_{\delta}[i])$. 

\newsentence{
In Figure~\ref{fig:bwt}, the red character $a$ in $L_{12}$ corresponds to that in $F_{12}$, and hence $\LF_{12}(6) = 3$. 
In addition, $\occ_{<}(L_{12}, b) = 6$ and $\rank(L_{12}, 7, b) = 4$, and hence $\LF_{12}(7) = 6 + 4 = 10$ by the LF formula. 
}

BWT can be separated into all the runs of the same character. We call each run BWT-run. 
For BWT $L_{\delta}$, $r$ BWT-runs $P_{1}, P_{2}, \ldots, P_{r}$ satisfy (i) $L_{\delta}= P_{1}, P_{2}, \ldots, P_{r}$ and 
(ii) each $P_{i}$ $(i=1, 2, \ldots, r)$ is a run of the same character in $L_{\delta}$. 
The RLBWT of a suffix $T_{\delta}$ is defined as a sequence of $r$ pairs $(P_{1}[1], |P_{1}|)$, $(P_{2}[1], |P_{2}|)$, $\ldots$, $(P_{r}[1], |P_{r}|)$. 
The RLBWT can be stored in $r(\log n + \log \sigma)$ bits, 
and we can recover $T_{\delta}$ from the RLBWT using LF function~(e.g., \cite{Nishimoto21-2}). 
Throughout this paper, $r$ denotes the number of BWT-runs in the BWT of $T$. 
\oldsentence{
Examples of BWT, LF function, and RLBWT can be found in Appendix~\ref{app:preliminary:examples}. 
}

\newsentence{
In Figure~\ref{fig:bwt}, 
the BWT-runs in 
the BWT $L_{11}$ of $T_{11}$ are $a, bbb, \$, bbb$, and $aaa$. 
The RLBWT of $T_{11}$ is $(a, 1)$, $(b, 3)$, $(\$, 1)$, $(b, 3)$, and $(a, 3)$. 
}

\section{DBWT}\label{sec:dbwt}
\newsentence{
\begin{figure}[t]
 \begin{center}
		\includegraphics[scale=0.8]{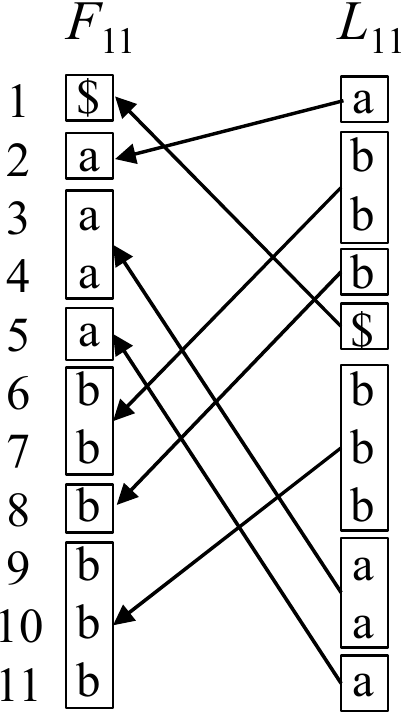}
	  \caption{
	  DBWT-repetitions and their corresponding F-intervals on $F_{11}$ for a DBWT $D_{11} = a, bb, b, \$, bbb, aa, a$ of BWT $L_{11}$ in Figure~\ref{fig:bwt}. 
	  Each rectangle on $L_{11}$ represents a DBWT-repetition, 
	  and each rectangle on $F_{11}$ represents an F-interval on $F_{11}$.
	  Each directed arrow indicates the F-interval corresponding to the DBWT-repetition on $L_{11}$. 	  
	  }
 \label{fig:dbwt}
 \end{center}
\end{figure}
}

The divided BWT (DBWT) is a general concept in the RLBWT and is the foundation of the LF-interval graph. 
Formally, the DBWT $D_{\delta}$ of BWT $L_{\delta}$ is defined as a sequence 
$L_{\delta}[p_{1}..(p_{2}-1)], L_{\delta}[p_{2}..(p_{3}-1)], \ldots, L_{\delta}[p_{k}..(p_{k+1}-1)]$ 
for $p_{1} = 1 < p_{2} < \cdots < p_{k} < p_{k+1} = n+1$, where $L_{\delta}[p_{i}..(p_{i+1}-1)]$ for each $i=1,2, \ldots ,k$ is a repetition of the same character.
We call each repetition of the same character in the DBWT \emph{DBWT-repetition}. 
A DBWT-repetition is not necessarily a run. 
DBWT $D_{\delta}$ is equal to the RLBWT of $T_{\delta}$ 
if and only if $L_{\delta}[p_{i}..(p_{i}-1)]$ for each $i \in \{1,2, \ldots ,k\}$ is a run. 

\newsentence{
In Figure~\ref{fig:dbwt}, sequence $D_{11} = a, bb, b, \$, bbb, aa, a$ of equal-letter repetitions is a DBWT for BWT $L_{11}$ of string $T_{11}$. 
The DBWT-repetitions in $D_{11}$ are the strings enclosed by the rectangles on $L_{11}$. 
}

The DBWT of a BWT is not unique because a BWT can be divided by various criteria. 
We present a criterion for DBWT in the following in order to efficiently build RLBWT. 
The LF function maps each DBWT-repetition $L_{\delta}[p_{i}..(p_{i+1}-1)]$ into the consecutive characters on interval $[\LF_{\delta}(p_{i}), \LF_{\delta}(p_{i+1} - 1)]$, which is called an \emph{F-interval} on $F_\delta$.
The LF formula enables us to compute $\LF_{\delta}(j)$ for each position $j \in [p_{i}, (p_{i+1}-1)]$ on the $i$-th DBWT-repetition in $O(1)$ time 
using the starting position $p_{i}$ of the DBWT-repetition and the F-interval $[\LF_{\delta}(p_{i}), \LF_{\delta}(p_{i+1} - 1)]$ corresponding to the DBWT-repetition as follows: $\LF_{\delta}(j) = \LF_{\delta}(p_{i}) + j - p_{i}$. 

\newsentence{
In Figure~\ref{fig:dbwt}, each F-interval on $F_{11}$ corresponding to a DBWT-repetition on $L_{11}$ is enclosed by a rectangle. 
The $F$-intervals on $F_{11}$ are $[1, 1]$, $[2, 2]$, $[3, 4]$, $[5, 5]$, $[6, 7]$, $[8, 8]$, and $[9, 11]$. The F-interval corresponding to the second DBWT-repetition $bb$ is $[6, 7]$. 
}

Let $\alpha$ be a user-defined parameter no less than $2$ (i.e., $\alpha \geq 2$).
DBWT-repetition $L_{\delta}[p_{i}..(p_{i+1}-1)]$ is said to \emph{cover} the starting position $\LF_{\delta}(p_{j})$ of an F-interval $[\LF_{\delta}(p_{j}), \LF_{\delta}(p_{j+1})]$ on $F_{\delta}$  
if interval $[p_{i}, (p_{i+1}-1)]$ on $F_{\delta}$ contains the position $\LF_{\delta}(p_{j})$~(i.e., $\LF_{\delta}(p_{j}) \in [p_{i}, (p_{i+1}-1)]$). 
The DBWT-repetition is said to be $\alpha$-\emph{heavy} if it covers at least $\alpha$ starting positions of the F-intervals on $F_{\delta}$ for parameter $\alpha \geq 2$.
Similarly, F-interval $[\LF_{\delta}(p_{i}), \LF_{\delta}(p_{i+1} - 1)]$ on $F_\delta$ is said to 
cover the starting position $p_{j}$ of a DBWT-repetition $L_{\delta}[p_{j}..(p_{j+1}-1)]$ 
if interval $[\LF_{\delta}(p_{i}), \LF_{\delta}(p_{i+1} - 1)]$ on $L_{\delta}$ contains the position $p_{j}$~(i.e., $p_{j} \in [\LF_{\delta}(p_{i}), \LF_{\delta}(p_{i+1} - 1)]$). 
An F-interval is said to be $\alpha$-heavy if it covers at least $\alpha$ starting positions of DBWT-repetitions for parameter $\alpha \geq 2$.
A DBWT is said to be $\alpha$-\emph{balanced} if the DBWT includes neither $\alpha$-heavy DBWT-repetitions nor F-intervals, and $D^{\alpha}_{\delta}$ denotes an $\alpha$-balanced DBWT of BWT $L_{\delta}$. 
\oldsentence{
Examples of DBWT, F-intervals, and $\alpha$-balanced DBWT can be found in Appendix~\ref{app:graph:examples}. 
}

\newsentence{
In Figure~\ref{fig:dbwt} with $\alpha = 3$, 
the fifth DBWT-repetition $bbb$ of $D_{11}$ covers two starting positions of F-intervals $[6, 7]$ and $[8, 8]$ on $F_{11}$, and  
the DBWT-repetition is not $3$-heavy. 
The F-interval $[9, 11]$ of the fifth DBWT-repetition covers the starting positions of two DBWT-repetitions $aa$ and $a$. 
Moreover, the F-interval of the fifth DBWT-repetition is not $3$-heavy. 
Thus, $D_{11}$ is $3$-balanced because $D_{11}$ includes neither $3$-heavy DBWT-repetitions nor F-intervals. 
}

In the next section, the $\alpha$-balanced DBWT is used to derive the time needed to update an LF-interval graph.

\section{LF-interval graph}\label{sec:lfgraph}
\newsentence{
\begin{figure}[t]
 \begin{center}
		\includegraphics[scale=0.7]{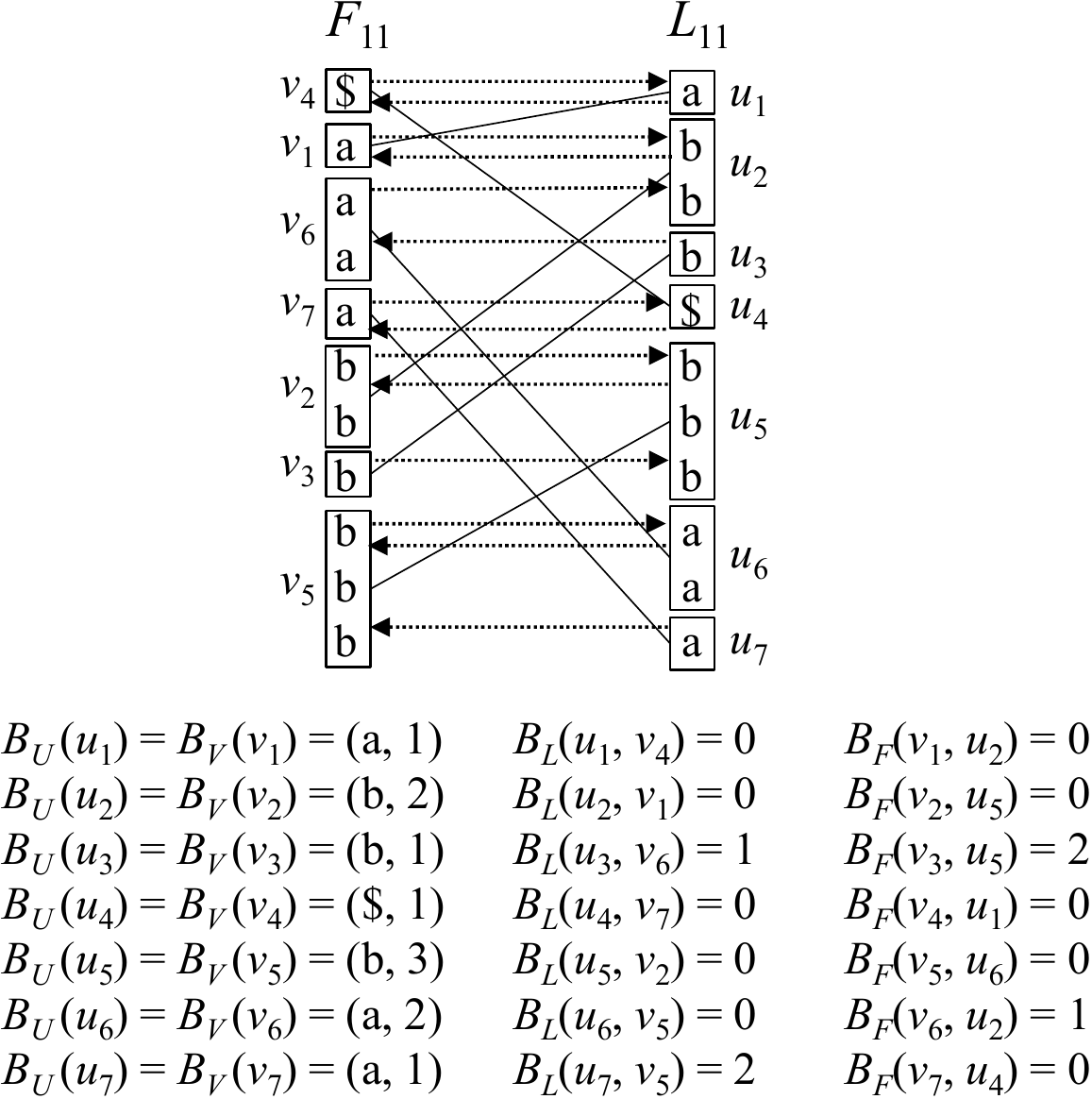}
	  \caption{
	  LF-interval graph $\graphds(D_{11})$ for DBWT $D_{11}$ in Figure~\ref{fig:dbwt}. 
	  }
 \label{fig:graph}
 \end{center}
\end{figure}
}

An LF-interval graph is a bipartite graph that represents both (i) the correspondence  between each pair of elements in $L_{\delta}$ 
and $F_{\delta}$ according to the LF function and (ii) a covering relationship between DBWT-repetitions and F-intervals on a DBWT. 
The LF-interval graph $\graphds(D_{\delta})$ for DBWT $D_{\delta}$ of $k$ DBWT-repetitions $L_{\delta}[p_{1}..(p_{2}-1)], L_{\delta}[p_{2}..(p_{3}-1)], \ldots, L_{\delta}[p_{k}..(p_{k+1}-1)]$ is defined as 4-tuple $(U \cup V$, $E_{LF} \cup E_{L} \cup {E}_{F}$, $B_{U} \cup B_{V}$,  $B_{L} \cup B_{F})$, as detailed in the following.

Set $U = \{ u_{1}, u_{2}, \ldots, u_{k} \}$ is a set of nodes, and $u_{i}$ for each $i\in \{1,2, \ldots, k\}$ represents the $i$-th DBWT-repetition $L_{\delta}[p_{i}..(p_{i+1}-1)]$ on DBWT $D_{\delta}$. 
Moreover, set $V = \{ v_{1}, v_{2}, \ldots, v_{k} \}$ is a set of nodes, and $v_i$ for each $i\in \{1,2, \ldots, k\}$ represents the F-interval $[\LF_{\delta}(p_{i}), \LF_{\delta}(p_{i+1} - 1)]$ mapped from the $i$-th DBWT-repetition represented as $u_i$ on DBWT $D_{\delta}$ by the LF function. 

The set $E_{LF}$ of undirected edges in LF-interval graph $\graphds(D_{\delta})$ represents 
the correspondence between DBWT-repetitions on DBWT $D_{\delta}$ and F-intervals on $F_{\delta}$ according to the LF function. 
Formally, $E_{LF} \subseteq (U \times V)$ is a set of undirected edges between $U$ and $V$, and 
$(u_i,v_j) \in E_{LF}$ holds if and only if the $i$-th DBWT-repetition $L_\delta[p_i..(p_{i+1}-1)]$ represented as $u_i$ is 
mapped to the $j$-th F-interval $[ \LF_\delta(p_j), \LF_\delta(p_{j+1}-1)]$ represented as $v_j$. 
Namely, $E_{LF} = \{ (u_{1}, v_{1}), (u_{2}, v_{2}), \ldots, (u_{k}, v_{k}) \}$. 

Two sets $E_{L}$ and $E_{F}$ of directed edges represent the covering relationship between DBWT-repetitions 
and F-intervals on DBWT $D_{\delta}$. 
Set $E_{L} \subseteq (U \times V)$ is a set of directed edges from $U$ to $V$, and 
$(u_i, v_j) \in E_{L}$ holds if and only if 
F-interval $[\LF_\delta(p_j), \LF_\delta(p_{j+1}-1)]$, represented as $v_j$, covers the starting position $p_i$ of DBWT-repetition $L_\delta[p_i..(p_{i+1}-1)]$, represented as $u_i$. 
Formally, $E_{L} = \{ (u_{i}, v_{j}) \mid 1 \leq i, j \leq k  \mbox{ s.t. } p_{i} \in [\LF_{\delta}(p_{j}), \LF_{\delta}(p_{j+1} - 1)] \}$.
Similarly, $E_{F} \subseteq (V \times U)$ is a set of directed edges from $V$ to $U$, and 
$(v_j, u_i) \in E_{F}$ holds if and only if 
DBWT-repetition $L_\delta[p_i..(p_{i+1}-1)]$, represented as $u_i$, covers the starting position $\LF_\delta(p_j)$ of F-interval $[\LF_\delta(p_j),\LF_\delta(p_{j+1}-1)]$,  represented as $v_j$. 
Formally, $E_{F} = \{ (v_{j}, u_{i}) \mid 1 \leq i, j \leq k  \mbox{ s.t. } \LF_{\delta}(p_{j}) \in [p_{i}, (p_{i+1} - 1)] \}$.

Function $B_{U} : U \rightarrow (\Sigma, \mathcal{N})$ is a label function for the set $U$ of nodes, and 
it maps each node $u_i \in U$ to a pair consisting of the character in $\Sigma$ and the length in $\mathcal{N}$ for the $i$-th DBWT-repetition represented by $u_i$. 
Namely, $B_{U}(u_{i}) = (L_{\delta}[p_{i}], p_{i+1} - p_{i})$. 
Similarly, 
$B_{V} : V \rightarrow (\Sigma, \mathcal{N})$ is a label function for the set $V$ of nodes, and 
it maps each node $v_i \in V$ to a pair consisting of the character in $\Sigma$ and the length in $\mathcal{N}$ for 
the repetition $F_{\delta}[\LF_{\delta}(p_{i})..\LF_{\delta}(p_{i+1} - 1)]$ of the same character on the F-interval represented by $v_i$. 
Namely, $B_{V}(v_{i}) = (F_{\delta}[\LF_{\delta}(p_{i})], \LF_{\delta}(p_{i+1} - 1) - \LF_{\delta}(p_{i}) + 1)$. 
For all $i \in \{ 1,2, \ldots, k \}$, $B_{U}(u_{i}) = B_{V}(v_{i})$ holds by the LF formula. 

Function $B_{L}: E_{L} \rightarrow \mathcal{N}$ is a label function for the set $E_{L}$ of directed edges, and it maps each edge $(u_i, v_j) \in E_{L}$ to an integer value representing the difference between the starting position $p_i$ of DBWT-repetition $L_{\delta}[p_{i}..(p_{i+1}-1)]$, represented as $u_i$, and 
the starting position $\LF_{\delta}(p_j)$ of F-interval $[\LF_{\delta}(p_{j}), \LF_{\delta}(p_{j+1}-1)]$, represented as $v_j$. 
Namely, $B_{L}(u_{i}, v_{j}) = p_{i} - \LF_{\delta}(p_{j})$.

Similarly, $B_{F}: E_{F} \rightarrow \mathcal{N}$ is a label function for the set $E_{F}$ of directed edges, and it maps each edge $(v_j, u_i) \in E_{F}$ to 
an integer value representing the difference between the starting position $\LF_\delta(p_j)$ of F-interval $[\LF_{\delta}(p_{j}), \LF_\delta(p_{j+1}-1)]$, represented as $v_j$, and the starting position $p_i$ of DBWT-repetition $L_{\delta}[p_{i}..(p_{i+1}-1)]$, 
represented as $u_i$. 
Namely, $B_{F}(v_{j}, u_{i}) = \LF_{\delta}(p_{j}) - p_{i}$. 

\oldsentence{
Examples of LF-interval graph can be found in Appendix~\ref{app:graph:examples}. 
}
\newsentence{
Figure~\ref{fig:graph} illustrates LF-interval graph $\graphds(D_{11})$ for DBWT $D_{11}$ in Figure~\ref{fig:dbwt}. 
For $U = \{ u_{1}, u_{2}, \ldots, u_{7} \}$ and $V = \{ v_{1}, v_{2}, \ldots, v_{7} \}$, 
each node $u_i \in U$ (respectively, $v_j \in V$) is enclosed by a rectangle on 
$L_{11}$ (respectively, $F_{11}$). 
We have $E_{LF} = \{ (u_{1}, v_{1})$, $(u_{2}, v_{2})$, $(u_{3}, v_{3})$, $(u_{4}, v_{4})$, $(u_{5}, v_{5})$, $(u_{6}, v_{6})$, $(u_{7}, v_{7})  \}$. 
We depict each undirected edge in set $E_{LF}$ by solid lines. 
Moreover, $E_{L} = \{ (u_{1}, v_{4})$, $(u_{2}, v_{1})$, $(u_{3}, v_{6})$, $(u_{4}, v_{7})$, $(u_{5}, v_{2})$, $(u_{6}, v_{5})$, $(u_{7}, v_{5})  \}$, 
and $E_{F} = \{ (v_{1}, u_{2})$, $(v_{2}, u_{5})$, $(v_{3}, u_{5})$, $(v_{4}, u_{1})$, $(v_{5}, u_{6})$, $(v_{6}, u_{2})$, $(v_{7}, u_{4}) \}$. 
Each directed edge in the two sets $E_{L}$ and $E_{F}$ is depicted by a dotted arrow. 
The four label functions $B_{U}$, $B_{V}$, $B_{L}$, and $B_{F}$ are listed under the LF-interval graph.
}

\subsection{Dynamic data structures for the LF-interval graph}\label{sec:dynamic}
Several dynamic data structures are used for efficiently updating LF-interval graph $\graphds(D_{\delta})$. 
Two doubly linked lists are used for supporting the insertions and deletions of nodes in $U$ and $V$.
The nodes in $U$ should be totally ordered with respect to the starting position of the DBWT-repetition, which is represented as a node in $U$. Namely, $u_1 < u_2 < \cdots < u_k$. 
The nodes in set $U$ are stored in a doubly linked list, where each node $u_i \in U$ has previous and next pointers 
connecting to the previous and next nodes, respectively, in the total order of nodes in $U$. 
Similarly, the nodes in $V$ should be totally ordered with respect to the starting position of the F-interval, which is represented as a node in $V$. 
Namely,  $v_{\pi_1} < v_{\pi_2} < \ldots < v_{\pi_k}$ holds for permutation $\pi_1, \pi_2, \ldots, \pi_k$ of sequence $1, 2, \ldots, k$ such that $\LF_{\delta}(p_{\pi_1}) < \LF_{\delta}(p_{\pi_2}) < \ldots < \LF_{\delta}(p_{\pi_k})$.
The nodes in $V$ are stored in another doubly linked list, where each node $v_i \in V$  has previous and next pointers connecting to 
the previous and next nodes in the increasing order of nodes in $V$, respectively. 
The space of two doubly linked lists storing nodes in $U$ and $V$ is $O(k \log n)$ bits of space. 

All the nodes corresponding to $\alpha$-heavy DBWT-repetitions in $U$ are stored in an array data structure in any order. 
Similarly, all the nodes corresponding to $\alpha$-heavy F-intervals in $V$ are stored in another array data structure in any order. 
The two arrays take $O(k \log n)$ bits of space. 
Each array stores nothing if $D_{\delta}$ is $\alpha$-balanced. 

An \emph{order maintenance data structure}~\cite{DBLP:conf/stoc/DietzS87} is used for comparing two nodes in $U$ with the total order of $U$, and the data structure supports the following three operations: 
(i) the order operation determines whether or not node $u_i \in U$ precedes node $u_j \in U$ in the total order of $U$; 
(ii) the insertion operation inserts node $u_i \in U$ right after node $u_j \in U$ in the total order of $U$; 
(iii) the deletion operation deletes node $u_i \in U$ from $U$.
The data structure supports these three operations in $O(1)$ time with $O(k\log n)$ bits of space, and 
it is used with a B-tree that stores the nodes in $V$, as explained below. 

A B-tree (a type of self-balancing search tree) is built on the set $V$ of nodes using the combination of the order
maintenance data structure, where 
each node $v_i$ in $V$ is totally ordered with respect to (i) the total order of the node $u_i$ in $U$ that is connected to $v_i$ by an edge in $E_{LF}$~(i.e., $(u_{i}, v_{i}) \in E_{LF}$) and (ii) the first character $L_{\delta}[p_{i}]$ 
of the DBWT-repetition is represented as $u_i$. 
For a node $v_{i} \in V$, 
the B-tree stores a pair $(u_i,L_{\delta}[p_{i}])$ as the key of the node $v_{i}$. 
Nodes in $V$ are totally ordered using the key, and $v_i \in V$ precedes $v_j \in V$ if and only if either of the following conditions holds: (i) $L_{\delta}[p_{i}] < L_{\delta}[p_{j}]$ or 
(ii) $L_{\delta}[p_{i}] = L_{\delta}[p_{j}]$ and $u_{i}$ precedes $u_{j}$ in the total order of $U$~(i.e., $i < j$).  
Condition (ii) is efficiently computed in $O(1)$ time by the order maintenance data structure of $U$. 
According to the following lemma, the order of keys in the B-tree is the same as that of the nodes stored in the doubly linked list of $V$~(i.e., the order of the nodes in the B-tree is $v_{\pi_1} < v_{\pi_2} < \ldots < v_{\pi_k}$). 

\begin{lemma}\label{lem:B_tree_order}
For two distinct nodes $v_{i}, v_{j} \in V$, 
the key of $v_{i}$ precedes that of $v_{j}$ in the B-tree of $V$ if and only if 
$v_{i}$ precedes $v_{j}$ in the doubly linked list of $V$~(i.e., $\LF_{\delta}(p_{i}) < \LF_{\delta}(p_{j})$).  
\end{lemma}
\begin{proof}
\oldsentence{See Appendix~\ref{app:B_tree_order}.}

\newsentence\bgroup
We use proof by contradiction to prove the lemma.
The keys of $v_{i}$ and $v_{j}$ are $(u_{i}, L_{\delta}[p_{i}])$ and $(u_{j}, L_{\delta}[p_{j}])$, respectively. 
Assume that Lemma~\ref{lem:B_tree_order} does not hold. 
Either of the following two statements then holds for $\LF_{\delta}(p_{i}) \geq \LF_{\delta}(p_{j})$: 
(i) $L_{\delta}[p_{i}] < L_{\delta}[p_{j}]$ or 
(ii) $L_{\delta}[p_{i}] = L_{\delta}[p_{j}]$ and $i < j$. 
However, $\LF_{\delta}(p_{i}) < \LF_{\delta}(p_{j})$ holds if $L_{\delta}[p_{i}] < L_{\delta}[p_{j}]$ 
because the characters in $F_{\delta}$ are sorted in lexicographical order. 
Hence, the first statement does not hold. 

Next, we have $p_{i} < p_{j}$ by $i < j$. 
$\LF_{\delta}(p_{i}) < \LF_{\delta}(p_{j})$ holds if $L_{\delta}[p_{i}] = L_{\delta}[p_{j}]$ and $p_{i} < p_{j}$ 
because $\LF_{\delta}(x) < \LF_{\delta}(y)$ holds by the LF formula for two integers $1 \leq x < y \leq \delta$ and $L_{\delta}[x] = L_{\delta}[y]$. 
Hence, the second statement does not hold. 
Because the assumption does not hold, Lemma~\ref{lem:B_tree_order} holds. 
\egroup

\end{proof}

The B-tree with the order maintenance data structure supports the three operations of search, insertion, and deletion for any node in $V$ in $O(\log k)$ time with $O(k \log n)$ bits of space. 

\subsection{Extension of BWT~(\texorpdfstring{\cite{DBLP:journals/algorithmica/PolicritiP18,DBLP:journals/jda/OhnoSTIS18})}{}}\label{lab:ExBWT}
The BWT of a suffix can be extended from the BWT of a shorter suffix \cite{DBLP:journals/algorithmica/PolicritiP18,DBLP:journals/jda/OhnoSTIS18}. 
In this section, we review the extension of BWT, which is used for updating the LF-interval graph. 
The BWT $L_{\delta+1}$ of a suffix $T_{\delta+1}$ of length $\delta+1$ can be computed from the BWT $L_{\delta}$ of the suffix $T_{\delta}$ of length $\delta$ using the following two steps: 
(i) special character $\$$ in $L_{\delta}$ is replaced with the first character $c$ of $T_{\delta+1}$~(i.e., $c = T[n-\delta]$);
(ii) special character $\$$ is inserted into $L_{\delta}$ at a position $\inspos$. 
Here, $\inspos$ is computed by the LF formula as follows: for the position $\reppos$ of special character $\$$ in $L_{\delta}$~(i.e., $L_{\delta}[\reppos] = \$$), $\inspos = \occ_{<}(L_{\delta}, c) + \rank(L_{\delta}, \reppos, c) + 1$. 
\oldsentence{
An example of the extension of BWT can be found in Appendix~\ref{app:graph:examples}. 
}

\newsentence{
In Figure~\ref{fig:bwt}, BWT $L_{12}$ of suffix $T_{12}$ can be extended from $L_{11}$ of suffix $T_{11}$. 
The first character $c$ of $T_{12}$ is $a$, and special character $\$$ is replaced with $a$ at $\reppos = 5$ on $L_{11}$. 
The insertion position $\inspos$ for $L_{11}$ is $3$ 
because $\occ_{<}(L_{11}, a) + \rank(L_{11}, \reppos, a) + 1 = 3$. 
}

\subsection{Foundation of updates of the LF-interval graph} \label{sec:update_graph}
\newsentence\bgroup
\begin{figure}[t]
 \begin{center}
		\includegraphics[scale=0.4]{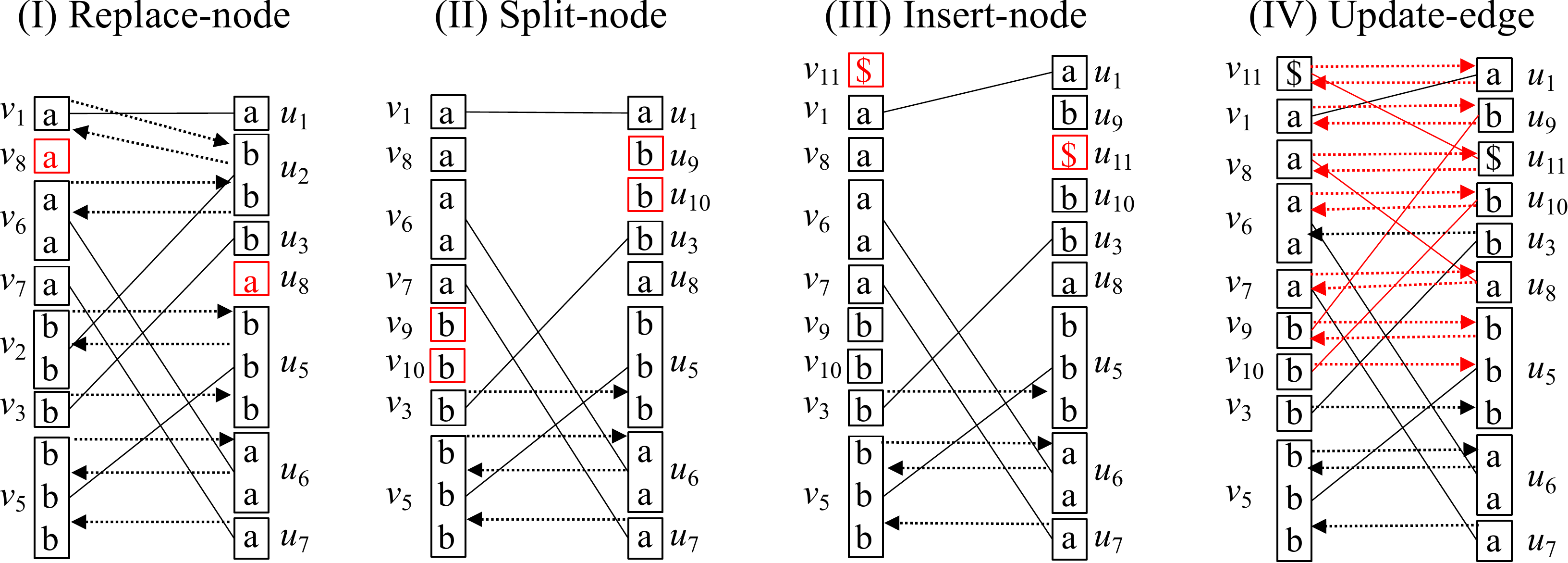}
	  \caption{
	  Each step in the update operation for the LF-interval graph in Figure~\ref{fig:graph}. 
	  New nodes and edges created in each step are colored in red. 
	  }
 \label{fig:insert}
 \end{center}
\end{figure}
\egroup

Given the first character $c$ in suffix $T_{\delta+1}$ of length $\delta+1$, 
an update operation of LF-interval graph $\graphds(D^{\alpha}_{\delta})$ for an $\alpha$-balanced DBWT $D^{\alpha}_{\delta} = L_{\delta}[p_{1}..(p_{2}-1)], L_{\delta}[p_{2}..(p_{3}-1)], \ldots, L_{\delta}[p_{k}..(p_{k+1}-1)]$ of $T_{\delta}$ updates $\graphds(D^{\alpha}_{\delta})$ to $\graphds(D^{2\alpha+1}_{\delta+1})$ for a $(2\alpha+1)$-balanced DBWT $D^{2\alpha+1}_{\delta+1}$ of $T_{\delta+1}$. 
The update operation updates the given LF-interval graph according to the extension of BWT. 
This operation consists of four main steps: (I) \emph{replace node}, (II) \emph{split node}, (III) \emph{insert node}, and (IV) \emph{update edge}. 
Note that the update operation presented in this section is a foundation for the ones presented in the following two subsections, where several modifications are made to the foundation for faster operation. 
\oldsentence{
Examples for the four main steps can be found in Appendix~\ref{app:basic_update:examples}. 
}

\subparagraph{(I) Replace node.} 
This step replaces the node $u_{i} \in U$ labeled $(\$, 1)$ with a new one $u_{i^\prime}$ labeled $(c, 1)$, and 
it updates $V$ according to the replacement of the node in $U$. 
Node $u_{i}$ can be found in $O(1)$ time by keeping track of it on $U$.
The doubly linked list of $U$ is updated according to the replacement. 
The node $v_{i} \in V$ connected to $u_{i}$ by edge $(u_{i},v_{i}) \in E_{LF}$ is removed from $V$, 
and a new node $v_{i^\prime}$ labeled $(c, 1)$ is inserted into $V$ at the position next to 
the most backward node $v_{g}$ of the nodes whose keys are smaller than key $(u_{i}, c)$. 
Node $v_{g}$ can be found in $O(\log k)$ time using the B-tree of $V$. 
This step takes $O(\log k)$ time in total. 

\newsentence\bgroup
Figure~\ref{fig:insert}-(I) shows an example of the replace-node step for the LF-interval graph $\graphds(D_{11})$ in Figure~\ref{fig:graph}. 
Node $u_4 \in U$ labeled $(\$, 1)$ on $\graphds(D_{11})$ is replaced with node $u_8$ labeled $(a, 1)$. 
Node $v_4 \in V$, which is connected to $u_4$ by edge $(v_4, u_4) \in E_{LF}$, is removed from $V$, and edge $(v_4, u_4)$ is removed from $E_{LF}$. 
A new node $v_8$ with label $(a, 1)$ is inserted into $V$. This node is inserted into the doubly linked list of $V$ 
at the position next to $v_{1}$~(i.e., $v_{g} = v_{1}$). 
\egroup

\subparagraph{(II) Split node.}
The insertion-node step (as the next step) inserts a new node representing special character $\$$ into $U$. 
However, before the insertion-node step, the split-node step splits a node $u_j \in U$ into two new nodes at an appropriate position on the doubly linked list of $U$. 
This step is executed for inserting the new node representing special character $\$$ into $U$ at a appropriate position in the insert-node step. 
Following the extension of BWT in Section~\ref{lab:ExBWT},
node $u_j \in U$ has label $(L_{\delta}[p_j], p_{j+1}-p_{j})$ for the two starting positions $p_j$ and $p_{j+1}$ satisfying $p_{j} < \inspos < p_{j+1}$ for the insertion position $\inspos$ of special character $\$$. 
Such a node $u_{j}$ exists if and only if 
(i) the BWT $L_{\delta+1}$ of $T_{\delta+1}$ does not have special character $\$$ as the last character~(i.e., $\inspos \neq \delta + 1$) 
and (ii) $p_{i} \neq \inspos$ for all $i \in \{ 1, 2, \ldots, k \}$.
This is because $p_{1} < p_{2} < \ldots < p_{k+1}$ and $p_{k+1} = \delta + 1$ hold. 

If node $u_{j}$ does not exist in $U$, 
this step does not split nodes. 
Otherwise, 
$u_{j}$ is replaced with two new nodes $u_{j^{\prime}}$ and $u_{j^{\prime}+1}$ in the doubly linked list of $U$, 
where $u_{j^{\prime}}$ is previous to $u_{j^{\prime}+1}$. 
The new nodes $u_{j^\prime}$ and $u_{j^\prime+1}$ are labeled as $(L_{\delta}[p_{j}], \inspos - p_{j})$ and 
$(L_{\delta}[p_{j}], p_{j+1} - \inspos)$ using insertion position $\inspos$, respectively. 

Although we do not know position $\inspos$ in the split-node step, we can find node $u_{j}$. 
This is because 
(i) set $V$ contains node $v_{\mathsf{gnext}}$ representing the F-interval starting at position $\inspos$ unless $\inspos = \delta + 1$, 
and (ii) $v_{\mathsf{gnext}}$ is next to $v_{g}$ in the doubly linked list of $V$ for the node $v_{g}$ searched for in the replace-node step. 
The following lemma ensures that 
we can find $u_{j}$ and compute the labels of the new nodes in $O(1)$ time. 


\begin{lemma}\label{lem:split_node_formula}
The following two statements hold after executing the replace-node step: 
(i) we can check whether $u_{j}$ exists or not in $O(1)$ time; 
(ii) we can find $u_{j}$ and compute the labels of two nodes $u_{j^\prime}$ and $u_{j^\prime+1}$ in $O(1)$ time. 
\end{lemma}
\oldsentence{
\begin{proof}
See Appendix~\ref{app:split_node_formula}. 
\end{proof}
}

\newsentence\bgroup
We prove Lemma~\ref{lem:split_node_formula}. 
The following lemmas can be used for finding node $u_j$ in the doubly linked list of $U$ in $O(1)$ time. 

\begin{lemma}\label{lem:v_g+1_pos}
For (i) node $v_{g} \in V$ searched for in the replace-node step 
and (ii) node $v_{\mathsf{gnext}} \in V$ next to $v_{g}$ in the doubly linked list before executing the replace-node step, 
$v_{\mathsf{gnext}}$ represents the F-interval starting at position $\inspos$ on $F_{\delta}$ 
if $v_{\mathsf{gnext}}$ exists; 
otherwise, $\inspos = \delta + 1$ and $u_{j} \not \in V$. 
\end{lemma}
\begin{proof}
By the extension of BWT, 
$F_{\delta+1}$ can be computed from $F_{\delta}$ by inserting the input character $c$ into $F_{\delta}$ at position $\inspos$. 
The replace-node step inserted node $v_{i^{\prime}}$ into the doubly linked list of $V$ at the position next to $v_{g} \in V$. 
Because the node $v_{i^\prime}$ represents the input character $c$, 
the F-interval of $v_{g}$ ends at position $(\inspos - 1)$ on $F_{\delta}$. 
Node $v_{\mathsf{gnext}}$ was next to $v_{g}$ in the list before the replace-node step was executed.
Thus, the F-interval of $v_{\mathsf{gnext}}$ starts at position $\inspos$ on $F_{\delta}$.
If set $V$ does not contain $v_{\mathsf{gnext}}$, 
then $\inspos$ must be $\delta + 1$ and $u_{j} \not \in V$. 
\end{proof}

\begin{lemma}\label{lem:connect}
Assume that $v_{\mathsf{gnext}} \in V$. 
For node $u_{x} \in U$ connected to $v_{\mathsf{gnext}}$ by the directed edge $(v_{\mathsf{gnext}}, u_x) \in E_{F}$, 
if the label $B_{F}(v_{\mathsf{gnext}}, u_x)$ of the directed edge is larger than $0$, then 
$u_{x} = u_{j}$. 
Otherwise~(i.e., $B_{F}(v_{\mathsf{gnext}}, u_{x}) = 0$), $u_{j}$ is not contained in $U$. 
\end{lemma}
\begin{proof}
Node $u_{x}$ is labeled $(L_{\delta}[p_x], p_{x+1}-p_{x})$, and 
the starting position $\LF_{\delta}(p_{\mathsf{gnext}})$ of the F-interval, represented as $v_{\mathsf{gnext}} \in V$, is covered by the DBWT-repetition, represented as $u_x$, resulting in $p_{x} \leq \LF_{\delta}(p_{\mathsf{gnext}}) < p_{x+1}$. 
For $B_{F}(v_{\mathsf{gnext}}, u_x) = 0$, 
we have $\LF_{\delta}(p_{\mathsf{gnext}}) = p_{x}$ 
because $B_{F}(v_{\mathsf{gnext}}, u_x) = \LF_{\delta}(p_{\mathsf{gnext}}) - p_{x}$. 
Because $\LF_{\delta}(p_{\mathsf{gnext}}) = \inspos$ by Lemma~\ref{lem:v_g+1_pos}, 
we obtain $p_{x} = \inspos < p_{x+1}$. 
This fact indicates that set $U$ does not contain $u_{j}$. 
Similarly, for $B_{F}(v_{\mathsf{gnext}}, u_x) > 0$, 
we obtain $p_{x} < \inspos < p_{x+1}$, which indicates that $u_{x} = u_{j}$. 
\end{proof}

Node $u_{j}$ is replaced with two new nodes $u_{j^{\prime}}$ and $u_{j^{\prime}+1}$ in the doubly linked list of $U$, 
where $u_{j^{\prime}}$ is previous to $u_{j^{\prime}+1}$. 
The next lemma guarantees that those two nodes can be labeled using edge label $B_F(v_{\mathsf{gnext}}, v_j)$.

\begin{lemma}\label{lem:edge_label}
Two split nodes $u_{j^\prime}$ and $u_{j^\prime+1}$ can be labeled as $(L_{\delta}[p_{j}], B_{F}(v_{\mathsf{gnext}}, u_{j}))$ and 
$(L_{\delta}[p_{j}], (p_{j+1} - p_{j}) - B_{F}(v_{\mathsf{gnext}}, u_{j}))$, respectively. 
\end{lemma}
\begin{proof}
The new nodes $u_{j^\prime}$ and $u_{j^\prime+1}$ are labeled as $(L_{\delta}[p_{j}], \inspos - p_{j})$ and 
$(L_{\delta}[p_{j}], p_{j+1} - \inspos)$ using insertion position $\inspos$, respectively. 
We obtain $\inspos - p_{j} = B_{F}(v_{\mathsf{gnext}}, u_{j})$ 
because $B_{F}(v_{\mathsf{gnext}}, u_{j}) = \LF_{\delta}(p_{\mathsf{gnext}}) - p_{j}$ and $\LF_{\delta}(p_{\mathsf{gnext}}) = \inspos$. 
Similarly, $p_{j+1} - \inspos = (p_{j+1} - p_{j}) - B_{F}(v_{\mathsf{gnext}}, u_{j})$. 
\end{proof}
We can compute two labels $B_{U}(u_{j})(p_{j+1} - p_{j})$ and $B_{F}(v_{\mathsf{gnext}}, u_{j})$ in $O(1)$ time using node $v_{\mathsf{gnext}}$. 
Because we can compute $v_{\mathsf{gnext}}$ in $O(1)$ time using the result of the replace-node step, 
we obtain Lemma~\ref{lem:split_node_formula}. 

\egroup

Next, set $V$ is updated according to the replacement of nodes in $U$, i.e., 
for undirected edge $(u_{j}, v_{j}) \in E_{LF}$, 
node $v_{j}$ is replaced with two new nodes $v_{j^\prime}$ and $v_{j^\prime+1}$ in the doubly linked list of $V$, 
where $v_{j^\prime}$ is previous to $v_{j^\prime+1}$. 
Nodes $v_{j^\prime}$ and $v_{j^\prime+1}$ have the same labels of $u_{j^\prime}$ and $u_{j^\prime+1}$, respectively. 
Therefore, this step takes $O(1)$ time. 

\newsentence\bgroup
Figure~\ref{fig:insert}-(II) illustrates an example of the split-node step. 
In this example, $\inspos = 3$, $v_{\mathsf{gnext}} = v_{6}$, $v_{j} = v_{2}$, $v_{j^{\prime}} = v_{9}$, and $v_{j^{\prime}+1} = v_{10}$ hold. 
In Figure~\ref{fig:graph}, 
the directed edge starting at node $v_{6}$ is labeled as integer $1$ by function $B_{F}(v_{6}, u_{2})$, 
and the directed edge points to node $u_{2}$ with label $(b, 2)$. 
Hence, node $u_{2}$ is replaced with two nodes $u_{9}$ and $u_{10}$. 
Two nodes $u_{9}$ and $u_{10}$ are labeled with pairs $(b, 1)$ and $(b, 1)$, respectively. 
Node $v_{2}$ is connected to $u_{2}$ by edge $(v_2, u_2) \in E_{LF}$, 
and $v_{2}$ is replaced with two nodes $v_{9}$ and $v_{10}$. 
Here, $v_{9}$ and $v_{10}$ are labeled with pairs $(b, 1)$ and $(b, 1)$, respectively. 
\egroup

\subparagraph{(III) Insert node.} 
This step inserts a new node $u_{x^{\prime}}$ labeled $(\$, 1)$ into $U$, 
and it updates $V$ according to the insertion of $U$. 
Analogous to the extension of BWT described in Section~\ref{lab:ExBWT}, 
the position for inserting the new node in the doubly linked list of $U$ is determined according to the following three cases: 
(i) Node $u_{j} \in U$ was found and it was split into two nodes $u_{j^{\prime}}$ and $u_{j^{\prime}+1}$ in the split-node step. 
In this case, node $u_{x^{\prime}}$ is inserted at the position next to $u_{j^{\prime}} \in U$ on the doubly linked list of $U$.
(ii) Node $u_{j}$ was not found, 
and new node $v_{i^{\prime}}$ is inserted at the position next to the last element on the doubly linked list of $V$ in the replace-node step. 
In this case, $u_{x^{\prime}}$ is inserted at the position next to the last element on the doubly linked list of $U$. 
(iii) Node $u_{j}$ was not found, 
and $v_{i^{\prime}}$ is inserted at the position previous to a node $v_{\mathsf{gnext}} \in V$ on the doubly linked list of $V$. 
In this case, $v_{\mathsf{gnext}}$ is connected to a node $u_{x} \in U$ by a directed edge in $E_{F}$, 
and $u_{x^{\prime}}$ is inserted into the doubly linked list of $U$ at the position previous to $u_{x}$. 

Next, this step creates a new node $v_{x^{\prime}}$ labeled $(\$, 1)$, 
and it is inserted into the doubly linked list of $V$. 
The new node is inserted at the top of the list, 
because the new label includes special character $\$$. 
This step takes $O(1)$ time. 

\newsentence\bgroup
Figure~\ref{fig:insert}-(III) illustrates an example of the insert-node step. 
Because the split-node step replaced node $u_{2}$ with two nodes $u_{9}$ and $u_{10}$, 
the insert-node step inserts node $u_{11}$ labeled $(\$, 1)$ at the position next to  
$u_{9}$ on the doubly linked list of $U$. 
On the other hand, the step inserts node $v_{11}$ labeled $(\$, 1)$ in the doubly linked list of $V$ at the position previous to node $v_{1}$.
\egroup

\subparagraph{(IV) Update edge.} 
\oldsentence{
This step updates the set $E_{LF}$ of undirected edges and two sets $E_{L}$ and $E_{F}$ of directed edges according to the two sets $U$ and $V$, which were updated in the previous steps. 
This step consists of two phases: 
(i) for the nodes of $u_i$, $v_i$, $u_{j}$, and $v_{j}$ removed in the replace-node and split-node steps, the edges 
connected to these nodes are removed from $E_{LF}$, $E_L$, and $E_F$; 
(ii) new edges connecting new nodes~(i.e., $u_{j^{\prime}}$, $u_{j^{\prime}+1}$, $v_{i^{\prime}}$, and $v_{j^{\prime}+1}$) 
are added to $E_{LF}$, $E_L$, and $E_F$. 
The labels of new directed edges are computed at the second phase. 
The number of removed edges and new edges can be bounded by $O(\alpha)$ 
because (i) every node in the LF-interval graph for an $O(\alpha)$-balanced DBWT is connected to $O(\alpha)$ edges, 
(ii) the DBWT represented by the given LF-interval graph $\graphds(D^{\alpha}_{\delta})$ is $\alpha$-balanced, 
and (iii) the LF-interval graph $\graphds(D^{2\alpha+1}_{\delta+1})$ outputted by this update operation represents a $(2\alpha+1)$-balanced DBWT. 
Because the number of updated edges is small, 
this step can be performed in $O(\alpha)$ time. 
See Appendix~\ref{app:time_update_edge_step} for the details of the update-edge step. 
Formally, we obtain the following lemma. 
\begin{lemma}\label{lem:time_update_edge_step}
The update-edge step takes $O(\alpha)$ time.
\end{lemma}
\begin{proof}
See Appendix~\ref{app:time_update_edge_step}. 
\end{proof}
}

\newsentence\bgroup
This step updates the set $E_{LF}$ of undirected edges and two sets $E_{L}$ and $E_{F}$ of directed edges according to the two sets $U$ and $V$, which were updated in the previous steps. 
The update-edge step consists of two phases: 
(i) for the nodes of $u_i$, $v_i$, $u_{j}$, and $v_{j}$ removed in the replace-node and split-node steps, the edges 
connected to these nodes are removed from $E_{LF}$, $E_L$, and $E_F$; 
(ii) new edges connecting new nodes~(i.e., $u_{j^{\prime}}$, $u_{j^{\prime}+1}$, $v_{i^{\prime}}$, and $v_{j^{\prime}+1}$) 
are appropriately added to $E_{LF}$, $E_L$, and $E_F$. 

For the nodes of $u_i$, $v_i$, $u_{j}$, and $v_{j}$ removed in the replace-node and split-node steps, the edges 
connected to these nodes are removed from $E_{LF}$, $E_L$, and $E_F$. 
Because the given LF-interval graph $\graphds(D^{\alpha}_{\delta})$ represents an $\alpha$-balanced DBWT, 
every node in the LF-interval graph for an $\alpha$-balanced DBWT is connected to $O(\alpha)$ edges. 
Hence, the number of removed edges is $O(\alpha)$, 
and the removal of all the edges can be performed in $O(\alpha)$ time. 

New undirected edges connecting new nodes are added to $E_{LF}$. 
If the split-node step created new nodes~(i.e., $u_{j^{\prime}}$, $u_{j^{\prime}+1}$, $v_{i^{\prime}}$, and $v_{j^{\prime}+1}$), four new edges $(u_{i^{\prime}}, v_{i^{\prime}})$, $(u_{j^{\prime}}, v_{j^{\prime}})$, $(u_{i^{\prime}+1}, v_{i^{\prime}+1})$, and $(u_{x^{\prime}}, v_{x^{\prime}})$ are added to $E_{LF}$; 
otherwise, two new edges, $(u_{i^{\prime}}, v_{i^{\prime}})$ and $(u_{x^{\prime}}, v_{x^{\prime}})$, are added to $E_{LF}$. 

New directed edges connecting new nodes are created, and each new directed edge is added to $E_{L}$ or $E_{F}$ appropriately. 
The following lemmas ensure the addition of all new edges finishes in $O(\alpha)$ time. 
\begin{lemma}\label{lem:addedge}
The number of new directed edges is $O(\alpha)$.
\end{lemma}
\begin{proof}
Each of the new directed edges starts from (i) a new node or (ii) the node connected to the tail of a removed directed edge. 
The number of new nodes is at most eight, and the number of removed edges is $O(\alpha)$. 
Hence, we obtain Lemma~\ref{lem:addedge}.
\end{proof}
\begin{lemma}\label{lem:computingedge}
Given a node, the new directed edge connected to it by the tail and its edge label on the LF-interval graph can be accessed in $O(1)$ time 
if the node is not $v_{j^{\prime}+1}$; 
otherwise, computing the edge and its label takes $O(\alpha)$ time. 
\end{lemma}
\begin{proof}
See Appendix~\ref{app:computingedge}.
\end{proof}
The two phases of the update-edge step take $O(\alpha)$ time in total. 
\egroup

\newsentence\bgroup
In Figure~\ref{fig:insert}-(IV), 
four edges $(u_{8}, v_{8})$, $(u_{9}, v_{9})$, $(u_{10}, v_{10})$, and $(u_{11}, v_{11})$ connecting new nodes are added to $E_{LF}$.
Six directed edges $(u_{1}, v_{11})$, $(u_{9}, v_{1})$, $(u_{11}, v_{8})$, $(u_{10}, v_{6})$, $(u_{8}, v_{7})$, and $(u_{5}, v_{9})$ are added to $E_{L}$. 
Similarly, seven directed edges $(v_{11}, u_{1})$, $(v_{1}, u_{9})$, $(v_{8}, u_{11})$, $(v_{6}, u_{10})$, $(v_{7}, u_{8})$, $(v_{9}, u_{5})$, and $(v_{10}, u_{5})$ are 
added to $E_{F}$.
\egroup

\subparagraph{Update of the data structures.}
Similar to the update-edge step, 
the four data structures in the LF-interval graph~(i.e., the order maintenance data structure, the B-tree of set $V$, and 
two arrays that store nodes representing $\alpha$-heavy DBWT-repetitions and F-intervals) 
are updated according to the removed nodes and new nodes. 

\oldsentence{
See Appendix~\ref{app:time_update_data_structures} for the details of the algorithm updating the four data structures. 
The following lemma concerning the update time of the four data structures. 
\begin{lemma}\label{lem:time_update_data_structures}
Updating the four data structures takes $O(\alpha + \log k)$ time.
\end{lemma}
\begin{proof}
See Appendix~\ref{app:time_update_data_structures}. 
\end{proof}
}

\newsentence\bgroup
\subparagraph{Update of the order maintenance data structure.}
For the order maintenance data structure, 
at most two deleted nodes $u_{i}$ and $u_{j}$ from $U$ are also removed from the order maintenance data structure in $O(1)$ time. 
In addition, at most four new nodes $u_{i^{\prime}}$, $u_{j^{\prime}}$, $u_{j^{\prime}+1}$, and $u_{x^{\prime}}$ in $U$ are inserted into the same order maintenance data structure in $O(1)$ time.  

\subparagraph{Update of the B-tree of $V$.}
The two nodes $v_{i}$ and $v_{j}$ removed from $V$ are also removed from the B-tree of $V$. 
In addition, four new nodes $v_{i^{\prime}}$, $v_{j^{\prime}}$, $v_{j^{\prime}+1}$, and $v_{x^{\prime}}$ in $V$ are inserted into the B-tree of $V$ with their keys 
$(u_{i^{\prime}}, c)$, $(u_{j^{\prime}}$, $L_{\delta}[p_{j}])$, $(u_{j^{\prime}+1}$, $L_{\delta}[p_{j}])$, and $(u_{x^{\prime}}, \$)$, respectively. 
Updating the B-tree of $V$ takes $O(\log k)$ time. 

\subparagraph{Update of two arrays for $\alpha$-heavy DBWT-repetitions and F-intervals.}
New nodes representing $\alpha$-heavy DBWT-repetitions or F-intervals are added to the two arrays. 
The following two lemmas ensure that the number of such nodes is at most one and it can be found in $O(\alpha)$ time. 
\begin{lemma}\label{lem:update_edge_count}
Let $u_{h} \in U$~(respectively, $v_{h^{\prime}} \in V$) be a node connected to the head of the directed edge starting at the new node $v_{j^{\prime} + 1} \in V$~(respectively, $u_{j^{\prime}+1} \in U$) created in the split-node step. 
(i) The only DBWT-repetition represented as $u_h$ can be $\alpha$-heavy (i.e., 
the DBWT-repetitions represented as all the other nodes except for $u_h$ in $U$ are not $\alpha$-heavy), 
and (ii) the only F-interval represented as $v_{h^{\prime}}$ can be $\alpha$-heavy. 
\end{lemma}
\begin{proof}
See Appendix~\ref{app:update_edge_count}. 
\end{proof}

\begin{lemma}\label{lem:verify_heavy}
Whether or not any node $u \in U$~(respectively, $v \in V$) represents an $\alpha$-heavy DBWT-repetition~(respectively, an $\alpha$-heavy F-interval) 
on an LF-interval graph can be verified in $O(\alpha)$ time. 
\end{lemma}
\begin{proof}
Every $\alpha$-heavy DBWT-repetition or F-interval has at least $\alpha$ directed edges. 
Finding at most $\alpha$ directed edges to a given node takes $O(\alpha)$ time on an LF-interval graph. 
Thus, Lemma~\ref{lem:verify_heavy} holds. 
\end{proof}

Hence, updating the four data structures stored in the LF-interval graph takes $O(\alpha + \log k)$ time. 
\egroup

The update operation takes $O(\alpha + \log k)$ time in total. 
The following lemma concerning the theoretical results on this update operation holds. 
\begin{theorem}\label{thm:insert_balance}
The following two statements hold: 
(i) the update operation takes $O(\alpha + \log k)$ time;
(ii) the update operation takes as input the LF-interval graph for an $\alpha$-balanced DBWT $D^{\alpha}_{\delta}$ of BWT $L_{\delta}$, 
and it outputs the LF-interval graph for a $(2\alpha + 1)$-balanced DBWT $D^{2\alpha+1}_{\delta+1}$ 
of BWT $L_{\delta+1}$ with at most two $\alpha$-heavy DBWT-repetitions and at most two $\alpha$-heavy F-intervals. 
\end{theorem}
\begin{proof}
See Appendix~\ref{app:insert_balance}. 
\end{proof}
From Theorem~\ref{thm:insert_balance}, the update operation outputs the LF-interval graph 
for a $(2\alpha + 1)$-balanced DBWT $D^{2\alpha+1}_{\delta+1}$ of BWT $L_{\delta+1}$, which is not $\alpha$-balanced.
The output DBWT $D^{2\alpha+1}_{\delta+1}$ is balanced 
into an $\alpha$-balanced DBWT $D^{\alpha}_{\delta+1}$ 
by the balancing operation presented in Section~\ref{sec:balancing}. 
The update of the LF-interval graph presented in this section takes $O(\alpha + \log{k})$ time, 
which results in an $O(\alpha n + n\log{k})$-time construction of the RLBWT from an input string of length $n$. 
The following two sections (Section~\ref{sec:slow_update} and 
Section~\ref{sec:fast_update}) present two modified updates of the LF-interval graph in $O(\alpha + \log{k})$-time and $O(\alpha)$-time, respectively, 
in order to achieve the $O(\alpha n + r\log{r})$-time construction of an RLBWT with $O(r\log{n})$ bits of working space.

\subsection{\texorpdfstring{$O(\alpha + \log{k})$}{O(alpha + log k)}-time update of LF-interval graph }\label{sec:slow_update}
\newsentence\bgroup
\begin{figure}[t]
 \begin{center}
		\includegraphics[scale=0.6]{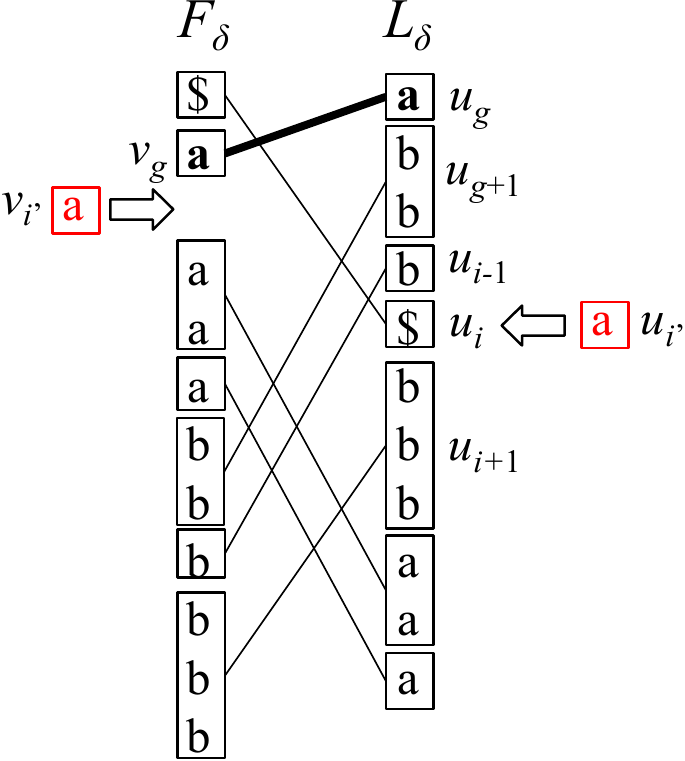}
	  \caption{
	  Replace-node step in update operation $\update(\graphds(D^{\alpha}_{\delta}), c)$.  
	  }
 \label{fig:replace}
 \end{center}
\end{figure}
\egroup

This section presents the update operation $\update(\graphds(D^{\alpha}_{\delta}), c)$ taking 
an LF-interval graph $\graphds(D^{\alpha}_{\delta})$ and the first character $c$ of suffix $T_{\delta+1}$ as input and running in $O(\alpha + \log{k})$ time 
by modifying the foundation of the update operation presented in Section~\ref{sec:update_graph}. 
For the node $u_{i-1}$ previous to the node $u_{i}$ representing special character $\$$ in the doubly linked list of $U$ and the node $u_{i+1}$ next to $u_{i}$, 
update operation $\update(\graphds(D^{\alpha}_{\delta}), c)$ is applied
if neither $u_{i-1}$ nor $u_{i+1}$ has labels including character $c$. 
 
We first present a modified update of the B-tree of $V$ in $O(\log{k})$ time. This update replaces the original update of the B-tree. 
The following lemma holds with respect to nodes searched for using the B-tree of $V$ in the replace-node step of this update operation.
\begin{lemma}\label{lem:mod}
For the node $v_{g} \in V$ searched for using the B-tree of $V$ in the replace-node step of the update operation $\update(\graphds(D^{\alpha}_{\delta}), c)$, 
$v_{g}$ satisfies any one of the following three properties: 
(i) for undirected edge $(u_{g}, v_{g}) \in E_{LF}$ and node $u_{g+1} \in U$ next to $u_{g}$ in the doubly linked list of $U$, 
the two consecutive nodes $u_{g}$ and $u_{g+1}$ have labels including different characters, and the label of $u_{g+1}$ does not include special character $\$$; 
(ii) for the node $u_{g+2} \in U$ next to $u_{g+1}$ in the doubly linked list of $U$, 
the two nodes $u_{g}$ and $u_{g+2}$ have labels including different characters, and the label of $u_{g+1}$ includes special character $\$$; 
(iii) the label of $u_{g}$ includes special character $\$$.
\end{lemma}
\begin{proof}
\oldsentence{
See Appendix~\ref{app:mod}.
}
\newsentence\bgroup
Recall that $v_{g} \in V$ is the most backward node of the nodes whose keys are smaller than 
key $(u_{i}, c)$ in the doubly linked list of $V$ 
for the node $u_{i} \in U$ representing special character $\$$ 
and input character $c$. 
Let $(u_{g}, c^{\prime})$ be the key of the node $v_{g}$.
If the labels of the first $(i-1)$ nodes $u_{1}, u_{2}, \ldots, u_{i-1}$ do not include character $c$, 
then the LF formula ensures that 
$c^{\prime} < c$ and $v_{g}$ is the most backward node of the nodes with keys that include characters smaller than $c$ in the doubly linked list of $V$.  
In this case, by the LF formula, neither $u_{g+1}$ nor $u_{g+2}$ have labels with character $c^{\prime}$, 
and thus $v_{g}$ satisfies any one of the three conditions of Lemma~\ref{lem:mod}.

Otherwise~(i.e., at least one node in the first $(i-1)$ nodes $u_{1}, u_{2}, \ldots, u_{i-1}$ has a label including character $c$), 
$g$ is the largest integer in $\{ 1, 2, \ldots, i-1 \}$ such that $u_{g}$ has a label with character $c$~(i.e., $c^{\prime} = c$). 
$g < i-1$ because 
$u_{i-1}$ does not have a label with character $c$ for the update operation $\update(\graphds(D^{\alpha}_{\delta}), c)$. 
Because $g < i-1$, 
neither $u_{g+1}$ nor $u_{g+2}$ have labels with character $c$, 
and hence $v_{g}$ satisfies any one of the three conditions of Lemma~\ref{lem:mod}.
\egroup
\end{proof}

Thus, the B-tree of $V$ stores only the nodes in $V$ satisfying one of the three conditions of Lemma~\ref{lem:mod}, 
because Lemma~\ref{lem:mod} ensures only such nodes are searched for in the replace-node step of this update operation. 
\oldsentence{
An illustration for the replace-node step can be found in Appendix~\ref{app:slow_update:examples}. 
}

\newsentence\bgroup
Figure~\ref{fig:replace} illustrates the replace-node of the update operation $\update(\graphds(D^{\alpha}_{\delta}), c)$. 
In this case, neither $u_{i-1}$ nor $u_{i+1}$ has labels including character $c$. 
The two nodes $u_{g}$ and $u_{g+1}$ have labels including different characters, 
and the latter node does not have a label including special character $\$$. 
Hence, node $v_{g}$ satisfies the first condition of Lemma~\ref{lem:mod}. 
\egroup

The target nodes inserted into the B-tree of $V$~(respectively, the target nodes deleted from the B-tree of $V$) 
are limited to four~(respectively, three) according to the following lemma. 
\begin{lemma}\label{lem:insert_and_delete_nodes_by_slow_update}
Nodes $u_{i} \in U$ and $v_{i} \in V$ are the nodes removed from $U$ and $V$ by the replace-node steps, respectively. 
Node $u_{i-1} \in U$ is the node previous to node $u_{i}$ in the doubly linked list of $U$, 
and $v_{i-1} \in V$ is the node connected to $u_{i-1}$ by undirected edge $(u_{i-1}, v_{i-1}) \in E_{LF}$. 
Node $v_{j} \in V$ is the node removed from $V$ by the split-node step, 
and $v_{j^\prime}, v_{j^\prime+1} \in V$ are the nodes newly created by the same step. 
Similarly, $v_{i^\prime} \in V$ and $v_{x^\prime} \in V$ are the nodes created by the replace-node and insert-node steps, respectively. 
Then, the following two statements hold for the update operation $\update(\graphds(D^{\alpha}_{\delta}), c)$: 
(i) targets inserted into the B-tree of $V$ can be limited to only four nodes $v_{i-1}, v_{i^\prime}, v_{x^\prime}$, and $v_{j^\prime+1}$; 
(ii) the targets deleted from the B-tree of $V$ can be limited to only three nodes, $v_{i-1}$, $v_{i}$, and $v_{j}$.

\end{lemma}
\oldsentence{
\begin{proof}
See Appendix~\ref{app:insert_and_delete_nodes_by_slow_update}.
\end{proof}
}
\newsentence\bgroup
\begin{proof}
\textbf{Proof of Lemma~\ref{lem:insert_and_delete_nodes_by_slow_update}-(i).}
Lemma~\ref{lem:insert_and_delete_nodes_by_slow_update}-(i) holds if the following two statements hold: 
(i) new node $v_{j^{\prime}}$ satisfies none of the three conditions in Lemma~\ref{lem:mod};
(ii) let $v_{x} \in \{ v_{1}, v_{2}, \ldots, v_{k} \}$ be a node such that 
(a) $v_{x} \not \in \{ v_{i}, v_{j}, v_{i-1} \}$ and (b) $v_{x}$ is not contained in the B-tree of $V$~(i.e., $v_{x}$ satisfies none of the three conditions in Lemma~\ref{lem:mod}), 
then $v_{x}$ still satisfies none of the three conditions in Lemma~\ref{lem:mod} after executing the update operation.

We show that the first statement holds. 
The new node $v_{j^{\prime}}$ is connected to $u_{j^{\prime}}$ 
by edge $(u_{j^{\prime}}, v_{j^{\prime}}) \in E_{LF}$.  
Node $u_{j^{\prime}}$ is previous to the new node $u_{x^{\prime}}$ representing special character $\$$, 
and $u_{x^{\prime}}$ is previous to new node $u_{j^{\prime}+1}$. 
The labels of both new nodes $u_{j^{\prime}}$ and $u_{j^{\prime}+1}$ include the same character, 
and hence $v_{j^{\prime}}$ satisfies none of the three conditions in Lemma~\ref{lem:mod}.

Next, we show that the second statement holds.
Node $v_{x}$ is connected to node $u_{x} \in U$ by edge $(u_{x}, v_{x}) \in E_{LF}$, 
and $u_{x+1} \in U$ is the node next to $u_{x}$ in the doubly linked list of $U$ before the update operation has been executed. 
Note that $u_{x+1}$ does not exist in $U$ only if $u_{x}$ is the last node in the doubly linked list of $U$. 
Because $v_{x}$ satisfies none of the three conditions in Lemma~\ref{lem:mod} and $v_{x} \neq v_{i-1}$, 
either of the following two conditions holds: 
(a) the labels of nodes $u_{x}$ and $u_{x+1}$ include the same character or 
(b) $u_{x}$ is the last node in the doubly linked list of $U$. 

We show that $u_{x}$ still satisfies none of the three conditions in Lemma~\ref{lem:mod} for Condition-(a) after executing the update operation. 
Let $u_{\mu} \in U$ be the node next to $u_{x}$ in the doubly linked list of $U$ after the update operation has been executed. 
Similarly, let $u_{\mu+1} \in U$ be the node next to $u_{\mu}$ in the doubly linked list of $U$ after the update operation has been executed. 
If $u_{\mu}$ represents special character $\$$~(i.e., $u_{\mu} = u_{x^{\prime}}$), 
then the labels of nodes $u_{x}$ and $u_{\mu+1}$ include the same character, 
which indicates that $v_{x}$ does not satisfy the three conditions of Lemma~\ref{lem:mod}. 
Otherwise, the labels of nodes $u_{x}$ and $u_{\mu}$ include the same character, 
which indicates that $v_{x}$ does not satisfy the three conditions of Lemma~\ref{lem:mod}. 

Next, we show that $u_{x}$ still satisfies none of the three conditions in Lemma~\ref{lem:mod} for Condition-(b) after executing the update operation. 
After the update operation has been executed, 
$u_{x}$ is still the last node in the doubly linked list of $U$ or 
$u_{\mu}$ represents special character $\$$. 
This fact indicates that $v_{x}$ does not satisfy the three conditions of Lemma~\ref{lem:mod}.

\textbf{Proof of Lemma~\ref{lem:insert_and_delete_nodes_by_slow_update}-(ii).}
Lemma~\ref{lem:insert_and_delete_nodes_by_slow_update}-(ii) holds if the following statement holds: 
(A) 
If a node $v_{x} \in \{ v_{1}, v_{2}, \ldots, v_{k} \} \setminus \{ v_{i}, v_{j}, v_{i-1} \}$ 
is contained in the B-tree of $V$, 
then $v_{x}$ satisfies one of the three conditions of Lemma~\ref{lem:mod} after 
update operation $\update(\graphds(D^{\alpha}_{\delta}), c)$ has been executed.

We show that statement A holds. 
Node $v_{x}$ is connected to node $u_{x} \in U$ by edge $(u_{x}, v_{x}) \in E_{LF}$, 
and $u_{x+1} \in U$ is the node next to $u_{x}$ in the doubly linked list of $U$ before the update operation has been executed. 
Let $c_{x}$ and $c_{x+1}$ be the characters included in the labels of $u_{x}$ and $u_{x+1}$, respectively. 
We have $c_{x} \neq \$$ by $v_{x} \neq v_{i}$. 
Similarly, $c_{x+1} \neq \$$ because $v_{x} \neq v_{i-1}$. 
Moreover, $c_{x} \neq c_{x+1}$ because $v_{x}$ is contained in the B-tree of $V$. 
Let $u_{\mu}$ be the node next to $u_{x}$ in the doubly linked list of $U$ after the update operation has been executed. Similarly, let $u_{\mu+1}$ be the node next to $u_{\mu}$ in the doubly linked list of $U$ after the update operation has been executed. 
If $u_{\mu}$ represents special character $\$$, 
then the label of $u_{\mu+1}$ includes $c_{x+1}$, 
and hence $v_{x}$ satisfies the second condition of Lemma~\ref{lem:mod}. 
Otherwise, the label of $u_{\mu}$ includes $c_{x+1}$, 
and hence $v_{x}$ satisfies the first condition of Lemma~\ref{lem:mod}. 
Therefore, statement A holds. 
\end{proof}
\egroup

Thus, all the nodes inserted into the B-tree of $V$ can be found by searching for only four nodes $v_{i-1}, v_{i^\prime}, v_{x^\prime}$, and $v_{j^\prime+1}$ and 
checking whether or not each one, $v_{i-1}, v_{i^\prime}, v_{x^\prime}$, and $v_{j^\prime+1}$, satisfies one of the three conditions presented in Lemma~\ref{lem:mod}. 
Similarly, all the nodes deleted from the B-tree of $V$ can be found by searching for only three nodes $v_{i-1}$, $v_{i}$, and $v_{j}$ 
and checking whether or not each of these nodes satisfies one of the three conditions. 
This modified update of the B-tree takes $O(\log k)$ time in total.

The algorithm of update operation $\update(\graphds(D^{\alpha}_{\delta}), c)$ 
is the same as that of the original update operation presented in Section~\ref{sec:update_graph} except for the algorithm updating the B-tree of $V$. 
Hence, update operation $\update(\graphds(D^{\alpha}_{\delta}), c)$ takes $O(\alpha + \log k)$ time in total. 
The following lemma concerning the theoretical results on this update operation holds. 
\begin{lemma}\label{lem:slow_update_result}
Assume that the B-tree of $V$ in LF-interval graph $\graphds(D^{\alpha}_{\delta})$ contains only nodes satisfying one of the three conditions of Lemma~\ref{lem:mod}: 
(i) 
update operation $\update(\graphds(D^{\alpha}_{\delta}), c)$ runs in $O(\alpha + \log k)$ time;  
(ii) the update operation outputs the LF-interval graph for 
a $(2\alpha + 1)$-balanced DBWT $D^{2\alpha+1}_{\delta+1}$ 
of BWT $L_{\delta+1}$ with at most two $\alpha$-heavy DBWT-repetitions and at most two $\alpha$-heavy F-intervals; 
(iii) 
the B-tree of $V$ in the outputted LF-interval graph contains only nodes satisfying one of the three conditions of Lemma~\ref{lem:mod}.  
\end{lemma}
In the next subsection, the second modified update operation of LF-interval graph achieves $O(\alpha)$ time using the B-tree of $V$ containing only nodes satisfying one of the three conditions of Lemma~\ref{lem:mod}.

\subsection{\texorpdfstring{$O(\alpha)$}{O(alpha)}-time update of LF-interval graph}\label{sec:fast_update}
This section presents fast update operation $\fastUpdate(\graphds(D^{\alpha}_{\delta}), c)$, which takes 
an LF-interval graph $\graphds(D^{\alpha}_{\delta})$ and the first character $c$ of suffix $T_{\delta+1}$ as input and runs in $O(\alpha)$ time. 
This time is achieved by modifying the foundation of the update operation presented in Section~\ref{sec:update_graph}, 
and the B-tree of $V$ needs to contain only nodes of satisfying one of the three conditions of Lemma~\ref{lem:mod} in the input and output LF-interval graphs, similar to the update operation in Section~\ref{sec:slow_update}~(Lemma~\ref{lem:slow_update_result}). 

For node $u_{i-1} \in U$ previous to node $u_{i} \in U$ that represents special character $\$$ in the doubly linked list of $U$ and node $u_{i+1} \in U$ next to $u_{i}$, 
the fast update operation is applied if 
either or both $u_{i-1}$ and $u_{i+1}$ have a label including character $c$; 
the update operation in Section~\ref{sec:slow_update} is applied otherwise. 
The large computational demand of the update operation on LF-interval graphs presented 
in Section~\ref{sec:update_graph} derives from the access and update of the B-tree of $V$ in $O(\log{k})$ time, resulting in an $O(\alpha + \log{k})$ time update of LF-interval graphs.
We present two improvements to the foundation of the update operation: (i) deletion and insertion operations of the B-tree of $V$ in $O(1)$ time and (ii) the replace-node step in $O(1)$ time without using the B-tree of $V$. 
The details of the fast operation are presented in Appendix~\ref{app:details_of_fast_update_operation}.

\subparagraph{Deletion and insertion operations of B-tree in constant time.}
\newsentence\bgroup
\begin{figure}[t]
 \begin{center}
		\includegraphics[scale=0.6]{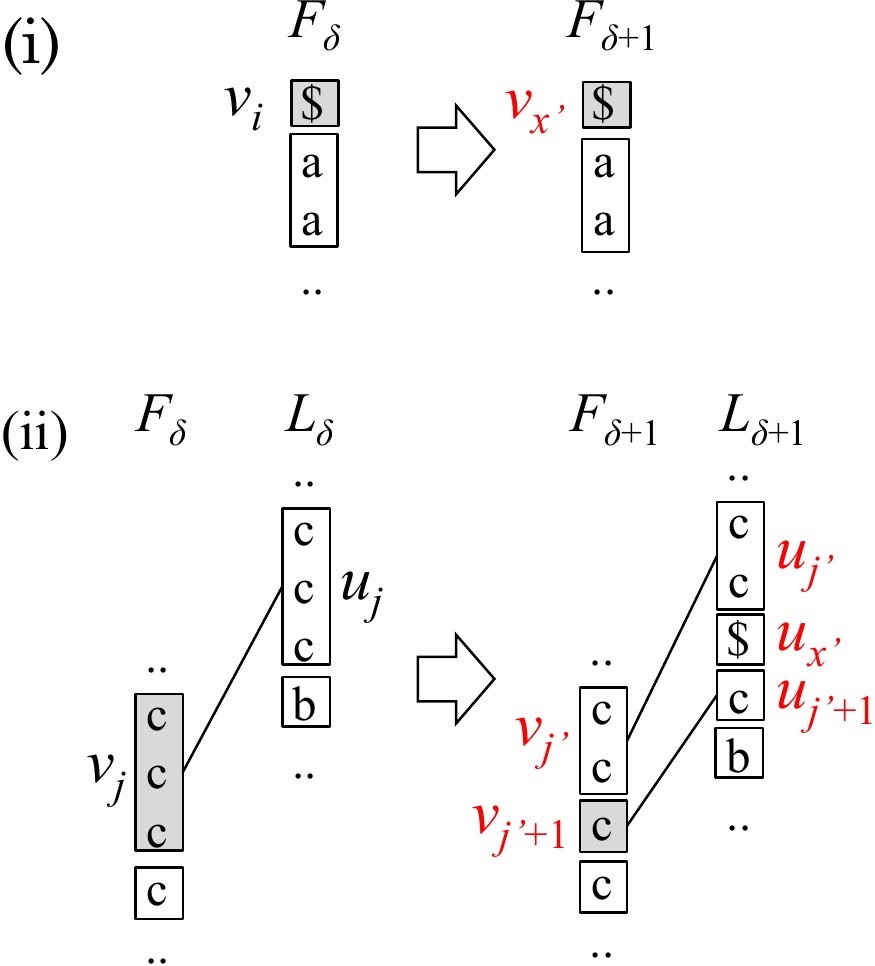}
	  \caption{
      (i) Replacement of $v_{i}$ with $v_{x^{\prime}}$ in the doubly linked list of $V$.
      (ii) Replacement of $v_{j}$ with $v_{j^{\prime}}$ and $v_{j^{\prime}+1}$ in the doubly linked list of $V$ and the replacement of $u_j$ 
      with  $u_{j}$, $u_{j^{\prime}+1}$ and $u_{x^{\prime}}$ in the doubly linked list of $U$.
      In both figures, gray nodes satisfy one of the three conditions of Lemma~\ref{lem:mod}.
	  }
 \label{fig:btree_replace}
 \end{center}
\end{figure}

\begin{figure}[t]
 \begin{center}
 \includegraphics[scale=0.5]{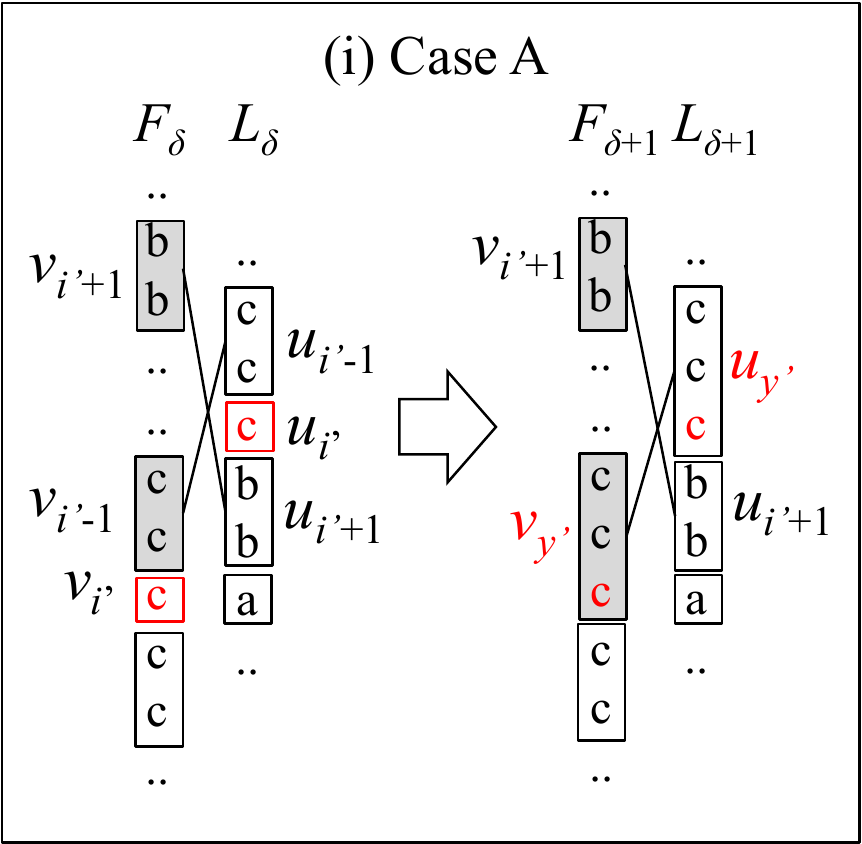}
 \includegraphics[scale=0.5]{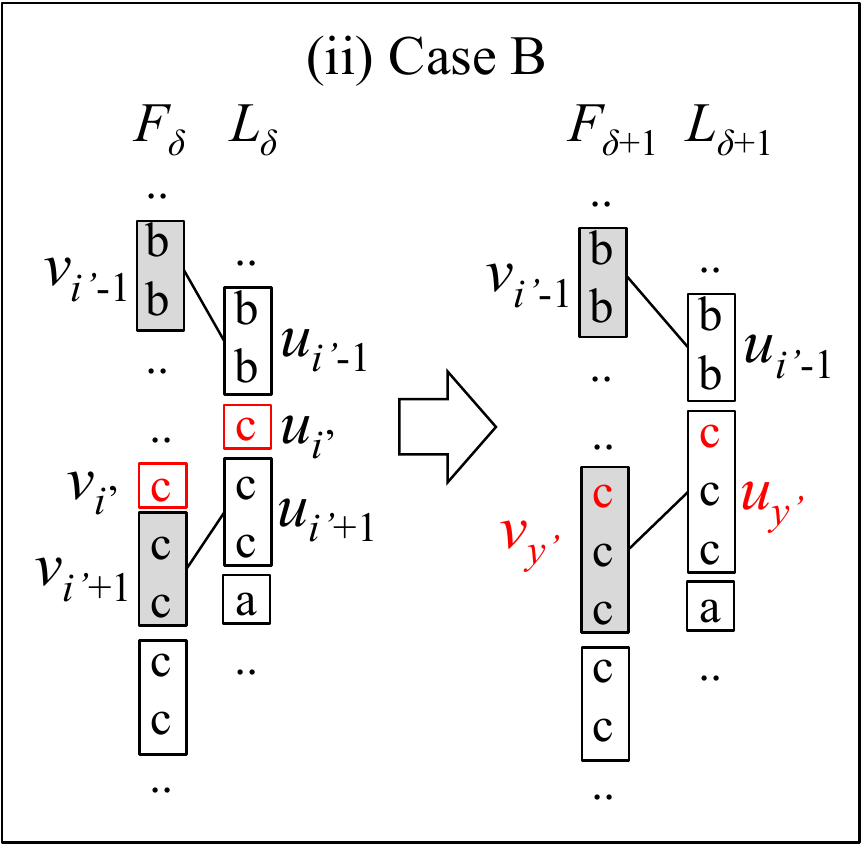}
 \includegraphics[scale=0.5]{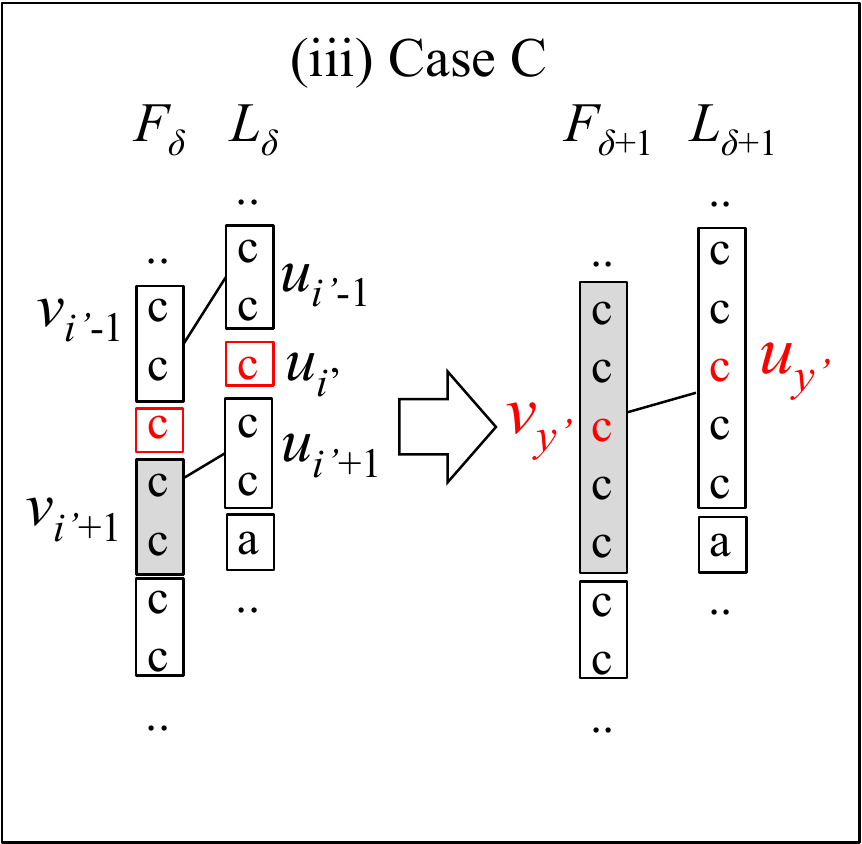}

 \end{center}
	  \caption{
	  Three cases A, B, and C for preprocessing deletions/insertions in the B-tree. 
	  Gray nodes are stored in the B-tree.
	  }
 \label{fig:btree_replace_with_merging}
\end{figure}
\egroup

Generally, inserting/deleting a key into/from the B-tree of $V$ takes $O(\log{k})$ time. 
We present $O(1)$-time deletion and insertion operations of a specific node in the B-tree of $V$ without the need for heavyweight operations to maintain the balance of the B-tree. 

Recall that (i) $u_{i^\prime} \in U$ and $v_{i^\prime} \in V$ are the nodes created by the replace-node step, 
(ii) $v_{j} \in V$ is the node removed from $V$ by the split-node step, 
and (iii) $u_{x^{\prime}} \in U$ and $v_{x^{\prime}} \in V$ are the nodes created by the insert-node step.  
Let $u_{i^\prime-1} \in U$~(respectively, $u_{i^\prime+1} \in U$) be the node previous to node $u_{i^\prime}$~(respectively, the node next to node $u_{i^\prime}$) 
in the doubly linked list of $U$ after the insert-node step has been executed. 
Then, there exist two nodes $v_{i^\prime-1}$ and $v_{i^\prime+1} \in V$ 
such that $(u_{i^\prime-1}, v_{i^\prime-1}), (u_{i^\prime+1}, v_{i^\prime+1}) \in E_{LF}$. 

In the fast update operation, 
the target nodes deleted from the B-tree are limited to at most four nodes $v_{i}, v_{j}, v_{i^{\prime}-1}$, and $v_{i^{\prime}+1}$, 
which is similar to Lemma~\ref{lem:insert_and_delete_nodes_by_slow_update}-(ii). 
In the doubly linked list of $V$, 
node $v_{i}$ is replaced with the new node $v_{x^{\prime}}$ created in the insert-node step of 
the update operation~(see Figure~\ref{fig:btree_replace}-(i)). 
Because the two nodes $v_{i}$ and $v_{x^{\prime}}$ represent special character $\$$, 
both satisfy the third condition of Lemma~\ref{lem:mod}. 
Thus, both nodes $v_{i}$ and $v_{x^{\prime}}$ are placed on the root of the B-tree. 
Hence, $v_{i}$ can be deleted and $v_{x^{\prime}}$ can be inserted into the B-tree in $O(1)$ time without balancing the B-tree. 

Next, in the doubly linked list of $V$, node $v_{j}$ is replaced with the two nodes $v_{j^{\prime}}$ and $v_{j^{\prime}+1} \in V$ 
that were created in the split-node step of the update operation~(see Figure~\ref{fig:btree_replace}-(ii)). 
Node $v_{j^\prime}$ satisfies none of the three conditions of Lemma~\ref{lem:mod}. 
Thus, $v_{j}$ is deleted from the B-tree of $V$, and $v_{j^{\prime}+1}$ (that is, not $v_{j^\prime}$) is only inserted into the B-tree of $V$ if $v_{j}$ is contained in the B-tree of $V$ because of the next lemma. 
\begin{lemma}\label{lem:b-tree_update_lemma1}
The following two statements hold: 
(i) $v_{j^{\prime}}$ satisfies none of the three conditions of Lemma~\ref{lem:mod}; 
(ii) for the node $u_{j} (\neq u_{i-1})$ connected to $v_{j}$ by undirected edge $(u_{j}, v_{j}) \in E_{LF}$, $v_{j^{\prime}+1}$ satisfies any one of the three conditions of Lemma~\ref{lem:mod} if and only if $v_{j}$ satisfies any one of the three conditions. 
\end{lemma}
\begin{proof}
\oldsentence{
See Appendix~\ref{app:b-tree_update_lemma1}. 
}

\newsentence\bgroup
(i) Nodes $v_{j^{\prime}}$, $v_{x^{\prime}}$, and $v_{j^{\prime}+1}$ are connected to 
three nodes $u_{j^{\prime}}, u_{x^{\prime}}$, and $u_{j^{\prime}+1} \in U$, respectively,
by edges in $E_{LF}$. 
The three nodes $u_{j^{\prime}}, u_{x^{\prime}}$, and $u_{j^{\prime}+1}$ are consecutive in the doubly linked list of $U$~(i.e., $u_{j^{\prime}}$ is previous to $u_{x^{\prime}}$, and $u_{x^{\prime}}$ is previous to $u_{j^{\prime}+1}$). 
The label of $u_{x^{\prime}}$ includes special character $\$$, 
and the labels of $v_{j^{\prime}}$ and $v_{j^{\prime}+1}$ include the same character $c^{\prime}$. 
Hence, $v_{j^{\prime}}$ satisfies none of the three conditions of Lemma~\ref{lem:mod}. 

(ii) 
Let $c_{j}$~(respectively, $c_{j^{\prime}+1}$) be the character included in the label of $u_{j}$~(respectively, $u_{j^{\prime}+1}$). 
Similarly, let $c_{j+1}$~(respectively, $c_{j^{\prime}+2}$) be the character included in the label of the node next to $u_{j}$~(respectively, 
the node next to $u_{j^{\prime}+1}$) in the doubly linked list of $U$. 
Character $c_{j+1} \neq \$$ by the assumption $u_{j} \neq u_{i-1}$. 
Moreover, $c_{j+1} = c_{j^{\prime}+2}$ because the node next to $u_{j}$ is equal to the node next to $u_{j^{\prime}+1}$ unless $u_{j} = u_{i-1}$. 
Node $v_{j}$ satisfies the first condition of Lemma~\ref{lem:mod} if and only if $v_{j^{\prime}+1}$ satisfies the first condition 
by $c_{j} = c_{j^{\prime}+1}$, $c_{j+1} = c_{j^{\prime}+2}$, and $c_{j}, c_{j+1}, c_{j^{\prime}+1}, c_{j^{\prime} + 2} \neq \$$. 
Nodes $v_{j}$ and $v_{j^{\prime}+1}$ do not satisfy the second condition of Lemma~\ref{lem:mod} by $c_{j+1}, c_{j^{\prime}+2} \neq \$$. 
Nodes $v_{j}$ and $v_{j^{\prime}+1}$ do not satisfy the third condition of Lemma~\ref{lem:mod} by $c_{j}, c_{j^{\prime}+1} \neq \$$. 
Hence, Lemma~\ref{lem:b-tree_update_lemma1}-(ii) holds. 
\egroup
\end{proof}

Because $v_{j}$ is replaced with $v_{j^{\prime} + 1}$ in the doubly linked list of $V$, 
the element representing $v_{j}$ can be replaced with the element representing $v_{j^{\prime}+1}$ in the B-tree from Lemma~\ref{lem:B_tree_order}. Hence, $v_{j}$ is deleted from the B-tree of $V$ and $v_{j^{\prime}+1}$ is inserted into the B-tree of $V$ in $O(1)$ time without balancing the B-tree. 
\oldsentence{
Illustrations of the replacement of two nodes $v_{i}$ and $v_{j}$ in the doubly linked list of $V$ can be found in Appendix~\ref{app:fast_update:examples}. 
}


The above procedure inserts $v_{j^{\prime}+1}$ into the B-tree of $V$ even if the node satisfies none of the three conditions of Lemma~\ref{lem:mod}.
This is because Lemma~\ref{lem:b-tree_update_lemma1}-(ii) assumes that $u_{j} \neq u_{i-1}$ holds, 
but $u_{j} \neq u_{i-1}$ is not always true.
If the assumption holds, $v_{j^{\prime}+1}$ is appropriately inserted into the B-tree of $V$. 
Otherwise (i.e., $v_{j^{\prime}+1} = v_{i^{\prime}-1}$ holds),  node $v_{j^{\prime}+1}$ is appropriately deleted from the B-tree of $V$ in $O(1)$ time, 
which is explained next.

Two nodes $v_{i^{\prime}-1}$ and $v_{i^{\prime} + 1}$ are appropriately deleted from the B-tree of $V$, 
and new nodes are inserted into the B-tree of $V$ in the update operation of the LF-interval graph. 
For updating the B-tree of $V$ in $O(1)$ time, 
we merge at most three nodes $v_{i^{\prime}}$, $v_{i^{\prime}-1}$, and $v_{i^{\prime}+1}$ into a new node in a prepossessing step. 
The merging of nodes and update of the B-tree is performed according to the following three cases.

\subparagraph{Case A: $v_{i^{\prime}-1}$ has a label including character $c$ and $v_{i^{\prime}+1}$ does not have a label including character $c$~(Figure~\ref{fig:btree_replace_with_merging}-(i)).} 

First, the two consecutive nodes $u_{i^{\prime}-1}$ with label $(c, \ell)$ and $u_{i^{\prime}}$ with label $(c, 1)$ are merged into a new one $u_{y^{\prime}}$ with label $(c, \ell + 1)$ in the doubly linked list of $U$. 
Then, set $V$ is updated according to the merge of the two nodes in $U$, i.e., the two consecutive nodes $v_{i^{\prime}-1}$ and $v_{i^{\prime}}$ are merged into a new one $v_{y^{\prime}}$ with label $(c, \ell + 1)$ in the doubly linked list of $V$. 

Node $v_{i^{\prime}+1}$ is kept in the B-tree of $V$~(if the node is contained in the B-tree). 
Node $v_{i^{\prime}-1}$ is deleted from the B-tree of $V$ and new node $v_{y^{\prime}}$ is inserted 
only if $v_{i^{\prime}-1}$ is contained in the B-tree of $V$.
Because $v_{i^{\prime}-1}$ is replaced by $v_{y^{\prime}}$ in the doubly linked list of $V$, 
$v_{i^{\prime}-1}$ can be deleted from the B-tree of $V$ and $v_{y^{\prime}}$ can be inserted in $O(1)$ time from Lemma~\ref{lem:B_tree_order}. 

\subparagraph{Case B: $v_{i^{\prime}-1}$ does not have a label including character $c$, 
and $v_{i^{\prime}+1}$ has a label including character $c$~(Figure~\ref{fig:btree_replace_with_merging}-(ii)).}
First, the two consecutive nodes $u_{i^{\prime}}$ and $u_{i^{\prime}+1}$ with label $(c, \ell^\prime)$ are merged into a new node $u_{y^{\prime}}$ with label $(c, \ell^{\prime} + 1)$ in the doubly linked list of $U$. 
Then, set $V$ is updated according to the merge of the two nodes in $U$.

In this case, 
$v_{i^{\prime}-1}$ is contained in the B-tree of $V$, 
and the node is kept. 
Node $v_{i^{\prime}+1}$ is deleted from the B-tree of $V$ and new node $v_{y^{\prime}}$ is inserted 
only if $v_{i^{\prime}+1}$ is contained in the B-tree of $V$.
Because $v_{i^{\prime}+1}$ is replaced with $v_{y^{\prime}}$ in the doubly linked list of $V$, 
$v_{i^{\prime}+1}$ is deleted from the B-tree of $V$ and $v_{y^{\prime}}$ is inserted in $O(1)$ time from Lemma~\ref{lem:B_tree_order}.

\subparagraph{Case C: both $v_{i^{\prime}-1}$ and $v_{i^{\prime}+1}$ have labels including the same character $c$~(Figure~\ref{fig:btree_replace_with_merging}-(iii)).} 
First, the three consecutive nodes $u_{i^{\prime}-1}$, $u_{i^{\prime}}$ and $u_{i^{\prime}+1}$ are merged into  
a new node $u_{y^{\prime}}$ with label $(c, \ell + \ell^{\prime} + 1)$ including the same character $c$ in the doubly linked list of $U$. 
Then, set $V$ is updated according to the merge of the two nodes in $U$. 

In this case, $v_{i^{\prime}-1}$ is not contained in the B-tree of $V$. 
Node $v_{i^{\prime}+1}$ is deleted from the B-tree and the new node $v_{y^{\prime}}$ is inserted if $v_{i^{\prime}+1}$ is contained in the tree; otherwise, $v_{y^{\prime}}$ is not inserted into the tree.
Because $v_{i^{\prime}+1}$ is replaced by $v_{y^{\prime}}$ in the doubly linked list of $V$, 
$v_{i^{\prime}+1}$ is deleted and $v_{y^{\prime}}$ is inserted into the B-tree of $V$ in $O(1)$ time from Lemma~\ref{lem:B_tree_order}. 


One of the three cases always holds~(see Appendix~\ref{app:details_of_fast_update_operation}).
The preprocessing of the B-tree of $V$~(i.e., the merging of nodes) does not affect Lemma~\ref{lem:b-tree_update_lemma1}. 
Hence, updating the B-tree takes $O(1)$ time.


\subparagraph{Replace-node step in constant time.}
The replace-node step is performed in $O(1)$ time by finding the position for inserting a new node $v_{i^\prime}$ 
with label $(c, 1)$ into the doubly linked list of $V$ without accessing the B-tree of set $V$. 
This is made possible 
if either or both $u_{i-1}$ and $u_{i+1}$ have labels including the same character $c$.  
The following lemma holds. 
\begin{lemma}\label{lem:fast_insertion}
If either or both $u_{i-1}$ and $u_{i+1}$ have labels including the same character $c$,
the position for inserting the new node $v_{i^{\prime}}$ into the doubly linked list of $V$ can be found in $O(1)$ time. 
\end{lemma}
\oldsentence{
\begin{proof}
See Appendix~\ref{app:fast_insertion}. 
\end{proof}
}
\newsentence\bgroup
\begin{proof}
For the node $v_{i-1} \in V$ connected to $u_{i-1}$ by edge $(u_{i-1}, v_{i-1}) \in E_{LF}$ and 
node $v_{i+1} \in V$ connected to $u_{i+1}$ by edge $(u_{i+1}, v_{i+1}) \in E_{LF}$, 
$v_{i^{\prime}}$ is inserted at the position next to $v_{i-1}$ by the LF formula 
if the label of $u_{i-1}$ includes character $c$; 
otherwise, 
the label of $u_{i+1}$ includes character $c$, 
and $v_{i^{\prime}}$ is inserted into the doubly linked list of $V$ 
at the position previous to $v_{i+1}$ by the LF formula. 
The two nodes $v_{i-1}$ and $v_{i+1}$ can be found in $O(1)$ time using the two nodes $u_{i-1}, u_{i+1} \in U$. 
\end{proof}
\egroup

The following lemma concerning the conclusion of the fast update operation holds. 
\begin{lemma}\label{lem:fast_update_result}
Assume that the B-tree of $V$ in LF-interval graph $\graphds(D^{\alpha}_{\delta})$ contains only the nodes that satisfy one of the three conditions of Lemma~\ref{lem:mod}: 
(i) fast update operation $\fastUpdate(\graphds(D^{\alpha}_{\delta}), c)$ runs in $O(\alpha)$ time; 
(ii) the fast update operation outputs the LF-interval graph for a $(2\alpha + 1)$-balanced DBWT $D^{2\alpha+1}_{\delta+1}$ 
of BWT $L_{\delta+1}$ with at most two $\alpha$-heavy DBWT-repetitions and at most two $\alpha$-heavy F-intervals; 
(iii) the B-tree of $V$ in the outputted LF-interval graph contains only nodes satisfying one of the three conditions of Lemma~\ref{lem:mod}.
\end{lemma}
\begin{proof}
See Appendix~\ref{app:fast_update_result}. 
\end{proof}

Both update operation $\update(\graphds(D^{\alpha}_{\delta}), c)$ and 
fast update operation $\fastUpdate(\graphds(D^{\alpha}_{\delta}), c)$ output the LF-interval graph 
for a $(2\alpha + 1)$-balanced DBWT $D^{2\alpha+1}_{\delta+1}$ of BWT $L_{\delta+1}$ with at most two $\alpha$-heavy DBWT-repetitions and at most two $\alpha$-heavy F-intervals.
The outputs are balanced by the balancing operation presented in the next subsection 
such that the LF-interval graph represents an $\alpha$-balanced DBWT $D^{\alpha}_{\delta+1}$ of BWT $L_{\delta+1}$.

\subsection{Balancing operation of the LF-interval graph}\label{sec:balancing}
Balancing operation $\balance(\graphds(D_{\delta}))$ takes the LF-interval graph $\graphds(D_{\delta})$ for a DBWT $D_{\delta}$ of BWT $L_{\delta}$ as input, 
and it outputs the LF-interval graph $\graphds(D^\alpha_{\delta})$ for an $\alpha$-balanced DBWT $D^\alpha_{\delta}$ of the same BWT. 
The basic idea behind the balancing operation is to iteratively remove each of nodes representing 
$\alpha$-heavy DBWT repetitions and $\alpha$-heavy F-intervals from the given LF-interval graph 
by splitting the chosen node into two nodes. 
The balancing operation repeats this process until it obtains the LF-interval graph for an $\alpha$-balanced DBWT. 

\oldsentence{
We explain the algorithm of the balancing operation. 
At each iteration, the balancing operation processes the LF-interval graph for an $O(\alpha)$-balanced DBWT with $O(\alpha)$ $\alpha$-heavy DBWT-repetitions and $O(\alpha)$ $\alpha$-heavy F-intervals. 
Let $(u_{i}, v_{i}) \in E_{LF}$ be an undirected edge such that 
node $u_{i} \in U$ represents an $\alpha$-heavy DBWT-repetition, 
or node $v_{i} \in V$ represents an $\alpha$-heavy F-interval. 
At the iteration, node $u_{i}$ is split into two nodes, 
and $v_{i}$ is split into two nodes according to the splitting of $u_{i}$. 
The LF-interval graph $\graphds(D^{O(\alpha)}_{\delta})$ is updated according to the splitting of $u_{i}$ and $v_{i}$. 
One iteration of the balancing operation takes $O(\alpha)$ time. 
See Appendix~\ref{app:details_of_balancing_opeartion} for the details of the balancing operation. 
}

\newsentence\bgroup
\begin{figure}[t]
 \begin{center}
		\includegraphics[scale=0.7]{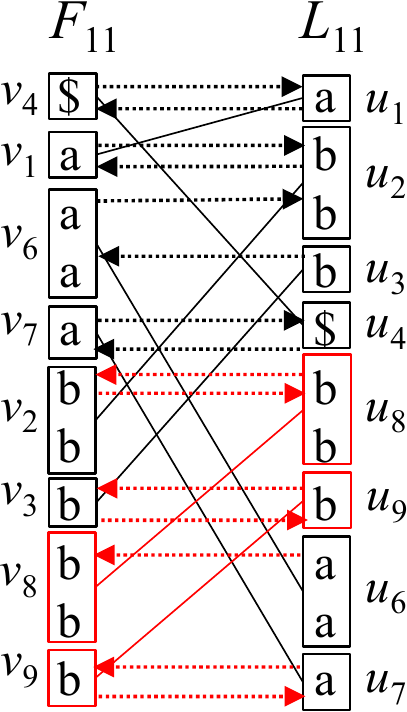}
	  \caption{
      LF-interval graph $\graphds(D_{11})$ updated by splitting node $u_{5}$ into two nodes $u_{8}$ and $u_{9}$. 
	  Here, $\graphds(D_{11})$ is the LF-interval graph in Figure~\ref{fig:graph}, 
	  and the label of $u_{8}$ is $(b, 2)$. 
	  Red rectangles and lines are new nodes and edges created by splitting node $u_{5}$. 
	  }
 \label{fig:split}
 \end{center}
\end{figure}

\begin{figure}[t]
     \begin{center}
		\includegraphics[scale=0.4]{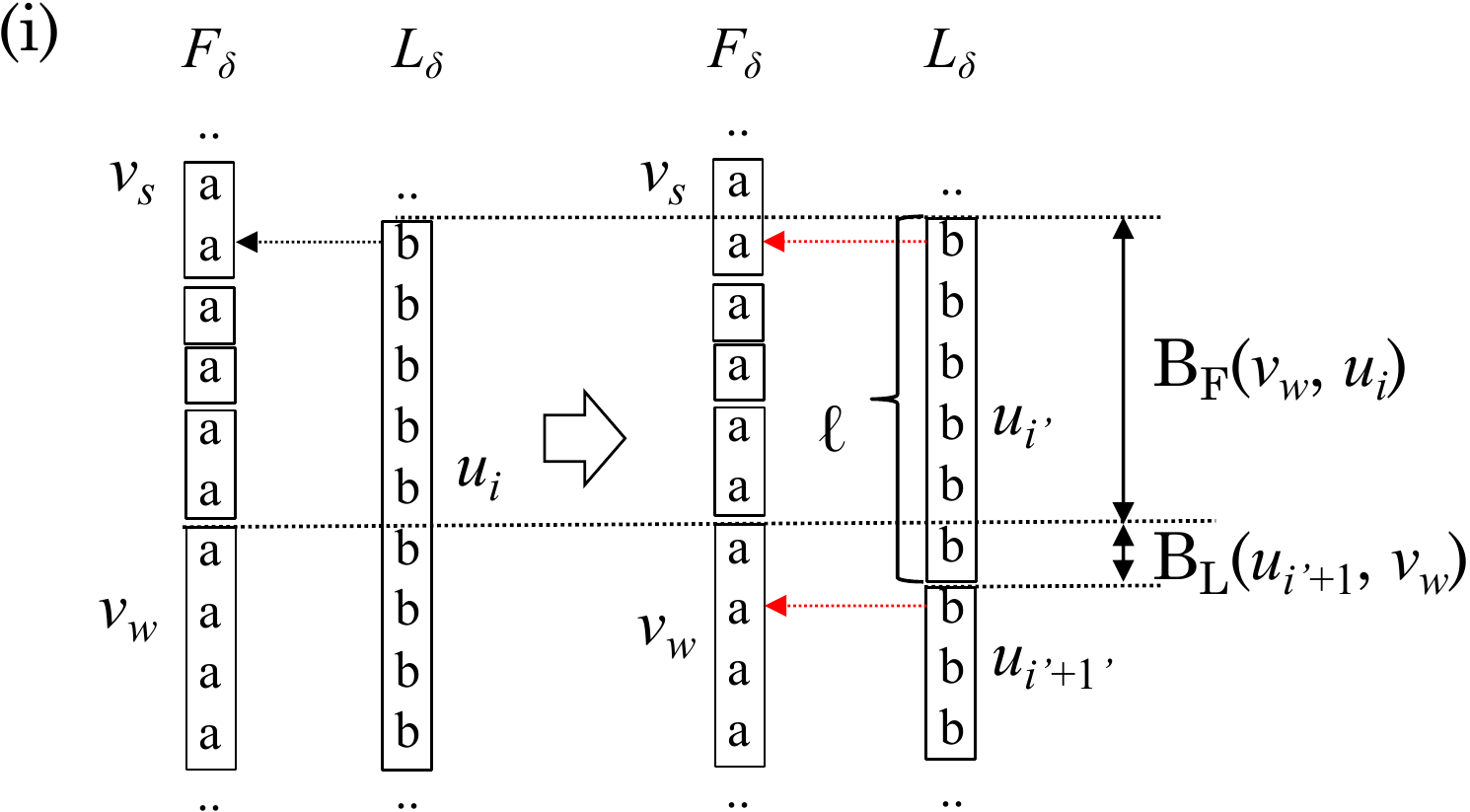}
		\includegraphics[scale=0.4]{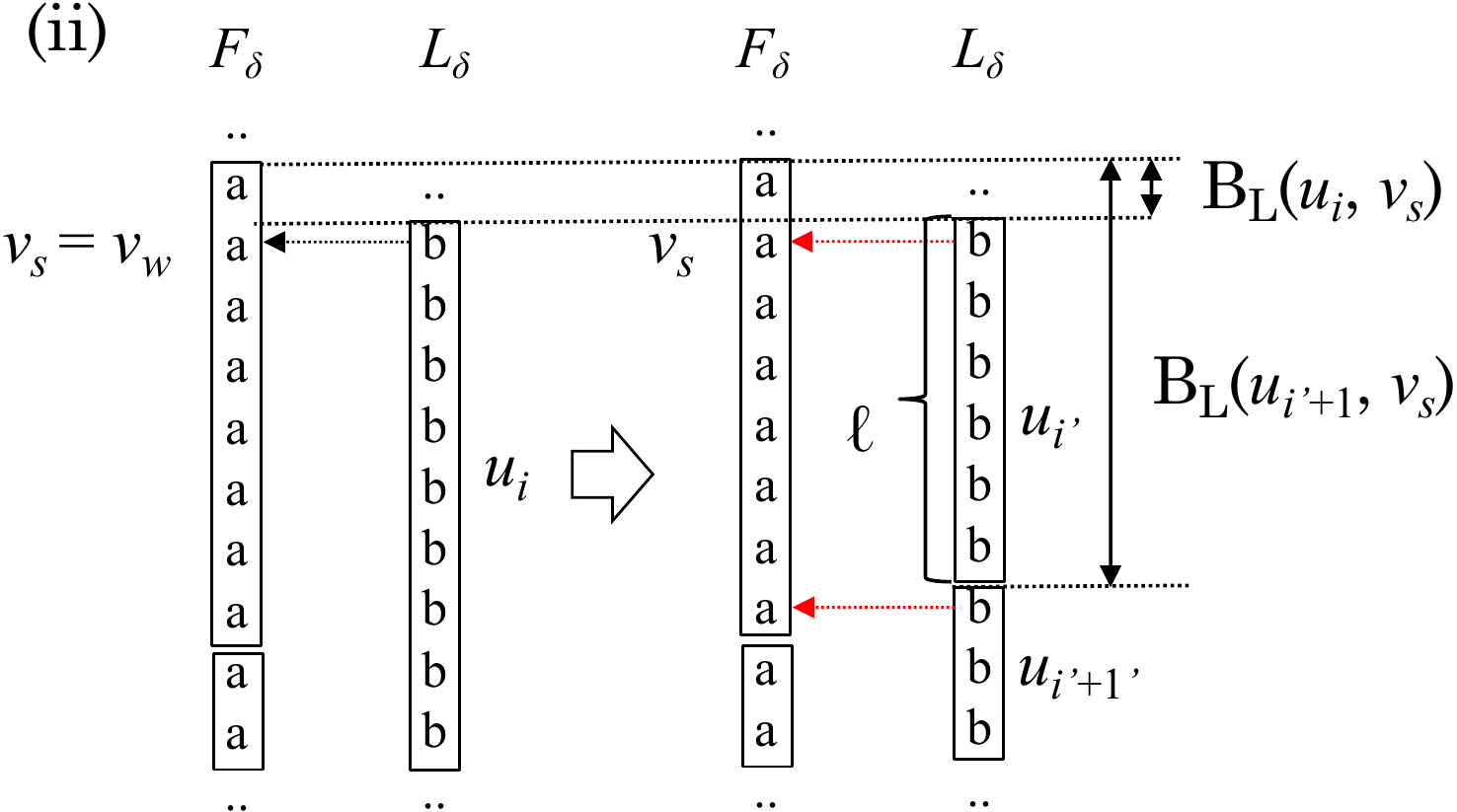}
     \end{center}
     \begin{center}
		\includegraphics[scale=0.4]{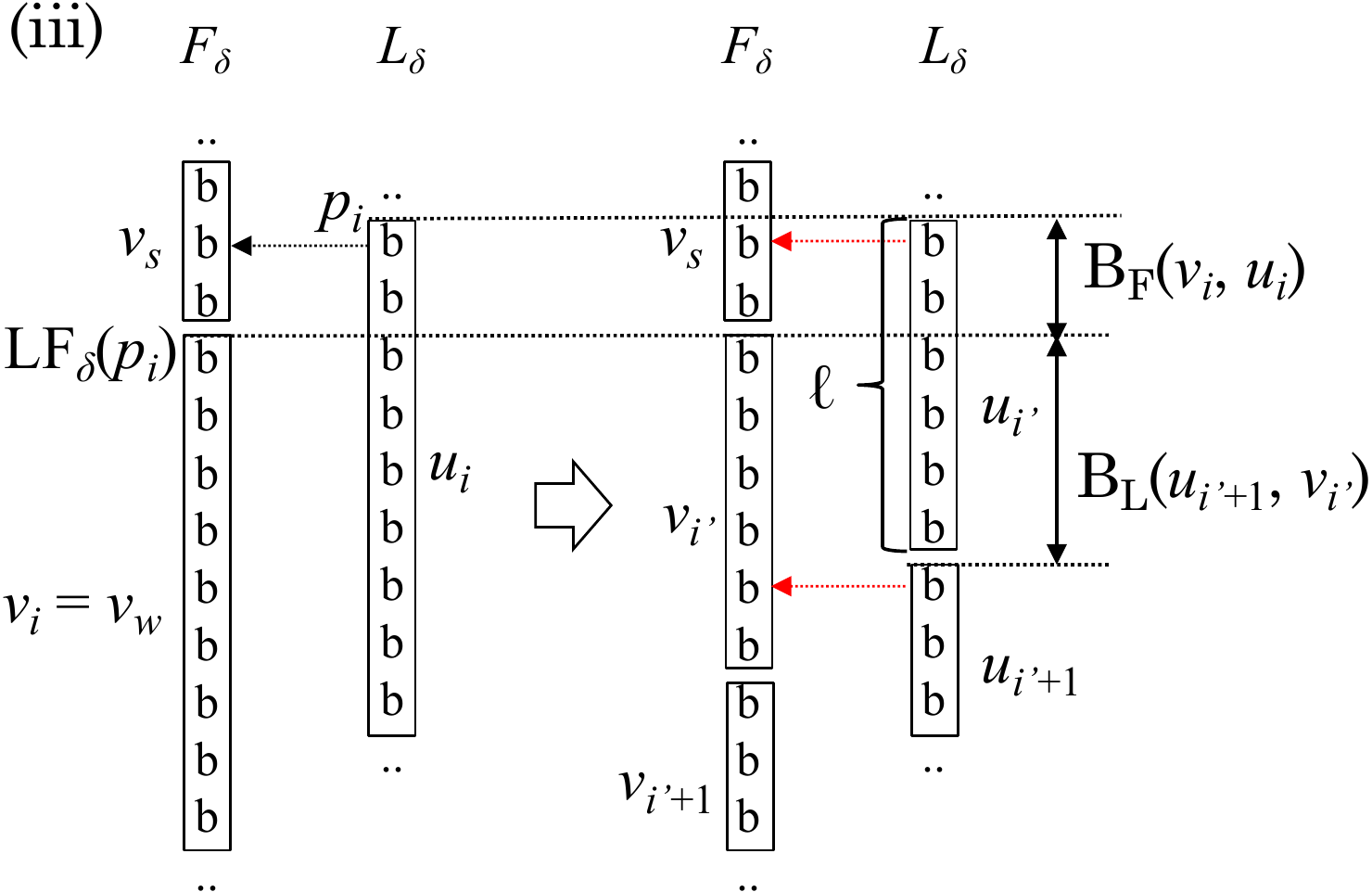}
		\includegraphics[scale=0.4]{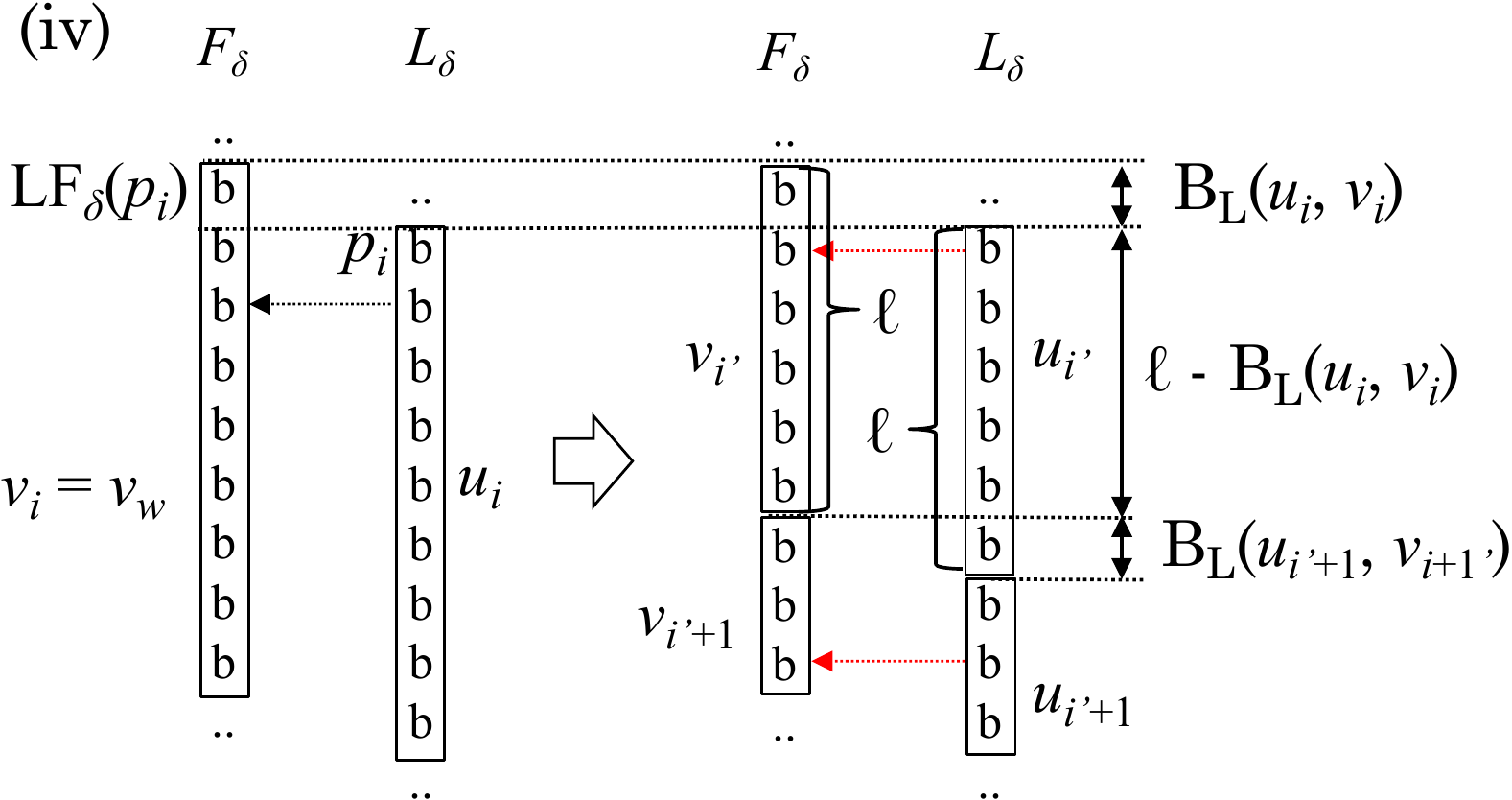}
     \end{center}

	  \caption{
	  Four cases for determining the new directed edge starting from the new node $u_{i^{\prime}+1}$ created by the balancing operation. 
	  Red arrows are the new directed edges.
	  }
 \label{fig:split_case}
\end{figure}

We explain the algorithm of the balancing operation. 
The balancing operation consists of iterations.  
We suppose that (i) each iteration of the balancing operation processes the LF-interval graph $\graphds(D^{O(\alpha)}_{\delta})$ 
for an $O(\alpha)$-balanced DBWT $D^{O(\alpha)}_{\delta} = L_{\delta}[p_{1}..(p_{2}-1)], L_{\delta}[p_{2}..(p_{3}-1)], \ldots, L_{\delta}[p_{k}..(p_{k+1}-1)]$ with $O(\alpha)$ $\alpha$-heavy DBWT-repetitions and $O(\alpha)$ $\alpha$-heavy F-intervals, 
and (ii) the B-tree of $V$ in $\graphds(D^{O(\alpha)}_{\delta})$ stores only nodes satisfying one of the three conditions of Lemma~\ref{lem:mod}. 
Let $(u_{i}, v_{i}) \in E_{LF}$ be an undirected edge such that 
node $u_{i} \in U$ represents an $\alpha$-heavy DBWT-repetition, 
or node $v_{i} \in V$ represents an $\alpha$-heavy F-interval. 
The two nodes $u_{i}$ and $v_{i}$ can be computed in $O(1)$ time using the two arrays stored in 
the LF-interval graph $\graphds(D^{O(\alpha)}_{\delta})$ unless the two arrays are empty. 
If the two arrays are empty, then DBWT $D^{O(\alpha)}_{\delta}$ is $\alpha$-balanced.
Hence, the balancing operation stops and returns the LF-interval graph $\graphds(D^{O(\alpha)}_{\delta})$. 

One iteration of the balancing operation consists of three steps. 
The first step replaces node $u_{i}$ labeled $(L_\delta[p_{i}], p_{i+1}-p_{i})$ with two new nodes $u_{i^{\prime}}$ and $u_{i^{\prime}+1}$ in the doubly linked list of $U$, where $u_{i^\prime}$ is previous to $u_{i^\prime+1}$. 
The labels of the nodes $u_{i^{\prime}}$ and $u_{i^{\prime}+1}$ are $(L_\delta[p_{i}], \ell)$ and $(L_\delta[p_{i}], p_{i+1} - p_{i} - \ell)$, respectively, using non-negative integer $\ell \in \{ 1, 2, \ldots, p_{i+1} - p_{i} - 1 \}$. 
The non-negative integer $\ell$ is set to a value using an edge label in $B_F$ according to whether node $u_i$ represents an $\alpha$-heavy DBWT-repetition or not. 

\subparagraph{Case 1: node $u_{i}$ represents an $\alpha$-heavy DBWT-repetition.}
Node $v_{\pi_{j}} \in V$ is the most forward node on the doubly linked list of $V$ of the nodes connected to $u_{i}$ by directed edges in set $E_{F}$ 
(i.e., $j = \min \{ x \mid 1 \leq x \leq k \mbox{ s.t. } (v_{\pi_{x}}, u_{i}) \in E_{F} \}$, 
where $\pi$ is the permutation introduced in Section~\ref{sec:dynamic});  
$t$ is the number of directed edges connected to $u_{i}$ in $E_{F}$~(i.e., $t = |\{ x \mid 1 \leq x \leq k \mbox{ s.t. } (v_{x}, u_{i}) \in E_{F} \}|$);  
$B_{F}(v_{\pi_{j + \lceil t/2 \rceil}}, u_{i})$ is the label of the directed edge from node $v_{\pi_{j + \lceil t/2 \rceil}}$ to $u_{i}$;  
$\ell$ is set to $B_{F}(v_{\pi_{j + \lceil t/2 \rceil}}, u_{i})$.

\subparagraph{Case 2: $u_{i}$ does not represent an $\alpha$-heavy DBWT-repetition, and $v_{i}$ represents an $\alpha$-heavy F-interval.}
Node $u_{j^{\prime}} \in U$ is the most forward node on the doubly linked list of $U$ of the nodes connected to $v_{i}$ by directed edges in set $E_{L}$ (i.e., $j^{\prime} = \min \{ x \mid 1 \leq x \leq k \mbox{ s.t. } (u_{x}, v_{i}) \in E_{L} \}$;
$t^{\prime}$ is the number of directed edges connected to $v_{i}$ in $E_{L}$; 
$B_{L}(u_{j^{\prime} + \lceil t^{\prime}/2 \rceil}, v_{i})$ is the label of the directed edge from node $u_{j^{\prime} + \lceil t^{\prime}/2 \rceil}$ to $v_{i}$.  
Integer $\ell$ is set to $B_{L}(u_{j^{\prime} + \lceil t^{\prime}/2 \rceil}, v_{i})$. 

Similarly, $v_{i} \in V$ is replaced with two new nodes $v_{i^{\prime}}$ and $v_{i^{\prime}+1}$ in the doubly linked list of $V$, 
and $v_{i^{\prime}}$ (respectively, $v_{i^{\prime}+1}$) has the same labels as $u_{i^{\prime}}$ (respectively, $u_{i^{\prime}+1}$). 
This step takes $O(\alpha)$ time. 

Figure~\ref{fig:split} illustrates an LF-interval graph that is balanced by splitting a node. 
Node $u_{5}$ labeled $(b, 3)$ is replaced with $u_{8}$ labeled $(b, 2)$ and $u_{9}$ labeled $(b, 1)$. 
Similarly, $v_{5}$ labeled $(b, 3)$ is replaced with $v_{8}$ labeled $(b, 2)$ and $v_{9}$ labeled $(b, 1)$. 

The second step updates the edges and their labels~(i.e., the five sets $E_{LF}, E_{L}, E_{F}, B_{L}$, and $B_{F}$) according to the splits of the two nodes $u_{i}$ and $v_{i}$. 
This step is similar to the update-edge step of the update operation, 
and the algorithms for updating the five sets take $O(\alpha)$ time, which is explained below. 

\subparagraph{Updating set $E_{LF}$.} 
Remove edge $(u_{i}, v_{i})$ from set $E_{LF}$, 
and insert two new edges $(u_{i^{\prime}}, v_{i^{\prime}})$ and $(u_{i^{\prime}+1}, v_{i^{\prime}+1})$ into the set for the four new nodes $u_{i^{\prime}}, u_{i^{\prime}+1}, v_{i^{\prime}}$, and $v_{i^{\prime}+1}$.
\subparagraph{Updating sets $E_{L}$ and $B_{L}$.}
The directed edges pointing to node $v_{i}$ are removed from set $E_{L}$. 
Formally, let $X \subseteq U$ be the set of nodes connected to the tails of the directed edges pointing to $v_{i}$~(i.e., $X = \{ u \mid (u, v) \in E_{L} \mbox{ s.t. } v = v_{i} \}$). 
Then, we remove all the directed edges starting from nodes in set $X$ from $E_{L}$. 
Similarly, the directed edge starting at $u_{i}$ is removed from $E_{L}$. 
The number of directed edges removed from $E_{L}$ is at most $(1 + |X|)$. 
Then, $|X| = O(\alpha)$ because $D^{O(\alpha)}_{\delta}$ is $O(\alpha)$-balanced.

By contrast, new directed edges starting from nodes in $X$ are inserted into set $E_{L}$. 
The new directed edge starting at each node $u \in X$ points to $v_{i^{\prime}}$ or $v_{i^{\prime}+1}$. 
If $B_{L}(u, v_{i}) < \ell$ for the length $\ell$ of the DBWT-repetition represented by new node $u_{i^{\prime}}$ 
and the label $B_{L}(u, v_{i})$ of the removed edge $(u, v_{i})$, 
then the new directed edge is $(u, v_{i^{\prime}})$; 
otherwise, it is $(u, v_{i^{\prime}+1})$. 
The label of the directed edge $(u, v_{i^{\prime}})$ is set to the label $B_{L}(u, v_{i})$ of the removed edge. 
Similarly, 
the label of the directed edge $(u, v_{i^{\prime}+1})$ is set to $(B_{L}(u, v_{i}) - \ell)$.

Next, the directed edge starting at new node $u_{i^{\prime}}$ is inserted into set $E_{L}$.
Let $v_{s} \in V$ be the node connected to $u_{i}$ by directed edge $(u_{i}, v_{s}) \in E_{L}$. 
Then, the directed edge starting at new node $u_{i^{\prime}}$ is determined by the following three cases: 
(i) $v_{s} \neq v_{i}$, (ii) $v_{s} = v_{i}$ and $B_{L}(u_{i}, v_{s}) < \ell$, and (iii) $v_{s} = v_{i}$ and $B_{L}(u_{i}, v_{s}) \geq \ell$, where $B_{L}(u_{i}, v_{s})$ is the label of directed edge $(u_{i}, v_{s})$, which is removed from $E_{L}$. 
For case (i), the directed edge starting at $u_{i^{\prime}}$ 
points to $v_{s}$ and its label $B_{L}(u_{i^{\prime}}, v_{s})$ is set to $B_{L}(u_{i}, v_{s})$. 
For case (ii), the directed edge starting at $u_{i^{\prime}}$ points to $v_{i^{\prime}}$ and its label $B_{L}(u_{i^{\prime}}, v_{i^{\prime}})$ is set to $B_{L}(u_{i}, v_{s})$. 
For case (iii), the directed edge starting at $u_{i^{\prime}}$ points to $v_{i^{\prime}+1}$ and 
its label $B_{L}(u_{i^{\prime}}, v_{i^{\prime}+1})$ is set to $(B_{L}(u_{i}, v_{s}) - \ell)$.

Next, the directed edge starting at new node $u_{i^{\prime}+1}$ is inserted into set $E_{L}$.
Let $v_{w} \in V$ be the node such that 
its F-interval $[\LF_{\delta}(p_{w}), \LF_{\delta}(p_{w+1}-1)]$ 
contains position $p_{i} + \ell$. 
Then, the directed edge starting at $u_{i^{\prime}+1}$ is determined by the following two cases: 
(i) $v_{w} \neq v_{i}$ or (ii) $v_{w} = v_{i}$.

For case (i), $v_{w}$ is either equal to $v_{s}$ or not equal to it.
If $v_{w} \neq v_{s}$, 
the new directed edge starting at $u_{i^{\prime}+1}$ points to $v_{w}$ and its label $B_{L}(u_{i^{\prime}+1}, v_{w})$ is set to $(\ell - B_{F}(v_{w}, u_{i}))$~(see also Figure~\ref{fig:split_case}-(i)). 
Otherwise, 
the new directed edge points to $v_{s}$ and its label $B_{L}(u_{i^{\prime}+1}, v_{s})$ is set to $(\ell + B_{L}(u_{i}, v_{s}))$~(see also Figure~\ref{fig:split_case}-(ii)).

For case (ii), node $v_{w} = v_{i}$ is split into two new nodes $v_{i^{\prime}}$ 
and $v_{i^{\prime}+1}$. If $p_{i} < \LF_{\delta}(p_{i})$, 
then the new directed edge starting at $u_{i^{\prime}+1}$ points to 
$v_{i^{\prime}}$ and its label $B_{L}(u_{i^{\prime}+1}, v_{i^{\prime}})$ is set to  
$(\ell - B_{F}(v_{i}, u_{i}))$~(see also Figure~\ref{fig:split_case}-(iii)). 
Otherwise, the new directed edge starting at $u_{i^{\prime}+1}$ points to $v_{i^{\prime}+1}$ and 
its label $B_{L}(u_{i^{\prime}+1}, v_{i^{\prime}+1})$ is set to $B_{L}(u_{i}, v_{i})$~(see also Figure~\ref{fig:split_case}-(iv)).

The following two lemmas can be used to determine node $v_{w}$ and verify whether $p_{i} < \LF_{\delta}(p_{i})$ in $O(\alpha)$ time. 
\begin{lemma}\label{lem:split_update_edge_L}
We can compute $v_{w}$ in $O(\alpha)$ time. 
\end{lemma}
\begin{proof}
If $B_{L}(u_{i}, v_{s}) + \ell \leq \ell^{\prime}$ for the length $\ell^{\prime}$ of the F-interval represented by $v_{s}$ 
and the label $B_{L}(u_{i}, v_{s})$ of directed edge $(u_{i}, v_{s})$ in $E_{L}$, 
then $v_{w} = v_{s}$. 
Otherwise, 
the directed edge starting at $v_{w}$ points to $u_{i}$, 
and $\ell \in [B_{F}(v_{w}, u_{i}) + 1, B_{F}(v_{w}, u_{i}) + \hat{\ell}]$ for the length $\hat{\ell}$ of the F-interval represented by $v_{w}$ 
and the label $B_{F}(v_{w}, u_{i})$ of directed edge $(v_{w}, u_{i})$. 
We find $v_{w}$ in the set of nodes consisting of (i) $v_{s}$ and 
(ii) the nodes connected to $u_{i}$ by directed edges in $E_{F}$. 
The number of directed edges pointing to $u_{i}$ is $O(\alpha)$ because $D^{O(\alpha)}_{\delta}$ is $O(\alpha)$-balanced. 
Hence, $v_{w}$ can be computed in $O(\alpha)$ time. 
\end{proof}
\begin{lemma}\label{lem:split_verify}
We can verify whether $p_{i} < \LF_{\delta}(p_{i})$ in $O(1)$ time without knowing $p_{i}$ and $\LF_{\delta}(p_{i})$. 
\end{lemma}
\begin{proof}
Let $u_{x} \in U$ be the node connected to $v_{i}$ by a directed edge in $E_{F}$. 
Then, $p_{i} < \LF_{\delta}(p_{i})$ if and only if either of the following two conditions holds: 
(i) $u_{x} \neq u_{i}$, and $u_{i}$ precedes $u_{x}$ in the doubly linked list of $U$; 
(ii) $u_{x} = u_{i}$.
We can verify whether $u_{i}$ precedes $u_{x}$ or not in $O(1)$ time using the order maintenance data structure. Hence, Lemma~\ref{lem:split_verify} holds. 
\end{proof}
Hence, updating sets $E_{L}$ and $B_{L}$ takes $O(\alpha)$ time in total.

\subparagraph{Updating sets $E_{F}$ and $B_{F}$.}
Updating sets $E_{F}$ and $B_{F}$ is symmetric to updating sets $E_{L}$ and $B_{L}$. 
Hence, we can update $E_{F}$ and $B_{F}$ in $O(\alpha)$ time in total.

The third step updates the order maintenance data structure, 
the B-tree of $V$, and the two arrays for $\alpha$-heavy DBWT-repetitions and F-intervals according to the splits of nodes $u_{i}$ and $v_{i}$. 
This step takes $O(\alpha)$ time, which is explained below. 

\subparagraph{Updating the  order maintenance data structure.}
Node $u_{i}$ is removed from the order maintenance data structure. 
By contrast, two new nodes $u_{i^{\prime}}$ and $u_{i^{\prime}+1}$ are inserted into the data structure. 
Updating the order maintenance data structure takes $O(1)$ time.

\subparagraph{Updating the B-tree of $V$.}
The target nodes deleted from the B-tree can be limited to node $v_{i}$. 
By contrast, the target nodes inserted into the B-tree of $V$ are limited to new node $v_{i^{\prime}+1}$. 
Node $v_{i}$ is deleted from the B-tree and the new node $v_{i^{\prime}+1}$ is inserted into the tree if $v_{i}$ is stored in the tree; otherwise, $v_{i^{\prime}+1}$ is not inserted into the tree.
If $v_{i}$ is contained in the B-tree of $V$, 
then $v_{i}$ is deleted from the B-tree and $v_{i^{\prime}+1}$ is inserted. 
Otherwise, $v_{i^{\prime}+1}$ is not inserted into the B-tree of $V$. 
This procedure ensures that the B-tree of $V$ stores only nodes satisfying one of the three conditions of Lemma~\ref{lem:mod} after the LF-interval graph has been updated. 
Because $v_{i}$ is replaced with $v_{i^{\prime}+1}$ in the doubly linked list of $V$, 
the deletion and insertion take $O(1)$ time from Lemma~\ref{lem:B_tree_order}.

\subparagraph{Updating the array for $\alpha$-heavy DBWT-repetitions.}
The target nodes deleted from the array for $\alpha$-heavy DBWT-repetitions can be limited to node $u_{i}$. 
Similarly, the target nodes inserted into the array can be limited to two nodes $u_{h}$ and $u_{h^{\prime}}$ 
for (i) the node $u_{h}$ connected to new node $v_{i^{\prime}}$ by directed edge $(v_{i^{\prime}}, u_{h}) \in E_{F}$ 
and (ii) the node $u_{h^{\prime}}$ connected to new node $v_{i^{\prime}+1}$ by directed edge $(v_{i^{\prime}+1}, u_{h^{\prime}}) \in E_{F}$. 
The array contains $O(\alpha)$ nodes because 
DBWT $D^{O(\alpha)}_{\delta}$ has $O(\alpha)$ $\alpha$-heavy DBWT-repetitions. 
Hence updating the array takes $O(\alpha)$ time using Lemma~\ref{lem:verify_heavy}.

\subparagraph{Updating the array for $\alpha$-heavy F-intervals.}
Updating the array for $\alpha$-heavy F-intervals is symmetric to updating the array for $\alpha$-heavy DBWT-repetitions. 
The array for $\alpha$-heavy F-intervals can be updated in $O(\alpha)$ time. 

Hence, the three steps take $O(\alpha)$ time if 
the balancing operation processes the LF-interval graph of an $O(\alpha)$-balanced DBWT at the iteration. 

\egroup

The following lemma concerning the theoretical results on the balancing operation holds. 
\begin{lemma}\label{lem:split_time}
For the LF-interval graph $\graphds(D^{2\alpha+1}_{\delta})$ for a $(2\alpha+1)$-balanced DBWT including at most two $\alpha$-heavy DBWT-repetitions and at most two $\alpha$-heavy F-intervals, 
we assume that the B-tree of $V$ in the LF-interval graph contains only nodes satisfying one of the three conditions of Lemma~\ref{lem:mod}. 
Then, balancing operation $\balance(\graphds(D^{2\alpha+1}_{\delta}))$ takes $O(\alpha)$ time per iteration for all $\alpha \geq 4$, 
and the B-tree of $V$ in the outputted LF-interval graph contains only nodes satisfying one of the three conditions of Lemma~\ref{lem:mod}. 
\end{lemma}
\begin{proof}
See Appendix~\ref{app:split_time}. 
\end{proof}

\section{R-comp algorithm}\label{sec:rcomp_algorithm}

In this section, we present the r-comp algorithm, and we also present its space and time complexities. 
The r-comp algorithm takes input string $T$ and parameter $\alpha \geq 2$, and it outputs the RLBWT of $T$. 
\oldsentence{
The pseudo-code of r-comp can be found in Appendix~\ref{app:rcomp}. 
}

The algorithm reads $T$ from its end to its beginning~(i.e., $T[n], T[n-1], \ldots, T[1]$), and 
it gradually builds the LF-interval graph $\graphds(D^{\alpha}_n)$ for an $\alpha$-balanced DBWT $D^\alpha_{n}$ of BWT $L_{n}$ of $T$. 
At each $\delta \in \{ 1, 2, \ldots, n-1 \}$, LF-interval graph $\graphds(D^{\alpha}_{\delta})$ for 
an $\alpha$-balanced DBWT $D^{\alpha}_{\delta}$ of BWT $L_{\delta}$ of $T[(n- \delta + 1)..n]$~(i.e., suffix $T_{\delta}$) is built. 
For the node $u_{i} \in U$ representing special character $\$$ in the doubly linked list of set $U$ of nodes 
in LF-interval graph $\graphds(D^{\alpha}_{\delta})$, 
two nodes $u_{i-1} \in U$ and $u_{i+1} \in U$ are previous and next to node $u_{i}$, respectively, in the list. 
If neither the label of $u_{i-1}$ nor the label of $u_{i+1}$ includes the $(n-\delta)$-th character $c$ of $T$~(i.e., the first character $c$ of $T_{\delta+1}$), 
the update operation  $\update(\graphds(D^{\alpha}_{\delta}), c)$ described in Section~\ref{sec:slow_update} is applied; otherwise the fast update operation  $\fastUpdate(\graphds(D^{\alpha}_{\delta}), c)$ described in Section~\ref{sec:fast_update} is applied. 
Both the update operation and the fast update operation output an LF-interval graph $\graphds(D^{2\alpha+1}_{\delta+1})$ for 
DBWT $D^{2\alpha+1}_{\delta+1}$ of BWT $L_{\delta+1}$ that is not $\alpha$-balanced. 
Thus, the r-comp algorithm balances the LF-interval graph such that it represents an $\alpha$-balanced DBWT $D^{\alpha}_{\delta+1}$ using the balancing operation in Section~\ref{sec:balancing} as $\graphds(D^{\alpha}_{\delta+1})$. 
After $(n-1)$ iterations of those steps, the LF-interval graph $\graphds(D^{\alpha}_n)$ for $\alpha$-balanced DBWT $D^{\alpha}_{n}$ of BWT $L_{n}$ is obtained; 
it is then converted into $D^{\alpha}_{n}$ using the doubly linked list of $U$.
Finally, $D^{\alpha}_{n}$ is converted into the RLBWT of $T$.

\newsentence\bgroup
\begin{algorithm*}[t]
\caption{R-comp algorithm. The algorithm takes string $T$ and parameter $\alpha \geq 2$ as input, and it outputs the RLBWT of $T$. $n:$ length of $T$; $D^{\alpha}_{\delta}:$ $\alpha$-balanced DBWT of BWT $L_\delta$ of $T[(n-\delta + 1)..n]$; $\graphds(D^{\alpha}_{\delta}):$ the LF-interval graph of $\alpha$-balanced DBWT $D^{\alpha}_{\delta}$.}
    \begin{algorithmic}[1]
    \Function{r-comp}{$T$, $\alpha$}
        \State{$D^{\alpha}_{1} \leftarrow \$$} \Comment{Initialize $D^{\alpha}_{1}$ as special character $\$$}
        \State{build $\graphds(D^{\alpha}_{1})$}
        \For{$\delta = 1, 2, \ldots, n-1$}\label{code:for}
        \State{$c \leftarrow T[n-\delta]$}\label{code:read} \Comment{Read character $c$ from $T[n - \delta]$}       
        \If{neither the label of $u_{i-1}$ nor the label of $u_{i+1}$ includes character $c$}\label{code:if} 
            \State{$\graphds(D^{2\alpha+1}_{\delta+1}) \leftarrow \update(\graphds(D^{\alpha}_{\delta}), c)$}\label{code:replace_insert} \Comment{The update operation in Sec.~\ref{sec:slow_update}}   
        \Else
            \State{$\graphds(D^{2\alpha+1}_{\delta+1}) \leftarrow \fastUpdate(\graphds(D^{\alpha}_{\delta}), c)$} \Comment{The fast update operation in Sec.~\ref{sec:fast_update}}\label{code:balance}
        \EndIf\label{code:if_end}
        \State{$\graphds(D^{\alpha}_{\delta+1}) \leftarrow \balance(\graphds(D^{2\alpha+1}_{\delta+1}))$} \Comment{The balancing operation in Sec.~\ref{sec:balancing}}
        \EndFor\label{code:for_end}        
        \State{Recover $D^{\alpha}_{n}$ from $\graphds(D^{\alpha}_{n})$}\label{code:recover}   
        \State{Convert $D^{\alpha}_{n}$ into the RLBWT of $T$}\label{code:convert}
        \State{Return the RLBWT}    
    \EndFunction
    \end{algorithmic}
    \label{algo:rcomp}
\end{algorithm*}

The pseudo-code of r-comp is given in Algorithm~\ref{algo:rcomp}. 
In Lines~\ref{code:for}-\ref{code:for_end}, 
r-comp creates the LF-interval graph for $\alpha$-balanced DBWT $D^{\alpha}_{\delta+1}$ of suffix $T_{\delta+1}$ 
using update and balancing operations of LF-interval graph for $1 \leq \delta \leq n-1$. 
In Lines~\ref{code:if}-\ref{code:if_end}, 
r-comp creates the LF-interval graph for $(2\alpha+1)$-balanced DBWT $D^{2\alpha+1}_{\delta+1}$ of $T_{\delta+1}$ by applying 
update operation $\update(\graphds(D^{\alpha}_{\delta}), c)$ or fast update operation $\fastUpdate(\graphds(D^{\alpha}_{\delta}), c)$ 
to the LF-interval graph for $\alpha$-balanced DBWT $D^{\alpha}_{\delta}$ of $T_{\delta}$. 
At Line~\ref{code:balance}, r-comp creates 
the LF-interval graph for $\alpha$-balanced DBWT $D^{\alpha}_{\delta+1}$ by applying the balancing operation to the LF-interval graph obtained by the update operation. 
At Line~\ref{code:recover}, r-comp recovers $\alpha$-balanced DBWT $D^{\alpha}_{n}$ of $T_{n} = T$ 
from its LF-interval graph $\graphds(D^{\alpha}_{n})$. 
At Line~\ref{code:convert}, r-comp converts $D^{\alpha}_{n}$ into the RLBWT of $T$. 
\egroup

\subsection{Space and time complexities}\label{sec:analysis}

\textbf{Space complexity.}
The r-comp algorithm requires $O(k \log n)$ bits of space for $k$ DBWT-repetitions in the DBWT $D^{\alpha}_{n}$ 
because the LF-interval graph $\graphds(D^{\alpha}_{n})$ for the DBWT requires $O(k \log n)$ bits of space. 
The value of $k$ depends on the number of executions of update, fast update, and balancing operations executed by the r-comp algorithm. 
The following lemma ensures that $k$ can be bounded by $O(r)$. 
\begin{lemma}\label{lem:modified_k_bound_lemma}
We modify the fast update operation. 
Then, the following three statements hold: 
(i) this modification does not affect Lemma~\ref{lem:fast_update_result}; 
(ii) $k \leq r + k_{\splitop}$; (iii) 
$k_{\splitop} \leq \frac{2r}{\lceil \alpha / 2 \rceil - 7}$ holds for any constant $\alpha \geq 16$. 
\end{lemma}
\begin{proof}
See Appendix~\ref{app:modified_k_bound_lemma}. 
\end{proof}
Because $k = O(r)$ by Lemma~\ref{lem:modified_k_bound_lemma}, the following theorem is obtained. 
\begin{theorem}\label{thm:rcomp_working_space}
The r-comp algorithm takes $O(r \log n)$ bits of working space for $\alpha \geq 16$. 
\end{theorem}

\textbf{Time complexity.}
The bottleneck of r-comp is the update operation of LF-interval graph with $O(\alpha + \log k)$ time in Section~\ref{sec:slow_update}. 
The number of executions of the update operation can be bounded by $O(r)$. 
This fact indicates that we can bound the running time of r-comp by $O(\alpha n + r \log k)$, i.e., $O(\alpha n + r \log r)$. 
Finally, we obtain the following theorem. 
\begin{theorem}\label{theo:rcomp_time}
R-comp runs in $O(\alpha n + r \log r)$ time for $\alpha \geq 16$. 
\end{theorem}
\oldsentence{
\begin{proof}
See Appendix~\ref{app:rcomp_time}.
\end{proof}
}
\newsentence\bgroup
\begin{proof}
R-comp consists of update and balancing operations of LF-interval graph. 
The B-tree of $V$ in the LF-interval graph for DBWT $D_{1}^{\alpha}$ 
contains only nodes satisfying one of the three conditions of Lemma~\ref{lem:mod}. 
This fact indicates that each operation of r-comp takes an LF-interval graph such that 
the B-tree of $V$ contains only nodes satisfying one of the three conditions of Lemma~\ref{lem:mod}. 
Because Lemmas~\ref{lem:slow_update_result}, \ref{lem:fast_update_result}, and \ref{lem:split_time} hold, 
the running time of r-comp is $O(n + k_{\mathsf{slow}} \log k + \alpha (k_{\mathsf{fast}} + k_{\splitop}))$ 
for (i) the number $k_{\mathsf{slow}}$ of executions of the update operation in Section~\ref{sec:slow_update}, 
(ii) the number $k_{\mathsf{fast}}$ of executions of the fast update operation in Section~\ref{sec:fast_update}, 
and (iii) $\alpha \geq 16$. 

We show that $k_{\mathsf{slow}} = O(r)$ for the number $r$ of BWT-runs in BWT $L_{n}$. 
When an update operation $\update(\graphds(D_{\delta}^{\alpha}), c)$ is executed for a DBWT $D_{\delta}^{\alpha}$ of BWT $L_{\delta}$, 
the BWT is changed into the next BWT $L_{\delta+1}$ using a replacement of special character $L_{\delta+1}[\reppos] = \$$ with input character $c$. 
The replacement creates a new BWT-run because 
the update operation ensures that the characters adjacent to special character $\$$ are not $c$~(i.e., $L_{\delta}[\reppos-1] \neq c$ and $L_{\delta}[\reppos+1] \neq c$). 
Hence, the BWT $L_{n}$ of $T$ consists of at least $k_{\mathsf{slow}}$ BWT-runs. 
By contrast, the BWT $L_{n}$ of $T$ consists of $r$ BWT-runs, 
and hence $k_{\mathsf{slow}} \leq r$. 

Clearly, $k_{\mathsf{fast}} = O(n)$. 
$k_{\splitop} = O(r)$ by Lemma~\ref{lem:modified_k_bound_lemma}-(iii). 
Because $k = O(r)$, 
the running time of r-comp is $O(\alpha n + r \log r)$. 
\end{proof}
\egroup

\section{Experiments}\label{sec:exp}
\subparagraph{Setup.}
\newcommand{\Times}[1]{#1$\times$}
\newcommand{\TimesTBL}[1]{(#1)}
\newcommand{\Dataset}[1]{\textsf{#1}}
\newcommand{\RCompG}{r-comp$_\textrm{saving}$}
\newcommand{\RCompS}{r-comp}

\begin{table*}[ht]
\footnotesize
\centering
\caption{
Statistics of datasets.
}
\label{tab:dataset}
\begin{tabular}{l||r|r|r|r|r|r}
\hline
String & $\sigma$ & $n$ [$10^3$] & $r$ [$10^3$] & $\PFP$ [$10^3$] & $n/r$ & $\PFP/r$ \\
\hline\hline
\Dataset{cere} & 5 & 461,287 & 11,575 & 94,221 & 40 & 8 \\
\Dataset{coreutils} & 236 & 205,282 & 4,684 & 40,079 & 44 & 9 \\
\Dataset{einstein.de.txt} & 117 & 92,758 & 101 & 1,949 & 915 & 19 \\
\Dataset{einstein.en.txt} & 139 & 467,627 & 290 & 7,596 & 1,611 & 26 \\
\Dataset{Escherichia\_Coli} & 15 & 112,690 & 15,044 & 53,316 & 7 & 4 \\
\Dataset{influenza} & 15 & 154,809 & 3,023 & 50,403 & 51 & 17 \\
\Dataset{kernel} & 160 & 257,962 & 2,791 & 17,156 & 92 & 6 \\
\Dataset{para} & 5 & 429,266 & 15,637 & 88,485 & 27 & 6 \\
\Dataset{world\_leaders} & 89 & 46,968 & 573 & 10,919 & 82 & 19 \\
\hline
\Dataset{boost} & 96 & 1,073,769 & 65 & 8,871 & 16,597 & 137 \\
\Dataset{samtools} & 112 & 1,074,236 & 629 & 14,767 & 1,708 & 23 \\
\Dataset{sdsl} & 126 & 1,076,495 & 597 & 14,209 & 1,803 & 24 \\
\hline
\Dataset{enwiki} & 207 & 37,849,201 & 70,190 & 975,218 & 539 & 14 \\
\hline
\Dataset{chr19.1000} & 5 & 59,125,169 & 45,143 & 816,365 & 1,310 & 18 \\
\hline
\end{tabular}
\end{table*}

We empirically tested the performance of the r-comp algorithm on strings from four datasets with different types of highly repetitive strings:
(i) nine strings from the Pizza\&Chili repetitive corpus \cite{dataset:pc-repetitive-corpus};
(ii) three strings (boost, samtools, and sdsl) of the latest revisions of Git repositories, each of which is 1GB in size;
(iii) a 37GB string (enwiki) of English Wikipedia articles with a complete edit history~\cite{dataset:enwiki-all-pages}; and
(iv) a 59GB string (\Dataset{chr19.1000}) obtained by concatenating chromosome 19 from 1,000 human genomes in the 1000 Genomes Project~\cite{1000Genomes}.
Table \ref{tab:dataset} shows the relevant statistics for each string in the datasets.
The ratio $n/r$ of string length $n$ to the number $r$ of BWT-runs in BWT is the compression ratio of each string. 
The strings with high compression ratios are versions of Wikipedia articles (\Dataset{einstein.de.txt} and \Dataset{einstein.en.txt}), revisions of Git repositories (\Dataset{boost}, \Dataset{samtools}, and \Dataset{sdsl}), and 1000 human genomes (\Dataset{chr19.1000}). 
The ratio $\PFP/r$ is the ratio of the total size $\PFP$ of the dictionary and factorization by prefix-free parsing to the size of the RLBWT for a string.

We implemented two versions of the r-comp algorithm (i.e., \RCompS{} and \RCompG{}). 
Here, \RCompS{} is the straightforward implementation of the r-comp algorithm presented in Section \ref{sec:rcomp_algorithm};  
\RCompG{} is a space-saving implementation of the r-comp algorithm. 
This version is more space efficient than \RCompS{} because it uses a grouping technique with parameter $g=16$.
Details of \RCompG{} are presented in Appendix~\ref{app:space_reduction_technique}.
We compared \RCompS{} and \RCompG{} with three state-of-the-art algorithms, one indirect construction algorithm of RLBWT (Big-BWT) and two direct construction algorithms of RLBWT (the PP and Faster-PP methods), which were reviewed in Section~\ref{sec:related_works} and summarized in Table~\ref{table:result}.
The implementations of those methods were downloaded from \url{https://gitlab.com/manzai/Big-BWT}, \url{https://github.com/xxsds/DYNAMIC}, and \url{https://github.com/itomomoti/OnlineRlbwt}, respectively.

We conducted all experiments on one core of an 48-core Intel Xeon Gold 6126 CPU at 2.6 GHz in a machine with 2 TB of RAM running the 64-bit version of CentOS 7.9.
All methods were written in C++ and compiled using \texttt{g++} version 7.3.0 with optimization flags \texttt{-O3}, \texttt{-DNDEBUG}, and \texttt{-march=naive}.
To determine working space, we measured the peak number of bytes allocated by the standard memory allocation functions \texttt{new} and \texttt{malloc}.
The source code of the r-comp algorithm is available at \url{https://github.com/kampersanda/rcomp}.

\subparagraph{Results.}
\begin{table*}[tb]
\footnotesize
\centering
\caption{
Experimental results for construction time in seconds and working space in MiB.
Each value in parentheses shows the ratio of the measured value (i.e., time or space) obtained by each method divided by the measured value obtained by \RCompG{}.
Because the PP method did not finish the constructions for \Dataset{enwiki} and \Dataset{chr19.1000} within 24 hours, these processes were terminated.
Similarly, the Faster-PP method did not finish the construction for \Dataset{chr19.1000} within 24 hours, and the process was terminated.
}
\label{tab:compare_result}

\scalebox{0.85}{
\begin{tabular}{l||cc|cc|cc|cc|cc}
\hline
String & \multicolumn{2}{c|}{\RCompG{}} & \multicolumn{2}{c|}{\RCompS{}} & \multicolumn{2}{c|}{PP} & \multicolumn{2}{c|}{Faster-PP} & \multicolumn{2}{c}{Big-BWT} \\
 & sec. & MiB & sec. & MiB & sec. & MiB & sec. & MiB & sec. & MiB \\
\hline\hline
\Dataset{cere} & 366 & 480 & 271 & 1,824 & 6,786 & 43.7 & 764 & 183 & 65.7 & 797 \\
 &  &  & \TimesTBL{0.74} & \TimesTBL{3.80} & \TimesTBL{18.6} & \TimesTBL{0.09} & \TimesTBL{2.09} & \TimesTBL{0.38} & \TimesTBL{0.18} & \TimesTBL{1.66} \\
\Dataset{coreutils} & 75.9 & 197 & 74.7 & 833 & 4,013 & 22.3 & 212 & 71.7 & 27.6 & 340 \\
 &  &  & \TimesTBL{0.98} & \TimesTBL{4.23} & \TimesTBL{52.9} & \TimesTBL{0.11} & \TimesTBL{2.80} & \TimesTBL{0.36} & \TimesTBL{0.36} & \TimesTBL{1.72} \\
\Dataset{einstein.de.txt} & 15.8 & 4.38 & 10.3 & 17.5 & 1,604 & 1.46 & 69.2 & 1.50 & 6.19 & 13.4 \\
 &  &  & \TimesTBL{0.65} & \TimesTBL{3.99} & \TimesTBL{101.7} & \TimesTBL{0.33} & \TimesTBL{4.39} & \TimesTBL{0.34} & \TimesTBL{0.39} & \TimesTBL{3.06} \\
\Dataset{einstein.en.txt} & 89.0 & 12.6 & 60.9 & 50.8 & 8,418 & 2.78 & 369 & 4.30 & 27.1 & 48.4 \\
 &  &  & \TimesTBL{0.68} & \TimesTBL{4.03} & \TimesTBL{94.6} & \TimesTBL{0.22} & \TimesTBL{4.14} & \TimesTBL{0.34} & \TimesTBL{0.30} & \TimesTBL{3.84} \\
\Dataset{Escherichia\_Coli} & 129 & 609 & 106 & 2,362 & 1,806 & 39.5 & 209 & 239 & 33.1 & 458 \\
 &  &  & \TimesTBL{0.83} & \TimesTBL{3.88} & \TimesTBL{14.0} & \TimesTBL{0.06} & \TimesTBL{1.62} & \TimesTBL{0.39} & \TimesTBL{0.26} & \TimesTBL{0.75} \\
\Dataset{influenza} & 58.3 & 128 & 55.1 & 562 & 2,255 & 13.0 & 161 & 45.5 & 29.6 & 430 \\
 &  &  & \TimesTBL{0.94} & \TimesTBL{4.38} & \TimesTBL{38.7} & \TimesTBL{0.10} & \TimesTBL{2.75} & \TimesTBL{0.35} & \TimesTBL{0.51} & \TimesTBL{3.36} \\
\Dataset{kernel} & 92.4 & 117 & 90.6 & 484 & 4,817 & 13.4 & 260 & 42.4 & 20.3 & 139 \\
 &  &  & \TimesTBL{0.98} & \TimesTBL{4.15} & \TimesTBL{52.1} & \TimesTBL{0.12} & \TimesTBL{2.82} & \TimesTBL{0.36} & \TimesTBL{0.22} & \TimesTBL{1.19} \\
\Dataset{para} & 326 & 633 & 274 & 2,459 & 6,607 & 46.7 & 734 & 250 & 64.7 & 750 \\
 &  &  & \TimesTBL{0.84} & \TimesTBL{3.89} & \TimesTBL{20.2} & \TimesTBL{0.07} & \TimesTBL{2.25} & \TimesTBL{0.40} & \TimesTBL{0.20} & \TimesTBL{1.18} \\
\Dataset{world\_leaders} & 8.86 & 24.2 & 9.72 & 102 & 750 & 3.36 & 36.9 & 8.34 & 5.88 & 93.4 \\
 &  &  & \TimesTBL{1.10} & \TimesTBL{4.20} & \TimesTBL{84.6} & \TimesTBL{0.14} & \TimesTBL{4.16} & \TimesTBL{0.34} & \TimesTBL{0.66} & \TimesTBL{3.86} \\
\hline
\Dataset{boost} & 152 & 2.85 & 74.6 & 11.5 & 18,028 & 1.19 & 769 & 1.00 & 55.2 & 73.0 \\
 &  &  & \TimesTBL{0.49} & \TimesTBL{4.02} & \TimesTBL{118.9} & \TimesTBL{0.42} & \TimesTBL{5.07} & \TimesTBL{0.35} & \TimesTBL{0.36} & \TimesTBL{25.63} \\
\Dataset{samtools} & 209 & 27.3 & 230 & 112 & 18,985 & 4.46 & 888 & 9.50 & 60.3 & 90.3 \\
 &  &  & \TimesTBL{1.10} & \TimesTBL{4.10} & \TimesTBL{90.8} & \TimesTBL{0.16} & \TimesTBL{4.25} & \TimesTBL{0.35} & \TimesTBL{0.29} & \TimesTBL{3.30} \\
\Dataset{sdsl} & 233 & 25.9 & 204 & 106 & 18,826 & 4.31 & 891 & 9.06 & 60.0 & 85.6 \\
 &  &  & \TimesTBL{0.88} & \TimesTBL{4.09} & \TimesTBL{80.7} & \TimesTBL{0.17} & \TimesTBL{3.82} & \TimesTBL{0.35} & \TimesTBL{0.26} & \TimesTBL{3.30} \\
\hline
\Dataset{enwiki} & 24,550 & 3,042 & 16,624 & 12,345 & n/a & n/a & 55,367 & 1,149 & 2,375 & 7,022 \\
 &  &  & \TimesTBL{0.68} & \TimesTBL{4.06} & & & \TimesTBL{2.26} & \TimesTBL{0.38} & \TimesTBL{0.10} & \TimesTBL{2.31} \\
\hline
\Dataset{chr19.1000} & 39,138 & 1,931 & 29,727 & 7,577 & n/a & n/a & n/a & n/a & 3,535 & 4,911 \\
 &  &  & \TimesTBL{0.76} & \TimesTBL{3.92} & & & & & \TimesTBL{0.09} & \TimesTBL{2.54} \\
\hline
\end{tabular}
}
\end{table*}

Table \ref{tab:compare_result} shows the experimental results for each method with respect to construction time and working space.
A comparison of the r-comp variants (\RCompS{} and \RCompG{}) shows that the working space of \RCompG{} was \Times{3.8--4.4} smaller than that of \RCompS{}, whereas the construction time of \RCompG{} was at most only \Times{2.0} slower and at most \Times{2.0} faster on \Dataset{world\_leaders} and \Dataset{samtools}.
These results show \RCompG{} has a high compression performance when compared with \RCompS{}.
Comparisons of the experimental results of \RCompG{} and the other methods are presented below. 

\RCompG{} was the fastest in the comparison to the direct RLBWT constructions of the PP and Faster-PP methods.
Especially for strings with a large ratio $n/r$, \RCompG{} was \Times{81--118} faster than the PP method and \Times{3.8--5.1} faster than the Faster-PP method;
the working space of \RCompG{} was \Times{2.4--6.1} larger than that that of the PP method and \Times{2.8--2.9} larger than that of the Faster-PP method, which shows that 
the working space of \RCompG{} is reasonable considering its construction time.
For the large dataset \Dataset{enwiki}, \RCompG{} finished the construction in 6.8 hours, whereas the Faster-PP method took 15.4 hours. 
In addition, the PP method did not finish within 24 hours. 
For the 1000 human genomes \Dataset{chr19.1000}, \RCompG{} finished the construction in 11 hours, whereas the PP and Faster-PP methods did not finish within 24 hours.

In the comparison between r-comp and Big-BWT, r-comp was more space efficient than Big-BWT on most strings.
Because $r$ was much smaller than $\PFP$, the results are consistent with the theoretical bound $O(r\log{n})$ of the working space of the r-comp algorithm (Theorem~\ref{thm:rcomp_working_space}).
On strings with large values of $\PFP / r$, whereas \RCompG{} was slower than Big-BWT, the difference in the construction times between \RCompG{} and Big-BWT were reasonable if one considers the space efficiency of \RCompG{}. 
For example, for \Dataset{boost}, r-comp was \Times{26} more space efficient and only \Times{2.7} slower;
for \Dataset{world\_leaders}, r-comp was \Times{3.9} more space efficient and only \Times{1.5} slower.

Overall, \RCompG{} was the fastest RLBWT construction in $O(r\log{n})$ bits of space.
Although Big-BWT was faster than \RCompG{}, it was not space efficient for 
several strings (e.g., \Dataset{boost}).
In practice, \RCompG{} achieved a better tradeoff between construction time and working space. 

\section{Conclusion}
We presented r-comp, the first optimal-time construction algorithm of RLBWT in $O(n)$ time with $O(r \log n)$ bits of working space for highly repetitive strings with $r = O(n / \log{n})$. 
Experimental results using benchmark and real-world datasets of highly repetitive strings demonstrated the superior performance of the r-comp algorithm. 

The idea behind the DBWT presented in this paper has a wide variety of applications, and 
it is applicable to the construction of various types of data structures. 
Therefore, a future task is to develop optimal-time constructions of various data structures for fast queries. 

\clearpage
\bibliographystyle{plainurl}
\bibliography{ref}

\clearpage
\appendix
\section{Proof of \texorpdfstring{$\PFP/r = \Omega(\sqrt{n})$}{PFP/r = Omega(sqrt(n))}}\label{app:bigbwt}
Prefix-free parsing creates (i) a set of $k$ strings $s_{1}, s_{2}, \ldots, s_{k}$ and 
(ii) a sequence of $k'$ integers $p(1), p(2), \ldots, p(k')$  for a given string $T$ of length $n$. 
The concatenation of $s_{p(1)}$, $s_{p(2)}$, $\ldots$, $s_{p(k')}$ is equal to $T$~(i.e., $s_{p(1)} s_{p(2)} \cdots s_{p(k')} = T$). 
Here, $\PFP$ is defined as $k' + \sum_{i = 1}^{k} |s_{i}|$, where $|s_{i}|$ is the length of $s_{i}$. 
See \cite{DBLP:journals/almob/BoucherGKLMM19} for a more detailed definition of prefix-free parsing.

Consider a string $T = a^{n}$ of length $n = m^{2}$ for a positive integer $m \geq 1$, 
where $a^{n}$ is the repetition of character $a$ with length $n$. 
The value of $\PFP$ is smallest if $s_{1} = a^{m}$, $k = 1$, and $P = 1, 1, \ldots, 1$ with $k' = m$. 
Hence, $\PFP \geq 2 \sqrt{n}$ for the string $T$. 
By contrast, the BWT of $T$ is $a^{n}$, and thus 
the number $r$ of BWT-runs in the BWT is $1$. 
Hence, $\PFP/r = \Omega(\sqrt{n})$ in the worst case.

\oldsentence{
\section{Supplementary examples and figures for Section~\ref{sec:preliminary}}\label{app:preliminary:examples}
\section{Details for DBWT and LF-interval graph}
\subsection{Supplementary examples and figures}\label{app:graph:examples}
\subsection{Proof of Lemma~\ref{lem:B_tree_order}}\label{app:B_tree_order}
}



\section{Details for Section~\ref{sec:update_graph}}
\oldsentence{
\subsection{Supplementary examples and figures}\label{app:basic_update:examples}
\subsection{Proof of Lemma~\ref{lem:split_node_formula}}\label{app:split_node_formula}
\subsection{Proof of Lemma~\ref{lem:time_update_edge_step}}\label{app:time_update_edge_step}
\subsection{Proof of Lemma~\ref{lem:time_update_data_structures}}\label{app:time_update_data_structures}
}


\subsection{Proof of Lemma~\ref{lem:computingedge}}\label{app:computingedge}
New directed edges in two sets $E_{L}$ and $E_{F}$ are determined by the DBWT-repetitions in DBWT $D^{2\alpha+1}_{\delta+1}$ and their F-intervals 
for the LF-interval graph $\graphds(D^{2\alpha+1}_{\delta+1})$ outputted by the update operation of LF-interval graph $\graphds(D^{\alpha}_{\delta})$.
To explain how to compute new edges, we specify the DBWT-repetitions in DBWT $D^{2\alpha+1}_{\delta+1}$ and their F-intervals. 


\subparagraph{DBWT-repetitions in DBWT $D^{2\alpha+1}_{\delta+1}$.}
We represent the DBWT-repetitions in DBWT $D^{2\alpha+1}_{\delta+1}$ and their F-intervals using 
a function $\shfn$. 
Function $\shfn(x) \in \{ 0, 1 \}$ returns $0$ for a given integer $x \in \{ 1, 2, \ldots, \delta \}$ if $x$ is smaller than the position $\inspos$ of special character $\$$ in BWT $L_{\delta+1}$; 
otherwise, it returns $1$. 
By the extension of BWT in Section~\ref{lab:ExBWT}, 
the $x$-th character of $L_{\delta}$ shifts by $\shfn(x)$~(i.e., $L_{\delta}[x]$ is moved to $L_{\delta+1}[x + \shfn(x)]$) 
for all $x \in \{ 1, 2, \ldots, \delta \} \setminus \{ \reppos \}$, where $\reppos$ is the position of special character $\$$ in BWT $L_{\delta}$. 
By contrast, 
the update operation (i) replaces the node $u_{i} \in U$ representing special character $\$$ with new node $u_{i^{\prime}}$, 
(ii) splits node $u_{j} \in U$ such that $p_{j} < \inspos < p_{j+1}$ into two new nodes $u_{j^{\prime}}$ and $u_{j^{\prime}+1}$, 
and (iii) inserts a new node $u_{x^{\prime}}$ representing special character $\$$ in the doubly linked list of set $U$ at the position next to a node. 
By the update operation, 
all the characters in each DBWT-repetition $L_{\delta}[p_{x}..p_{x+1}-1]$ in $D^{\alpha}_{\delta}$ 
shift in the same direction~(i.e., $\shfn(p_{x}) = \shfn(p_{x}+1) = \cdots = \shfn(p_{x+1}-1)$) except for the $i$-th and $j$-th DBWT-repetitions. 
The following lemma ensures that each node in $U$ represents a substring of BWT $L_{\delta+1}$. 

\begin{lemma}\label{lem:shiftU}
After the insert-node step of the update operation has been executed, the following four statements hold:  
(i) $u_{x}$ represents substring $L_{\delta+1}[p_{x} + \shfn(p_{x})..p_{x+1}-1 + \shfn(p_{x})]$ for all $x \in \{ 1, 2, \ldots, k \} \setminus \{ i, j \}$;  
(ii) $u_{i^{\prime}}$ represents substring $L_{\delta+1}[p_{i} + \shfn(p_{i})]$; 
(iii) $u_{x^{\prime}}$ represents substring $L_{\delta+1}[\inspos]$; 
(iv) $u_{j^{\prime}}$ and $u_{j^{\prime}+1}$ represent 
substrings $L_{\delta+1}[p_{j}..\inspos-1]$ and $L_{\delta+1}[\inspos+1..p_{j+1}]$, respectively. 
\end{lemma}
\begin{proof}
Lemma~\ref{lem:shiftU} follows from the replace-node, split-node, and insert-node steps of the update operation. 
\end{proof}

Similarly, 
the $x$-th character of $F_{\delta}$ shifts by $\shfn(x)$~(i.e., $F_{\delta}[x]$ is moved to $F_{\delta+1}[x + \shfn(x)]$) 
for all $x \in \{ 1, 2, \ldots, \delta \}$. 
By contrast, 
the update operation 
(i) replaces the node $v_{i} \in V$ representing special character $\$$ with new node $v_{x^{\prime}}$ representing the same character, 
(ii) splits the node $v_{j} \in V$ such that $p_{j} < \inspos < p_{j+1}$ into two new nodes $v_{j^{\prime}}$ and $v_{j^{\prime}+1}$, 
and (iii) inserts a new node $v_{i^{\prime}}$ representing special character $\$$ in the doubly linked list of set $V$ at the position next to a node. 
By the update operation, 
all the characters in each F-interval $[\LF_{\delta}(p_{x})..\LF_{\delta}(p_{x+1}-1)]$ for $D^{\alpha}_{\delta}$ 
shift in the same direction~(i.e., $\shfn(\LF_{\delta}(p_{x})) = \shfn(\LF_{\delta}(p_{x})+1) = \cdots = \shfn(\LF_{\delta}(p_{x+1}-1))$). 
The following lemma ensures that each node in $V$ represents an interval on $F_{\delta+1}$. 
\begin{lemma}\label{lem:shiftV}
After the insert-node step of the update operation has been executed, the following four statements hold: 
(i) $v_{x}$ represents interval $[\LF_{\delta}(p_{x}) + \shfn(\LF_{\delta}(p_{x})), \LF_{\delta}(p_{x+1}-1) + \shfn(\LF_{\delta}(p_{x}))]$ on $F_{\delta+1}$ for all $x \in \{ 1, 2, \ldots, k \} \setminus \{ i, j \}$; 
(ii) $v_{i^{\prime}}$ represents interval $[\inspos, \inspos]$ on $F_{\delta+1}$; 
(iii) $v_{x^{\prime}}$ represents interval $[1, 1]$ on $F_{\delta+1}$; 
(iv) $v_{j^{\prime}}$ and $v_{j^{\prime}+1}$ represent 
intervals $[\LF_{\delta}(p_{j}) + \shfn(\LF_{\delta}(p_{j})), \LF_{\delta}(p_{j}) + \shfn(\LF_{\delta}(p_{j})) + (\inspos - p_{j} - 1)]$ 
and $[\LF_{\delta}(p_{j}) + \shfn(\LF_{\delta}(p_{j})) + (\inspos - p_{j}), \LF_{\delta}(p_{j+1}-1) + \shfn(\LF_{\delta}(p_{j}))]$ on $F_{\delta+1}$, respectively. 
\end{lemma}
\begin{proof}
Lemma~\ref{lem:shiftV} follows from the replace-node, split-node, and insert-node steps of the update operation. 
\end{proof}

The following lemma ensures that each node $v_{x}$ in $V$ represents the F-interval of the DBWT-repetition represented by $u_{x}$ on $F_{\delta+1}$. 
\begin{lemma}\label{lem:new_DBWT_repetition_and_F-interval}
After the insert-node step of the update operation has been executed, 
the following four statements hold: 
\begin{enumerate}
    \item $v_{x}$ represents the F-interval of the DBWT-repetition of $u_{x}$ for all $x \in \{ 1, 2, \ldots, k \} \setminus \{ i, j \}$, i.e., 
    $[\LF_{\delta}(p_{x}) + \shfn(\LF_{\delta}(p_{x})), \LF_{\delta}(p_{x+1}-1) + \shfn(\LF_{\delta}(p_{x}))] = [\LF_{\delta+1}(p_{x} + \shfn(p_{x})), \LF_{\delta+1}(p_{x+1}-1 + \shfn(p_{x})) ]$;
    \item $v_{i^{\prime}}$ represents the F-interval of the DBWT-repetition of $u_{i^{\prime}}$;
    \item $v_{x^{\prime}}$ represents the F-interval of the DBWT-repetition of $u_{x^{\prime}}$;
    \item $v_{j^{\prime}}$ and $v_{j^{\prime}+1}$ represent the F-intervals of the DBWT-repetitions of $u_{j^{\prime}}$ and $u_{j^{\prime}+1}$, respectively.
\end{enumerate}
\end{lemma}
\begin{proof}
Note that the DBWT-repetition of $u_{x}$ and the interval of $v_{x}$ are $L_{\delta+1}[p_{x} + \shfn(p_{x})..p_{x+1}-1 + \shfn(p_{x})]$ and $[\LF_{\delta}(p_{x}) + \shfn(\LF_{\delta}(p_{x})),  \LF_{\delta}(p_{x+1}-1) + \shfn(\LF_{\delta}(p_{x}))]$, respectively, by Lemma~\ref{lem:shiftU}-(i) and Lemma~\ref{lem:shiftV}-(i). 
Lemma~\ref{lem:new_DBWT_repetition_and_F-interval}-(1) holds if the following two conditions hold: 
(i) $\LF_{\delta+1}(p_{x} + \shfn(p_{x})) = \LF_{\delta}(p_{x}) + \shfn(\LF_{\delta}(p_{x}))$ 
and (ii) $\LF_{\delta+1}(p_{x+1}-1 + \shfn(p_{x})]) =  \LF_{\delta}(p_{x+1}-1) + \shfn(\LF_{\delta}(p_{x}))$. 

For the input character $c$ of the update operation, 
if (a) $c$ is smaller than $L_{\delta+1}[p_{x}]$ or (b) $c = L_{\delta+1}[p_{x}]$ and $p_{i} < p_{x}$~(i.e., $\reppos < p_{x}$), 
then $\LF_{\delta+1}(p_{x} + \shfn(p_{x})) = \LF_{\delta}(p_{x}) + 1$ and $\inspos \leq \LF_{\delta}(p_{x})$ by the LF formula. 
Moreover, $\LF_{\delta}(p_{x}) + \shfn(\LF_{\delta}(p_{x})) = \LF_{\delta}(p_{x}) + 1$ by $\inspos \leq \LF_{\delta}(p_{x})$. 
Otherwise~(i.e., (a) $c$ is larger than $L_{\delta+1}[p_{x}]$ or (b) $c = L_{\delta+1}[p_{x}]$ and $p_{i} > p_{x}$), 
$\LF_{\delta+1}(p_{x} + \shfn(p_{x})) = \LF_{\delta}(p_{x})$ and $\inspos > \LF_{\delta}(p_{x})$ by the LF formula. 
Furthermore, $\LF_{\delta}(p_{x}) + \shfn(\LF_{\delta}(p_{x})) = \LF_{\delta}(p_{x})$ by $\inspos > \LF_{\delta}(p_{x})$. 
Hence, Condition (i) holds. Similarly, Condition (ii) holds using the LF formula. 
Therefore, Lemma~\ref{lem:new_DBWT_repetition_and_F-interval}-(1) holds. 
Similarly, we can prove Lemma~\ref{lem:new_DBWT_repetition_and_F-interval}-(2-4). 
\end{proof}

\subparagraph{Computing new edges in set $E_{L}$.}
\begin{figure}[t]
 \begin{center}
		\includegraphics[scale=0.6]{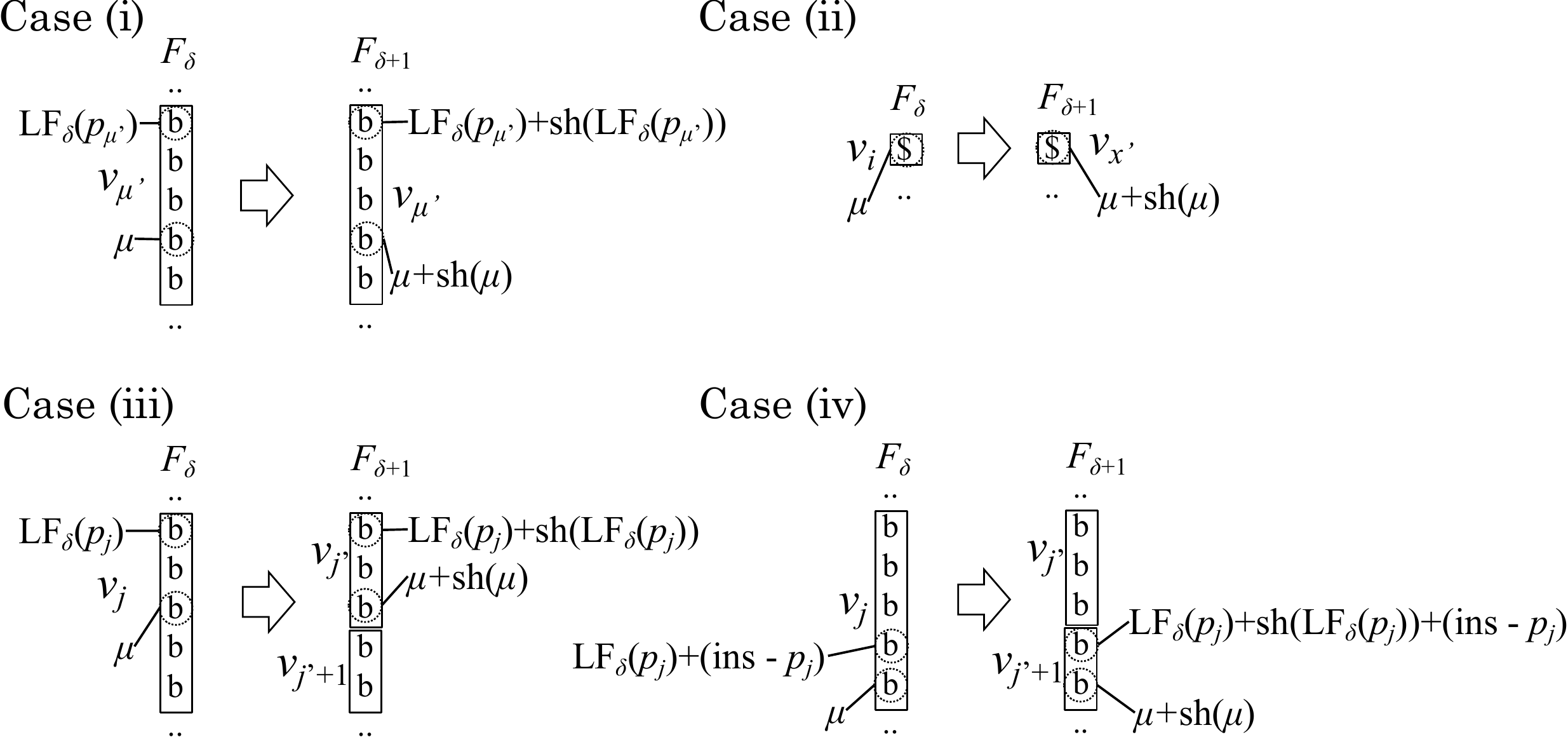}
	  \caption{
	  Four cases in the proof of Lemma~\ref{lem:v_e_computation_lemma_U}.
	  }
 \label{fig:myu_node}
 \end{center}
\end{figure}

We show that for a given node $u_{h} \in U$, 
the new directed edge $(u_{h}, v_{h^{\prime}}) \in E_{L}$ starting at the given node $u_{h}$ and the label $e$ of this edge can be computed in $O(1)$ time. 
Here, $v_{h^{\prime}} \in V$ is the node connected to the head of the new directed edge. 
The following lemma can be used to compute node $v_{h^{\prime}}$ and label $e$.

\begin{lemma}\label{lem:v_e_computation_lemma_U}
Let $\mu \in \{ 1, 2, \ldots, \delta \}$ be the integer such that 
the DBWT-repetition of $u_{h}$ starts at position $\mu + \shfn(\mu)$ on $L_{\delta+1}$~(i.e., $u_{h}$ represents a substring of $L_{\delta+1}$ starting at position $\mu + \shfn(\mu)$), 
and let $v_{\mu^{\prime}}$ be the node in $\{ v_{1}, v_{2}, \ldots, v_{k} \}$ such that 
the F-interval of $v_{\mu^{\prime}}$ contains position $\mu$ on $F_{\delta}$. 
Then, we can compute $v_{h^{\prime}}$ and $e$ in $O(1)$ time if we know $v_{\mu^{\prime}}$ and integer $\mu - \LF_{\delta}(p_{\mu^{\prime}})$. 
\end{lemma}
\begin{proof}
First, we show that $\shfn(\mu) = \shfn(\LF_{\delta}(p_{\mu^{\prime}}))$ holds for the starting position $\LF_{\delta}(p_{\mu^{\prime}})$ of F-interval $[\LF_{\delta}(p_{\mu^{\prime}}), \LF_{\delta}(p_{\mu^{\prime}+1}-1)]$,  represented as $v_{\mu^{\prime}}$ on $F_{\delta}$. 
The node $v_{i^{\prime}}$ representing the input character $c$ is inserted into the doubly linked list of $V$ at the position next 
to a node in the replace-node step of the update operation, 
and the F-interval of $v_{i^{\prime}}$ is $[\inspos, \inspos]$ on $F_{\delta+1}$. 
If $v_{i^{\prime}}$ precedes $v_{\mu^{\prime}}$ in the doubly linked list of $V$, 
then $\inspos \leq \LF_{\delta}(p_{\mu^{\prime}}) \leq \mu$ by $\mu \in [\LF_{\delta}(p_{\mu^{\prime}}), \LF_{\delta}(p_{\mu^{\prime}+1}-1)]$. 
Otherwise, $\LF_{\delta}(p_{\mu^{\prime}}) \leq \mu < \inspos$. 
Hence, $\shfn(\mu) = \shfn(\LF_{\delta}(p_{\mu^{\prime}}))$.

Second, we prove Lemma~\ref{lem:v_e_computation_lemma_U}. We consider the following four cases:  
(i) $v_{\mu^{\prime}} \not \in \{ v_{i}, v_{j} \}$; (ii) $v_{\mu^{\prime}} = v_{i}$; (iii) $v_{\mu^{\prime}} = v_{j}$ and $\mu - \LF_{\delta}(p_{\mu^{\prime}}) < \inspos - p_{j}$; and (iv) $v_{\mu^{\prime}} = v_{j}$ and $\mu - \LF_{\delta}(p_{\mu^{\prime}}) \geq \inspos - p_{j}$. 
For case (i), 
$v_{\mu^{\prime}}$ represents F-interval $[\LF_{\delta}(p_{\mu^{\prime}}), \LF_{\delta}(p_{\mu^{\prime}+1}-1)]$ on $F_{\delta}$, and the position $\mu$ on $F_{\delta}$ is contained in the F-interval. 
By Lemma~\ref{lem:shiftV}-(i), 
the F-interval of $v_{\mu^{\prime}}$ on $F_{\delta+1}$ is $[\LF_{\delta}(p_{\mu^{\prime}}) + \shfn(\LF_{\delta}(p_{\mu^{\prime}})), \LF_{\delta}(p_{\mu^{\prime}+1}-1) + \shfn(\LF_{\delta}(p_{\mu^{\prime}}))]$. 
This F-interval contains the position $\mu + \shfn(\mu)$ on $F_{\delta+1}$ by $\mu \in [\LF_{\delta}(p_{\mu^{\prime}}), \LF_{\delta}(p_{\mu^{\prime}+1}-1)]$ 
and $\shfn(\mu) = \shfn(\LF_{\delta}(p_{\mu^{\prime}}))$~(see also Figure~\ref{fig:myu_node} for case (i)).
$v_{h^{\prime}} = v_{\mu^{\prime}}$ because the DBWT-repetition of $u_{h}$ starts at position $\mu + \shfn(\mu)$ on $L_{\delta+1}$. 
In addition, $e = (\mu + \shfn(\mu)) - (\LF_{\delta}(p_{\mu^{\prime}}) + \shfn(\LF_{\delta}(p_{\mu^{\prime}})))$, 
and thus $e = \mu - \LF_{\delta}(p_{\mu^{\prime}})$ by $\shfn(\mu) = \shfn(\LF_{\delta}(p_{\mu^{\prime}}))$. 

For case (ii), $\mu = 1$ because the F-interval of $v_{i}$ is $[1, 1]$~(see also Figure~\ref{fig:myu_node} for case (ii)). 
$\inspos > 1$ because $c \neq \$$. 
Moreover, $\mu + \shfn(\mu) = 1$ by $\inspos > 1$ and $\mu = 1$.  
The F-interval of $v_{x^{\prime}}$ is $[1, 1]$ on $F_{\delta+1}$. 
Hence, $v_{h^{\prime}} = v_{x^{\prime}}$ and $e = 0$.

For case (iii), 
the F-interval of $v_{j^{\prime}}$ is $[\LF_{\delta}(p_{j}) + \shfn(\LF_{\delta}(p_{j})), \LF_{\delta}(p_{j}) + \shfn(\LF_{\delta}(p_{j})) + (\inspos - p_{j})-1]$ on $F_{\delta+1}$ by Lemma~\ref{lem:shiftV}-(iv)~(see also Figure~\ref{fig:myu_node} for case (iii)). 
The F-interval of $v_{j^{\prime}}$ contains position $\mu + \shfn(\mu)$ because 
$\mu \in [\LF_{\delta}(p_{j}), \LF_{\delta}(p_{j}) + (\inspos - p_{j})-1]$ and $\shfn(\mu) = \shfn(\LF_{\delta}(p_{j}))$. 
Hence, $v_{h^{\prime}} = v_{j^{\prime}}$ and $e = (\mu - \LF_{\delta}(p_{j})) + (\shfn(\mu) - \shfn(\LF_{\delta}(p_{j})))$. 
We have $e = \mu - \LF_{\delta}(p_{j})$ by $\shfn(\mu) = \shfn(\LF_{\delta}(p_{j}))$.

Similarly, for case (iv), 
the F-interval of $v_{j^{\prime}+1}$ is $[\LF_{\delta}(p_{j}) + \shfn(\LF_{\delta}(p_{j})) + (\inspos - p_{j}), \LF_{\delta}(p_{j+1}-1) + \shfn(\LF_{\delta}(p_{j}))]$ on $F_{\delta+1}$~(see also Figure~\ref{fig:myu_node} for case (iv)). 
The F-interval of $v_{j^{\prime}+1}$ contains position $\mu + \shfn(\mu)$ 
because $\mu \in [\LF_{\delta}(p_{j}) + (\inspos - p_{j}), \LF_{\delta}(p_{j+1}-1)]$ 
and $\shfn(\mu) = \shfn(\LF_{\delta}(p_{j}))$. 
We have $v_{h^{\prime}} = v_{j^{\prime}+1}$ and $e = (\mu - \LF_{\delta}(p_{j}) + (\inspos - p_{j})) + (\shfn(\mu) - \shfn(\LF_{\delta}(p_{j})))$. 
Furthermore, $e = \mu - \LF_{\delta}(p_{j}) + (\inspos - p_{j})$ by 
$\shfn(\mu) - \shfn(\LF_{\delta}(p_{j})) = 0$. 
For (i) node $v_{g} \in V$ searched for in the replace-node step 
and (ii) node $v_{\mathsf{gnext}} \in V$ next to $v_{g}$ in the doubly linked list before executing the replace-node step, 
$v_{\mathsf{gnext}}$ is connected to $u_{j}$ by the directed edge $E_{F}(v_{\mathsf{gnext}}, u_{j})$, 
and we know that $\inspos - p_{j} = B_{F}(v_{\mathsf{gnext}}, u_{j})$~(see the proof of Lemma~\ref{lem:edge_label}). 
The label $B_{F}(v_{\mathsf{gnext}}, u_{j})$ of the directed edge is obtained by the split-node step of the update operation. 
Hence, we can compute $v_{h^{\prime}}$ and $e$ in $O(1)$ time. 
\end{proof}

We show that $v_{h^{\prime}}$ and label $e$ in $O(1)$ time using Lemma~\ref{lem:v_e_computation_lemma_U}. 
We consider five cases for $u_{h}$: 
(i) $u_{h} = u_{i'}$; 
(ii) $u_{h} = u_{j^{\prime}}$; 
(iii) $u_{h} = u_{j^{\prime}+1}$; 
(iv) $u_{h} = u_{x^{\prime}}$; and 
(v) $u_{h}$ is not a new node~(i.e., $u_{h} \in \{ u_{1}, u_{2}, \ldots, u_{k} \} \setminus \{ u_{i}, u_{j} \}$).

\subparagraph{Case (i).}
The DBWT-repetition of $u_{i^{\prime}}$ starts at position $p_{i} + \shfn(p_{i})$ on $L_{\delta+1}$ by Lemma~\ref{lem:shiftU}-(ii). 
Before the update operation is executed, 
$u_{i}$ is connected to a node $v_{\mu^{\prime}} \in V$ by directed edge $(u_{i}, v_{\mu^{\prime}}) \in E_{L}$, 
and the F-interval of $v_{\mu^{\prime}}$ contains position $p_{i}$ on $F_{\delta}$. 
The label of the directed edge $(u_{i}, v_{\mu^{\prime}})$ is $B_{L}(u_{i}, v_{\mu^{\prime}}) = p_{i} - \LF_{\delta}(p_{\mu^{\prime}})$, 
and hence, we can compute $v_{h^{\prime}}$ and $e$ using Lemma~\ref{lem:v_e_computation_lemma_U} in $O(1)$ time.

\subparagraph{Case (ii).}
The DBWT-repetition of $u_{j^{\prime}}$ starts at position $p_{j} + \shfn(p_{j})$ on $L_{\delta+1}$ by Lemma~\ref{lem:shiftU}-(iv) and $\shfn(p_{j}) = 0$. 
Before the update operation is executed, 
$u_{j}$ is connected to a node $v_{\mu^{\prime}}$ by directed edge $(u_{j}, v_{\mu^{\prime}}) \in E_{L}$, 
and the F-interval of $v_{\mu^{\prime}}$ contains position $p_{j}$ on $F_{\delta}$. 
The label of directed edge $(u_{j}, v_{\mu^{\prime}})$ is $B_{L}(u_{j}, v_{\mu^{\prime}}) = p_{j} - \LF_{\delta}(p_{\mu^{\prime}})$, 
and hence, we can compute $v_{h^{\prime}}$ and $e$ using Lemma~\ref{lem:v_e_computation_lemma_U} in $O(1)$ time.

\subparagraph{Case (iii).}
The DBWT-repetition of $u_{j^{\prime} + 1}$ starts at position $(\inspos+1)$ on $L_{\delta+1}$ by Lemma~\ref{lem:shiftU}-(iv). 
The new node $v_{i^{\prime}}$ representing the input character $c$ is inserted into the doubly linked list of $V$ at the position previous to 
node $v_{\mathsf{gnext}} \in V$, which was next to the node $v_{g}$ searched for in the replace-node step. 
We have $v_{\mathsf{gnext}} \neq v_{i}$ by $c > \$$, 
and hence $v_{\mathsf{gnext}} \in \{ v_{1}, v_{2}, \ldots, v_{k} \} \setminus \{ v_{i}, v_{j}\}$ or 
$v_{\mathsf{gnext}} = v_{j}$. 
If $v_{\mathsf{gnext}} \in \{ v_{1}, v_{2}, \ldots, v_{k} \} \setminus \{ v_{i}, v_{j}\}$, then
the F-interval of $v_{\mathsf{gnext}}$ starts at position $(\inspos+1)$ on $F_{\delta+1}$, 
and hence $v_{h^{\prime}} = v_{\mathsf{gnext}}$ and $e = 0$. 
Otherwise, the F-interval of $v_{j^{\prime}}$ starts at position $(\inspos+1)$ on $F_{\delta+1}$, 
and hence $v_{h^{\prime}} = v_{j^{\prime}}$ and $e = 0$. 
This fact indicates that we can compute $v_{h^{\prime}}$ and $e$ in $O(1)$ time.

\subparagraph{Case (iv).}
The DBWT-repetition of $u_{x^{\prime}}$ starts at position $\inspos$ on $L_{\delta+1}$ by Lemma~\ref{lem:shiftU}-(iii), 
and the F-interval of $v_{i^{\prime}}$ starts at position $\inspos$ on $F_{\delta+1}$ by Lemma~\ref{lem:shiftV}-(ii). 
Hence, $v_{h^{\prime}} = v_{i^{\prime}}$ and $e = 0$. 
This fact indicates that we can compute $v_{h^{\prime}}$ and $e$ in $O(1)$ time.

\subparagraph{Case (v).}
The DBWT-repetition of $u_{h}$ starts at position $p_{h} + \shfn(p_{h})$ on $L_{\delta+1}$ by Lemma~\ref{lem:shiftU}-(i). 
Before the update operation is executed, 
$u_{h}$ is connected to a node $v_{\mu^{\prime}}$ by directed edge $(u_{h}, v_{\mu^{\prime}}) \in E_{L}$, 
and the F-interval of $v_{\mu^{\prime}}$ contains position $p_{h}$ on $F_{\delta}$. 
The label of the directed edge is $B_{L}(u_{h}, v_{\mu^{\prime}}) = p_{h} - \LF_{\delta}(p_{\mu^{\prime}})$, 
and hence we can compute $v_{h^{\prime}}$ and $e$ using Lemma~\ref{lem:v_e_computation_lemma_U} in $O(1)$ time.

\subparagraph{Computing new edges in set $E_{F}$.}
Next, we show that for a given node $v_{h} \in V$, 
the new directed edge $(v_{h}, u_{h^{\prime}}) \in E_{F}$ starting at the given node $v_{h}$ and the label $e$ of the new directed edge can be computed in $O(t)$ time. 
Here, (i) $u_{h^{\prime}} \in U$ is the node connected to the head of the new directed edge, 
and (ii) $t = 1$ if $v_{h} \neq v_{j^{\prime}+1}$; otherwise $t = \alpha$. 
The following lemma can be used to compute node $u_{h^{\prime}}$ and label $e$.
\begin{lemma}\label{lem:v_e_computation_lemma_V}
Let $\mu \in \{ 1, 2, \ldots, \delta \}$ be the integer such that 
the F-interval of $v_{h}$ starts at position $\mu + \shfn(\mu)$ on $F_{\delta+1}$, 
and let $u_{\mu^{\prime}}$ be the node in $\{ u_{1}, u_{2}, \ldots, u_{k} \}$ such that 
the DBWT-repetition of $u_{\mu^{\prime}}$ contains position $\mu$ on $L_{\delta}$. 
Then, we can compute $u_{h^{\prime}}$ and $e$ in $O(1)$ time if we know $u_{\mu^{\prime}}$ and integer $\mu - p_{\mu^{\prime}}$. 
\end{lemma}
\begin{proof}
Lemma~\ref{lem:v_e_computation_lemma_V} is symmetric to Lemma~\ref{lem:v_e_computation_lemma_U}. 
We can prove Lemma~\ref{lem:v_e_computation_lemma_V} by modifying the proof of Lemma~\ref{lem:v_e_computation_lemma_U}.
\end{proof}

We consider five cases for $v_{h}$:
(i) $v_{h} = v_{i^{\prime}}$; 
(ii) $v_{h} = v_{j^{\prime}}$; 
(iii) $v_{h} = v_{j^{\prime}+1}$; 
(iv) $v_{h} = v_{x^{\prime}}$; and
(v) $v_{h}$ is not a new node~(i.e., $v \in \{ v_{1}, v_{2}, \ldots, v_{k} \} \setminus \{ v_{i}, v_{j} \}$).

\subparagraph{Case (i).}
The F-interval of $v_{i^{\prime}}$ starts at position $\inspos$ on $F_{\delta+1}$ by Lemma~\ref{lem:shiftV}-(ii), 
and the DBWT-repetition of $u_{x^{\prime}}$ starts at position $\inspos$ on $L_{\delta+1}$ by Lemma~\ref{lem:shiftU}-(iii). 
Hence, $u_{h^{\prime}} = u_{x^{\prime}}$ and $e = 0$. 
This fact indicates that we can compute $u_{h^{\prime}}$ and $e$ in $O(1)$ time.

\subparagraph{Case (ii).}
The F-interval of $v_{j^{\prime}}$ starts at position $\LF_{\delta}(p_{j}) + \shfn(\LF_{\delta}(p_{j}))$ on $F_{\delta+1}$ by Lemma~\ref{lem:shiftV}-(iv). 
Before the update operation is executed, 
$v_{j}$ is connected to a node $u_{\mu^{\prime}} \in U$ by directed edge $(v_{j}, u_{\mu^{\prime}}) \in E_{F}$, 
and the DBWT-repetition of $u_{\mu^{\prime}}$ contains position $\LF_{\delta}(p_{j})$ on $L_{\delta}$. 
The label of directed edge $(v_{j}, u_{\mu^{\prime}})$ is $B_{F}(v_{j}, u_{\mu^{\prime}}) = \LF_{\delta}(p_{j}) - p_{\mu^{\prime}}$, 
and hence we can compute $u_{h^{\prime}}$ and $e$ using Lemma~\ref{lem:v_e_computation_lemma_V} in $O(1)$ time.

\subparagraph{Case (iii).}
\begin{figure}[t]
 \begin{center}
		\includegraphics[scale=0.7]{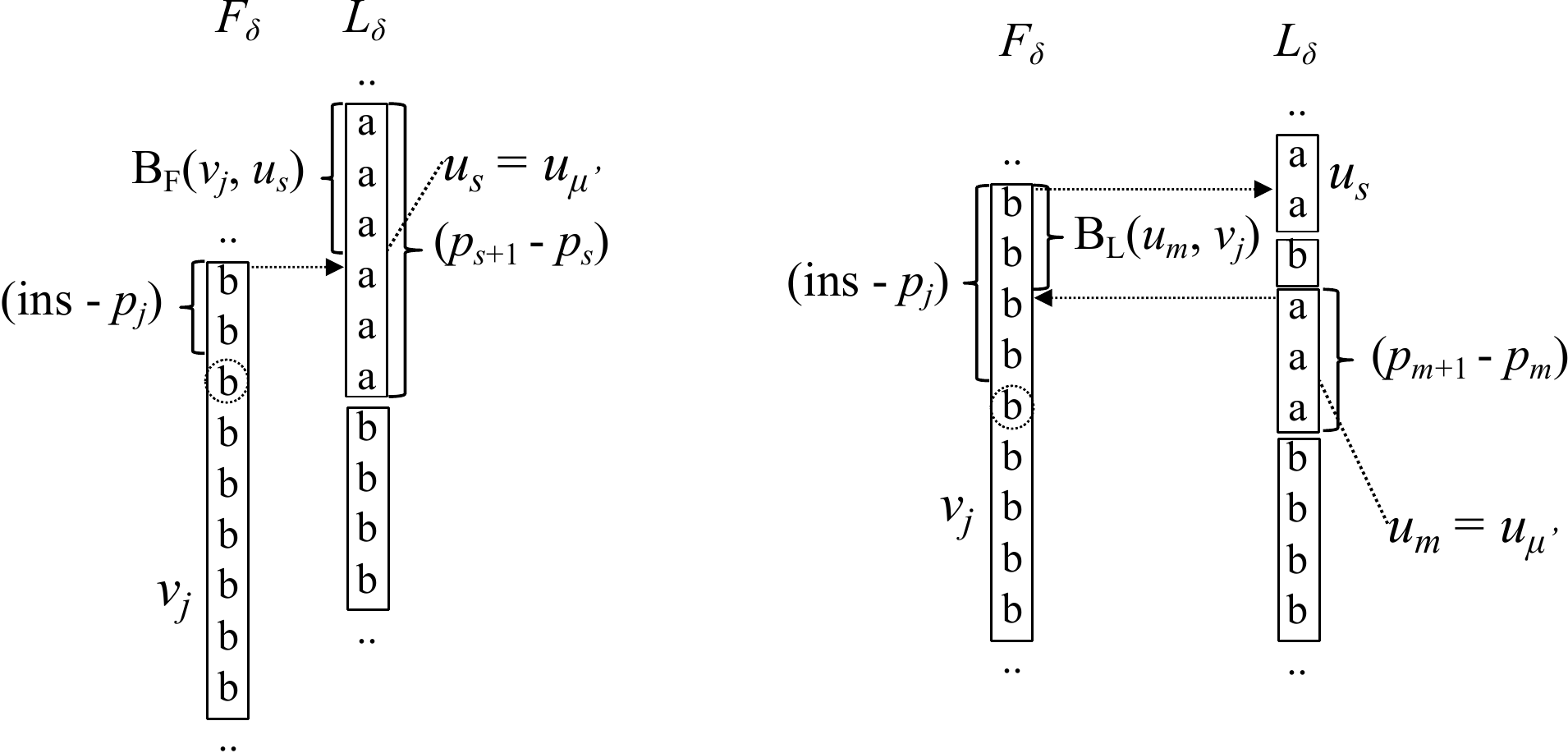}
	  \caption{
      (Left) Three nodes $v_{j}$, $u_{s}$, and $u_{\mu^{\prime}}$ for $B_{F}(v_{j}, u_{s}) + (\inspos - p_{j}) < (p_{s+1} - p_{s})$ in Lemma~\ref{lem:comp_vj}. 
      (Right) Four nodes $v_{j}$, $u_{s}$, $u_{m}$, and $u_{\mu^{\prime}}$ for $B_{F}(v_{j}, u_{s}) + (\inspos - p_{j}) \geq (p_{s+1} - p_{s})$. 
	  Dotted circles represent the starting position of the F-interval of $v_{j^{\prime}+1}$, 
	  i.e., the characters in the dotted circles are $F_{\delta}[\LF_{\delta}(p_{j}) + (\inspos - p_{j})]$. 
	  }
 \label{fig:u_w}
 \end{center}
\end{figure}

The F-interval of $v_{j^{\prime}+1}$ starts at position $\LF_{\delta}(p_{j}) + (\inspos - p_{j}) + \shfn(\LF_{\delta}(p_{j}) + (\inspos - p_{j}))$ on $F_{\delta+1}$ by Lemma~\ref{lem:shiftV}-(iv) and $\shfn(\LF_{\delta}(p_{j})) = \shfn(\LF_{\delta}(p_{j}) + (\inspos - p_{j}))$. 
Let $u_{\mu^{\prime}} \in \{ u_{1}, u_{2}, \ldots, u_{k} \}$ be the node such that the DBWT-repetition of $u_{\mu^{\prime}}$ contains position $\LF_{\delta}(p_{j}) + (\inspos - p_{j})$ on $L_{\delta}$~(i.e., $p_{\mu^{\prime}} \leq \LF_{\delta}(p_{j}) + (\inspos - p_{j}) < p_{\mu^{\prime}+1}$). 
Then, we can compute $u_{h^{\prime}}$ and $e$ using Lemma~\ref{lem:v_e_computation_lemma_V} in $O(1)$ time after computing node $u_{\mu^{\prime}}$ and integer value $((\LF_{\delta}(p_{j}) + (\inspos - p_{j})) - p_{\mu^{\prime}})$. 
The following lemma can be used to compute $u_{\mu^{\prime}}$ and integer value $((\LF_{\delta}(p_{j}) + (\inspos - p_{j})) - p_{\mu^{\prime}})$ in $O(\alpha)$ time. 

\begin{lemma}\label{lem:comp_vj}
Let $u_{s} \in U$ be the node connected to node $v_{j}$ by directed edge $(v_{j}, u_{s}) \in E_{F}$ before the update operation has been executed. Then, the following two statements hold: 
(i) 
$u_{\mu^{\prime}} = u_{s}$ if $B_{F}(v_{j}, u_{s}) + (\inspos - p_{j}) < (p_{s+1} - p_{s})$ for the label $B_{F}(v_{j}, u_{s})$ of directed edge $(v_{j}, u_{s})$. 
Otherwise, let $u_{m} \in \{ u_{1}, u_{2}, \ldots, u_{k} \}$ be the node such that 
$(u_{m}, v_{j}) \in E_{L}$ and $B_{L}(u_{m}, v_{j}) \leq \inspos - p_{j} < B_{L}(u_{m}, v_{j}) + (p_{m+1} - p_{m})$. 
Then, $u_{\mu^{\prime}} = u_{m}$. 
(ii) We can compute $u_{\mu^{\prime}}$ and integer value $((\LF_{\delta}(p_{j}) + (\inspos - p_{j})) - p_{\mu^{\prime}})$ in $O(\alpha)$ time 
using node $v_{j}$ and LF-interval graph $\graphds(D^{\alpha}_{\delta})$. 
\end{lemma}
\begin{proof}
(i) 
Figure~\ref{fig:u_w} illustrates four nodes $v_{j}$, $u_{s}$, $u_{m}$, and $u_{\mu^{\prime}}$. 
The DBWT-repetition of $u_{s}$ contains position $(\LF_{\delta}(p_{j}) + (\inspos - p_{j}))$ on $L_{\delta}$ 
if $B_{F}(v_{j}, u_{s}) + (\inspos - p_{j}) < p_{s+1} - p_{s}$~(see also the left of Figure~\ref{fig:u_w}). 
Hence, $u_{s} = u_{\mu^{\prime}}$. 
Otherwise, $u_{m}$ exists and its DBWT-repetition contains position $(\LF_{\delta}(p_{j}) + (\inspos - p_{j}))$ on $L_{\delta}$~(see also the right of Figure~\ref{fig:u_w}). 
Hence, $u_{\mu^{\prime}} = u_{m}$.

(ii) 
We show that we can get $u_{\mu^{\prime}}$ in $O(\alpha)$ time. 
Recall that $\inspos - p_{j} = B_{F}(v_{\mathsf{gnext}}, u_{j})$ by the proof of Lemma~\ref{lem:edge_label} 
for (i) node $v_{g} \in V$ searched for in the replace-node step 
and (ii) node $v_{\mathsf{gnext}} \in V$ next to $v_{g}$ in the doubly linked list before executing the replace-node step. 
$p_{s+1}-p_{s}$ is stored in the label of node $u_{s}$, 
and hence, we can determine whether $B_{F}(v_{j}, u_{s}) + (\inspos - p_{j}) < (p_{s+1} - p_{s})$ or not in $O(1)$ time. 
If $B_{F}(v_{j}, u_{s}) + (\inspos - p_{j}) < (p_{s+1} - p_{s})$, 
then $u_{s} = u_{\mu^{\prime}}$ by Lemma~\ref{lem:comp_vj}-(i). 
Otherwise, $u_{\mu^{\prime}} = u_{m}$. 
Node $u_{m}$ is a node whose directed edge points to $v_{j}$, 
and we can verify whether $u = u_{m}$ in $O(1)$ time for a given node $u \in U$. 
The number of directed edges pointing to $v_{j}$ is no more than $\alpha$ because $D^{\alpha}_{\delta}$ is $\alpha$-balanced. 
Hence, we can find $u_{\mu^{\prime}}$ in $O(\alpha)$ time. 

Next, we can compute integer value $((\LF_{\delta}(p_{j}) + (\inspos - p_{j})) - p_{\mu^{\prime}})$ in $O(1)$ time after finding $u_{\mu^{\prime}}$ 
because $p_{\mu^{\prime}} - \LF_{\delta}(p_{j})$ is stored in the label of directed edge $(u_{\mu^{\prime}}, v_{j}) \in E_{L}$. 
Therefore, Lemma~\ref{lem:comp_vj}-(ii) holds.
\end{proof}

\subparagraph{Case (iv).}
The F-interval of $v_{x^{\prime}}$ starts at position $\LF_{\delta}(p_{i}) + \shfn(\LF_{\delta}(p_{i}))$ on $F_{\delta+1}$ by Lemma~\ref{lem:shiftV}-(iii), $\LF_{\delta}(p_{i}) = 1$, and $\shfn(\LF_{\delta}(p_{i})) = 0$. 
Before the update operation is executed, 
$v_{i}$ is connected to node $u_{1} \in U$ by directed edge $(v_{i}, u_{1}) \in E_{F}$, 
and the DBWT-repetition of $u_{1}$ contains position $\LF_{\delta}(p_{i})$ on $L_{\delta}$. 
The label of directed edge $(v_{i}, u_{1})$ is $B_{F}(v_{i}, u_{1}) = \LF_{\delta}(p_{i}) - p_{1}$, 
and hence we can compute $u_{h^{\prime}}$ and $e$ using Lemma~\ref{lem:v_e_computation_lemma_V} in $O(1)$ time.

\subparagraph{Case (v).}
The F-interval of $v_{h}$ starts at position $\LF_{\delta}(p_{h}) + \shfn(\LF_{\delta}(p_{h}))$ on $F_{\delta+1}$ by Lemma~\ref{lem:shiftV}-(i). 
Before the update operation is executed, 
$v_{h}$ is connected to a node $u_{\mu^{\prime}}$ by a directed edge in $E_{F}$, 
and the DBWT-repetition of $u_{\mu^{\prime}}$ contains position $\LF_{\delta}(p_{h})$ on $L_{\delta}$. 
The label of the directed edge is $\LF_{\delta}(p_{h}) - p_{\mu^{\prime}}$, 
and hence we can compute $u_{h^{\prime}}$ and $e$ using Lemma~\ref{lem:v_e_computation_lemma_V} in $O(1)$ time.

Therefore, Lemma~\ref{lem:computingedge} holds.

\subsection{Proof of Lemma~\ref{lem:update_edge_count}}\label{app:update_edge_count}
\subparagraph{Proof of Lemma~\ref{lem:update_edge_count}-(i).}
The following lemma states the relationship between (i) the set $E_{F}$ in $\graphds(D^{\alpha}_{\delta})$ 
and (ii) the set $E_{F}$ in $\graphds(D^{2\alpha + 1}_{\delta+1})$. 
\begin{lemma}\label{lem:new_edge_goal_F}
The following five statements hold: 
(i) For a node $u_{x} \in \{ u_{1}, u_{2}, \ldots, u_{k} \} \setminus \{ u_{i}, u_{j} \}$, 
if directed edge $(v_{x^{\prime}}, u_{x})$ is contained in the $E_{F}$ stored in $\graphds(D^{2\alpha + 1}_{\delta+1})$, 
then directed edge $(v_{i}, u_{x})$ is contained in the $E_{F}$ stored in $\graphds(D^{\alpha}_{\delta})$. 
(ii) 
For a node $u_{x} \in \{ u_{1}, u_{2}, \ldots, u_{k} \} \setminus \{ u_{i}, u_{j} \}$, 
if directed edge $(v_{j^{\prime}}, u_{x})$ is contained in the $E_{F}$ stored in $\graphds(D^{2\alpha + 1}_{\delta+1})$, 
then directed edge $(v_{j}, u_{x})$ is contained in the $E_{F}$ stored in $\graphds(D^{\alpha}_{\delta})$. 
(iii) Directed edge $(v_{i^{\prime}}, u_{x^{\prime}})$ is contained in the $E_{F}$ stored in $\graphds(D^{2\alpha + 1}_{\delta+1})$. 
(iv) For a node $v_{x} \in \{ v_{1}, v_{2}, \ldots, v_{k} \} \setminus \{ v_{i}, v_{j} \}$, 
if either directed edge $(v_{x}, u_{j^{\prime}})$ or $(v_{x}, u_{j^{\prime}+1})$ is contained in the $E_{F}$ stored in $\graphds(D^{2\alpha + 1}_{\delta+1})$, 
then directed edge $(v_{x}, u_{j})$ is contained in the $E_{F}$ stored in $\graphds(D^{\alpha}_{\delta})$. 
(v) For two nodes $u_{x} \in \{ u_{1}, u_{2}, \ldots, u_{k} \} \setminus \{ u_{i}, u_{j} \}$ and $v_{y} \in \{ v_{1}, v_{2}, \ldots, v_{k} \} \setminus \{ v_{i}, v_{j} \}$, 
directed edge $(v_{y}, u_{x})$ is contained in the $E_{F}$ stored in LF-interval graph $\graphds(D^{2\alpha + 1}_{\delta+1})$ 
if and only if $(v_{y}, u_{x})$ is contained in the $E_{F}$ stored in LF-interval graph $\graphds(D^{\alpha}_{\delta})$.  
\end{lemma}
\begin{proof}
(i) The DBWT-repetition of $u_{x}$ starts at the first position on $L_{\delta+1}$ 
because the F-interval of $v_{x^{\prime}}$ is $[1, 1]$ on $F_{\delta+1}$ by Lemma~\ref{lem:shiftV}-(iii). 
The DBWT-repetition of $u_{x}$ starts at the first position on $L_{\delta}$~(i.e., $p_{x} = 1$) by Lemma~\ref{lem:shiftU}-(i) and $\shfn(1) = 0$. 
The F-interval of $v_{i}$ is $[1, 1]$ on $F_{\delta}$, 
and hence the DBWT-repetition of $u_{x}$ covers the starting position of the F-interval of $v_{i}$ on $L_{\delta}$, 
i.e., Lemma~\ref{lem:new_edge_goal_F}-(i) holds.

(ii) The DBWT-repetition of $u_{x}$ is substring $L_{\delta+1}[p_{x} + \shfn(p_{x})..p_{x+1} - 1 + \shfn(p_{x})]$ by Lemma~\ref{lem:shiftU}-(i), 
and the F-interval of $v_{j^{\prime}}$ starts at position $\LF_{\delta}(p_{j}) + \shfn(\LF_{\delta}(p_{j}))$ on $F_{\delta+1}$ by Lemma~\ref{lem:shiftV}-(iv). 
$\LF_{\delta}(p_{j}) + \shfn(\LF_{\delta}(p_{j})) \in [p_{x} + \shfn(p_{x}), p_{x+1} - 1 + \shfn(p_{x})]$ and $\shfn(\LF_{\delta}(p_{j})) = \shfn(p_{x})$ because 
the DBWT-repetition of $u_{x}$ covers the starting position of the F-interval of $v_{j^{\prime}}$. 
$\LF_{\delta}(p_{j}) \in [p_{x}, p_{x+1} - 1]$ by $\shfn(\LF_{\delta}(p_{j})) = \shfn(p_{x})$ 
and $\LF_{\delta}(p_{j}) + \shfn(\LF_{\delta}(p_{j})) \in [p_{x} + \shfn(p_{x}), p_{x+1} - 1 + \shfn(p_{x})]$. 
Before the update operation is executed, 
the DBWT-repetition of $u_{x}$ is substring $L_{\delta}[p_{x}..p_{x+1} - 1]$, 
and the F-interval of $v_{j}$ starts at position $\LF_{\delta}(p_{j})$ on $F_{\delta}$. 
The DBWT-repetition of $u_{x}$ covers the starting position of the F-interval of $v_{j}$ on $L_{\delta}$ by $\LF_{\delta}(p_{j}) \in [p_{x}, p_{x+1} - 1]$, 
and hence Lemma~\ref{lem:new_edge_goal_F}-(ii) holds. 

(iii) The DBWT-repetition of $u_{x^{\prime}}$ starts at position $\inspos$ on $L_{\delta+1}$ by Lemma~\ref{lem:shiftU}-(iii), 
and the F-interval of $u_{i^{\prime}}$ starts at position $\inspos$ on $F_{\delta+1}$ by Lemma~\ref{lem:shiftV}-(ii). 
Hence, the DBWT-repetition of $u_{x^{\prime}}$ covers the starting position $\inspos$ of the F-interval of $v_{i^{\prime}}$ on $L_{\delta}$, 
i.e., Lemma~\ref{lem:new_edge_goal_F}-(iii) holds. 

(iv) The DBWT-repetitions of $u_{j^{\prime}}$ and $u_{j^{\prime}+1}$ are substrings 
$L_{\delta+1}[p_{j}..\inspos-1]$ and $L_{\delta+1}[\inspos+1..p_{j+1}]$, respectively, by Lemma~\ref{lem:shiftU}-(iv). 
The F-interval of $v_{x}$ starts at position $\LF_{\delta}(p_{x}) + \shfn(\LF_{\delta}(p_{x}))$ on $F_{\delta+1}$ by Lemma~\ref{lem:shiftV}-(i). 
If the DBWT-repetitions of $u_{j^{\prime}}$ covers the starting position of the F-interval of $v_{x}$, 
then $\LF_{\delta}(p_{x}) < \inspos$ and $\LF_{\delta}(p_{x}) + \shfn(\LF_{\delta}(p_{x})) \in [p_{j}, \inspos-1]$. 
In addition, $\shfn(\LF_{\delta}(p_{x})) = 0$ by $\LF_{\delta}(p_{x}) < \inspos$, 
and hence $\LF_{\delta}(p_{x}) \in [p_{j}, \inspos-1]$. 
The DBWT-repetition of $u_{j}$ is $L_{\delta}[p_{j}..p_{j+1}-1]$, 
and $[p_{j}, p_{j+1}-1] \supseteq [p_{j}, \inspos-1]$ by $p_{j} < \inspos < p_{j+1}$. 
Hence, the DBWT-repetition of $u_{j}$ covers the starting position $\LF_{\delta}(p_{x})$ of the F-interval of $v_{x}$ on $L_{\delta}$. 

Otherwise~(i.e., the DBWT-repetitions of $u_{j^{\prime}+1}$ cover the starting position of the F-interval of $v_{x}$), 
$\LF_{\delta}(p_{x}) > \inspos$ and $\LF_{\delta}(p_{x}) + \shfn(\LF_{\delta}(p_{x})) \in [\inspos + 1, p_{j+1}]$. 
$\shfn(\LF_{\delta}(p_{x})) = 1$ by $\LF_{\delta}(p_{x}) > \inspos$, 
and hence $\LF_{\delta}(p_{x}) \in [\inspos, p_{j+1}-1]$. 
The DBWT-repetition $L_{\delta}[p_{j}..p_{j+1}-1]$ of $u_{j}$ covers the starting position $\LF_{\delta}(p_{x})$ of the F-interval of $v_{x}$ on $L_{\delta}$ by $[p_{j}, p_{j+1}-1] \supseteq [\inspos, p_{j+1}-1]$ and $\LF_{\delta}(p_{x}) \in [\inspos, p_{j+1}-1]$. 
Therefore, Lemma~\ref{lem:new_edge_goal_F}-(iv) holds. 

(v) The DBWT-repetition of $u_{x}$ is substring $L_{\delta+1}[p_{x} + \shfn(p_{x})..p_{x+1} - 1 + \shfn(p_{x})]$ by Lemma~\ref{lem:shiftU}-(i), 
and the F-interval of $v_{y}$ starts at position $\LF_{\delta}(p_{y}) + \shfn(\LF_{\delta}(p_{y}))$ on $F_{\delta+1}$ by Lemma~\ref{lem:shiftV}-(i). 
If the DBWT-repetition of $u_{x}$ covers the starting position of the F-interval of $v_{y}$ on $L_{\delta+1}$, 
then $\LF_{\delta}(p_{y}) \in [p_{x}, p_{x+1} - 1]$ by $\shfn(p_{x}) = \shfn(\LF_{\delta}(p_{y}))$ and $\LF_{\delta}(p_{y}) + \shfn(\LF_{\delta}(p_{y})) \in [p_{x} + \shfn(p_{x}), p_{x+1} - 1 + \shfn(p_{x})]$. 
Before the update operation is executed, 
the DBWT-repetition of $u_{x}$ is $L_{\delta}[p_{x}..p_{x+1} - 1]$ on $L_{\delta}$, 
and the F-interval of $v_{y}$ starts at position $\LF_{\delta}(p_{y})$ on $F_{\delta}$. 
Hence, the DBWT-repetition of $u_{x}$ covers the starting position $\LF_{\delta}(p_{y})$ of the F-interval of $v_{y}$ on $L_{\delta}$, i.e., $(v_{y}, u_{x})$ is contained in the $E_{F}$ stored in LF-interval graph $\graphds(D^{\alpha}_{\delta})$. 

By contrast, if $(v_{y}, u_{x})$ is contained in the $E_{F}$ stored in LF-interval graph $\graphds(D^{\alpha}_{\delta})$, 
then $\LF_{\delta}(p_{y}) + \shfn(\LF_{\delta}(p_{y})) \in [p_{x} + \shfn(p_{x}), p_{x+1} - 1 + \shfn(p_{x})]$ by 
$\LF_{\delta}(p_{y}) \in [p_{x}, p_{x+1} - 1]$ and $\shfn(p_{x}) = \shfn(\LF_{\delta}(p_{y}))$. 
The DBWT-repetition $L_{\delta+1}[p_{x} + \shfn(p_{x})..p_{x+1} - 1 + \shfn(p_{x})]$ of $u_{x}$ covers the starting position $\LF_{\delta}(p_{y}) + \shfn(\LF_{\delta}(p_{y}))$ of the F-interval of $v_{y}$ on $L_{\delta+1}$ by 
$\LF_{\delta}(p_{y}) + \shfn(\LF_{\delta}(p_{y})) \in [p_{x} + \shfn(p_{x}), p_{x+1} - 1 + \shfn(p_{x})]$. 
Therefore, Lemma~\ref{lem:new_edge_goal_F}-(v) holds.
\end{proof}

Function $\eqfn(u, u^{\prime})$ returns $1$ if two given nodes $u$ and $u^{\prime}$ are the same; otherwise, it returns $0$. 
The following lemma uses this function to bound the number of directed edges pointing to a given node $u \in U$ for $\graphds(D^{2\alpha+1}_{\delta+1})$. 
\begin{lemma}\label{lem:the_number_of_directed_edges_to_u} 
Node $u_{h} \in U$ is a node connected to the head of the directed edge starting at the new node $v_{j^{\prime} + 1} \in V$ created in the split-node step. 
After the update operation of $\graphds(D^{\alpha}_{\delta})$ has been executed, 
the following three statements hold: 
(i) for any node $u \in \{ u_{1}, u_{2}, \ldots, u_{k} \} \setminus \{ u_{i}, u_{j} \}$, 
the number of directed edges pointing to $u$ is at most $\alpha - 1 + \eqfn(u, u_{h})$; 
(ii) for any node $u \in \{ u_{i^{\prime}}, u_{x^{\prime}} \}$, 
the number of directed edges pointing to $u$ is at most one; 
(iii) for any node $u \in \{ u_{j^{\prime}}, u_{j^{\prime}+1} \}$,
the number of directed edges pointing to $u$ is at most $\alpha - 1 + \eqfn(u, u_{h})$. 
\end{lemma}
\begin{proof}
(i) 
Let $u_{s}$ and $u_{s^{\prime}}$ be the two nodes connected to the heads of the directed edges starting from $v_{i}$ and $v_{j}$, 
respectively, in $\graphds(D^{\alpha}_{\delta})$~(i.e., $(v_{i}, u_{s}), (v_{j}, u_{s^{\prime}}) \in E_{F}$). 
The number of directed edges starting from nodes in $\{ v_{1}, v_{2}, \ldots, v_{k} \} \setminus \{ v_{i}, v_{j} \}$ and pointing to $u$ 
is at most $\alpha - (1 + \eqfn(u, u_{s}) + \eqfn(u, u_{s^{\prime}}))$ in $\graphds(D^{\alpha}_{\delta})$ 
because the number of directed edges pointing to $u$ is at most $(\alpha - 1)$ in $\graphds(D^{\alpha}_{\delta})$. 
By Lemma~\ref{lem:new_edge_goal_F}-(v), in $\graphds(D^{2\alpha + 1}_{\delta+1})$, 
the number of directed edges starting from nodes in $\{ v_{1}, v_{2}, \ldots, v_{k} \} \setminus \{ v_{i}, v_{j} \}$ and pointing to $u$ 
is at most $\alpha - (1 + \eqfn(u, u_{s}) + \eqfn(u, u_{s^{\prime}}))$. 
The new directed edges pointing to $u$ is at most three by Lemma~\ref{lem:new_edge_goal_F}, 
and the directed edges start from $v_{j^{\prime}}, v_{j^{\prime}+1}$, and $v_{x^{\prime}}$. 
If the directed edge starting from $v_{x^{\prime}}$ points to $u$, 
then $\eqfn(u, u_{s}) = 1$~(i.e., $u$ is connected to the head of the directed edge starting at $v_{i}$) by Lemma~\ref{lem:new_edge_goal_F}-(i). 
Similarly, the directed edge starting from $v_{j^{\prime}}$ points to $u$, 
then $\eqfn(u, u_{s^{\prime}}) = 1$ by Lemma~\ref{lem:new_edge_goal_F}-(ii). 
Hence, in $\graphds(D^{2\alpha + 1}_{\delta+1})$, 
the number of directed edges pointing to $u$ is at most $\alpha - 1 + \eqfn(u, u_{h})$. 

(ii) The number of directed edges pointing to a node is at most one if 
the node represents a DBWT-repetition of length $1$. 
The two nodes $u_{i^{\prime}}$ and $u_{x^{\prime}}$ 
represent two DBWT-repetitions of length $1$. 
Hence, Lemma~\ref{lem:the_number_of_directed_edges_to_u}-(ii) holds. 

(iii) Because $u$ is a new node, 
$u$ is only connected to the heads of new directed edges in $E_{F}$. 
The number of new directed edges starting from nodes in $\{ v_{1}, v_{2}, \ldots, v_{k} \} \setminus \{ v_{i}, v_{j} \}$ 
and pointing to $u$ is at most $\alpha - (1 + \eqfn(u_{j}, u_{s}) + \eqfn(u_{j}, u_{s^{\prime}}) )$ by Lemma~\ref{lem:new_edge_goal_F}-(iv) 
because the number of directed edges pointing to $u_{j}$ is at most $(\alpha - 1)$ in $\graphds(D^{\alpha}_{\delta})$. 
The number of new directed edges starting from new nodes in $V$ and pointing to $u$ are at most three, 
and the directed edges start from $v_{j^{\prime}}, v_{j^{\prime}+1}$, and $v_{x^{\prime}}$. 
If $u$ is connected to the head of the directed edge starting at $v_{x^{\prime}}$, 
then $\eqfn(u_{j}, u_{s}) = 1$~(i.e., $u_{j}$ is connected to the head of the directed edge starting at $v_{i}$) 
because the DBWT-repetition of $u_{j}$ starts at the first position of $L_{\delta}$ in this case, 
and the F-interval of $v_{i}$ starts at the first position of $F_{\delta}$. 
Similarly, if $u$ is connected to the head of the directed edge starting at $v_{j^{\prime}}$, 
then $\eqfn(u_{j}, u_{s^{\prime}}) = 1$~(i.e., $u_{j}$ is connected to the head of the directed edge starting at $v_{j}$). 
Hence, 
the number of directed edges pointing to $u$ is at most $\alpha - 1 + \eqfn(u, u_{h})$ in $\graphds(D^{2\alpha + 1}_{\delta+1})$.
\end{proof}
As a result, Lemma~\ref{lem:update_edge_count}-(i) holds by Lemma~\ref{lem:the_number_of_directed_edges_to_u}.

\subparagraph{Proof of Lemma~\ref{lem:update_edge_count}-(ii).}
Set $E_{L}$ is symmetric to set $E_{F}$, and hence, the following two lemmas hold. 

\begin{lemma}\label{lem:new_edge_goal_U}
The following five statements hold: 
(i) For a node $v \in \{ v_{1}, v_{2}, \ldots, v_{k} \} \setminus \{ v_{i}, v_{j} \}$, 
if directed edge $(u_{i^{\prime}}, v)$ is contained in the $E_{L}$ stored in $\graphds(D^{2\alpha + 1}_{\delta+1})$, 
then directed edge $(u_{i}, v)$ is contained in the $E_{L}$ stored in $\graphds(D^{\alpha}_{\delta})$. 
(ii) 
For a node $v \in \{ v_{1}, v_{2}, \ldots, v_{k} \} \setminus \{ v_{i}, v_{j} \}$, 
if directed edge $(u_{j^{\prime}}, v)$ is contained in the $E_{L}$ stored in $\graphds(D^{2\alpha + 1}_{\delta+1})$, 
then directed edge $(u_{j}, v)$ is contained in the $E_{L}$ stored in $\graphds(D^{\alpha}_{\delta})$. 
(iii) Directed edge $(u_{x^{\prime}}, v_{i^{\prime}})$ is contained in the $E_{L}$ stored in $\graphds(D^{2\alpha + 1}_{\delta+1})$. 
(iv) For a node $u \in \{ u_{1}, u_{2}, \ldots, u_{k} \} \setminus \{ u_{i}, u_{j} \}$, 
if either directed edge $(u, v_{j^{\prime}})$ or $(u, v_{j^{\prime}+1})$ is contained in the $E_{L}$ stored in $\graphds(D^{2\alpha + 1}_{\delta+1})$, 
then directed edge $(u, v_{j})$ is contained in the $E_{L}$ stored in $\graphds(D^{\alpha}_{\delta})$. 
(v) For two nodes $u \in \{ u_{1}, u_{2}, \ldots, u_{k} \} \setminus \{ u_{i}, u_{j} \}$ and $v \in \{ v_{1}, v_{2}, \ldots, v_{k} \} \setminus \{ v_{i}, v_{j} \}$, 
directed edge $(u, v)$ is contained in the $E_{L}$ stored in LF-interval graph $\graphds(D^{2\alpha + 1}_{\delta+1})$ 
if and only if $(u, v)$ is contained in the $E_{L}$ stored in LF-interval graph $\graphds(D^{\alpha}_{\delta})$. 
\end{lemma}
\begin{proof}
This lemma is symmetric to Lemma~\ref{lem:new_edge_goal_F}. 
We can prove Lemma~\ref{lem:new_edge_goal_U} by modifying the proof of Lemma~\ref{lem:new_edge_goal_F}. 
\end{proof}

\begin{lemma}\label{lem:the_number_of_directed_edges_to_v}
Node $v_{h^{\prime}} \in V$ is a node connected to the head of the directed edge starting at the new node $u_{j^{\prime} + 1} \in U$ created in the split-node step. 
After the update operation of $\graphds(D^{\alpha}_{\delta})$ has been executed, 
the following three statements hold: 
(i) for any node $v \in \{ v_{1}, v_{2}, \ldots, v_{k} \} \setminus \{ v_{i}, v_{j} \}$, 
the number of directed edges pointing to $v$ is at most $\alpha - 1 + \eqfn(v, v_{h^{\prime}})$; 
(ii) for any node $v \in \{ v_{i^{\prime}}, v_{x^{\prime}} \}$, 
the number of directed edges pointing to $v$ is at most one; 
(iii) for any node $v \in \{ v_{j^{\prime}}, v_{j^{\prime}+1} \}$,
the number of directed edges pointing to $v$ is at most $\alpha - 1 + \eqfn(v, v_{h^{\prime}})$. 
\end{lemma}
\begin{proof}
This lemma is symmetric to Lemma~\ref{lem:the_number_of_directed_edges_to_u}. 
We can prove Lemma~\ref{lem:the_number_of_directed_edges_to_v} by modifying the proof of Lemma~\ref{lem:the_number_of_directed_edges_to_u}. 
\end{proof}

Hence, Lemma~\ref{lem:update_edge_count}-(ii) holds by Lemma~\ref{lem:the_number_of_directed_edges_to_v}.


\subsection{Proof of Theorem~\ref{thm:insert_balance}}\label{app:insert_balance}
(i) The four steps of the update operation and updating the four data structures in the given LF-interval graph take $O(\alpha + \log k)$ time in total.

(ii) Recall that LF-interval graph $\graphds(D^{\alpha}_{\delta})$ consists of nine sets $U, V$, $E_{LF}, E_{L}, {E}_{F}$, 
$B_{U}, B_{V}$, $B_{L}$, and $B_{F}$. 
Clearly, the five sets $U$, $V$, $E_{LF}$, $B_{U}$, and $B_{V}$ are correctly updated. 
The directed edges in the $E_{F}$ stored in LF-interval graph $\graphds(D^{2\alpha + 1}_{\delta+1})$ 
can be divided into (i) directed edges connected to new nodes~(i.e., $u_{i^{\prime}}, v_{i^{\prime}}, u_{j^{\prime}}, v_{j^{\prime}}, u_{j^{\prime}+1}, 
v_{j^{\prime}+1}$, $u_{x^{\prime}}$, and $v_{x^{\prime}}$) 
and (ii) all other directed edges. 
The directed edges connected to new nodes are created in the update-edge step of the update operation. 
By Lemma~\ref{lem:new_edge_goal_F}-(v), the other directed edges are contained in the $E_{F}$ stored in LF-interval graph $\graphds(D^{\alpha}_{\delta})$. 
Each directed edge in the $E_{F}$ stored in $\graphds(D^{\alpha}_{\delta})$ is not removed from $E_{L}$ in the update-edge step 
if and only if the directed edge is contained in the $E_{F}$ stored in $\graphds(D^{2\alpha + 1}_{\delta+1})$. 
Hence, set $E_{F}$ has been correctly updated. 
Similarly, set $E_{L}$ has been correctly updated. 
The labels of new directed edges are computed in the update-edge step, 
and hence, the following lemma ensures that $B_{L}$ and $B_{F}$ are correctly updated. 
\begin{lemma}\label{lem:static_edge}
For two nodes $u_{x} \in \{ u_{1}, u_{2}, \ldots, u_{k} \} \setminus \{ u_{i}, u_{j} \}$ and $v_{y} \in \{ v_{1}, v_{2}, \ldots, v_{k} \} \setminus \{ v_{i}, v_{j} \}$, 
the following two statements hold: 
(i) label $B_{L}(u_{x}, v_{y})$ in $\graphds(D^{2\alpha + 1}_{\delta+1})$ is equal to 
label $B_{L}(u_{x}, v_{y})$ in $\graphds(D^{\alpha}_{\delta})$; 
(ii) label $B_{F}(v_{y}, u_{x})$ in $\graphds(D^{2\alpha + 1}_{\delta+1})$ is equal to 
label $B_{F}(v_{y}, u_{x})$ in $\graphds(D^{\alpha}_{\delta})$. 
\end{lemma}
\begin{proof}
The DBWT-repetition of $u_{x}$ on $L_{\delta+1}$ is $L_{\delta+1}[p_{x} + \shfn(p_{x})..p_{x+1}-1 + \shfn(p_{x})]$ by Lemma~\ref{lem:shiftU}-(i). 
The F-interval of $v_{y}$ on $F_{\delta+1}$ is $[\LF_{\delta}(p_{y}) + \shfn(\LF_{\delta}(p_{y})), \LF_{\delta}(p_{y+1}-1) + \shfn(\LF_{\delta}(p_{y}))]$ by Lemma~\ref{lem:shiftV}-(ii).

(i) We have $\inspos \leq \LF_{\delta}(p_{y})$ or $\LF_{\delta}(p_{y+1}-1) < \inspos$ because new node $v_{i^{\prime}}$ is inserted into the doubly linked list of $V$ at the position next to a node in the replace-node step of the update operation, and the new node $v_{i^{\prime}}$ represents the F-interval starting at position $\inspos$ on $F_{\delta+1}$. 
We have $\shfn(p_{x}) = \shfn(\LF_{\delta}(p_{y}))$ by (1) $p_{x} \in [\LF_{\delta}(p_{y}), \LF_{\delta}(p_{y+1}-1)]$ and 
(2) $\inspos \leq \LF_{\delta}(p_{y})$ or $\LF_{\delta}(p_{y+1}-1) < \inspos$. 
The label $B_{L}(u_{x}, v_{y})$ in $\graphds(D^{2\alpha + 1}_{\delta+1})$ is $p_{x} + \shfn(p_{x}) - (\LF_{\delta}(p_{y}) + \shfn(\LF_{\delta}(p_{y}))) = p_{x} - \LF_{\delta}(p_{y})$. 
The label $B_{L}(u_{x}, v_{y})$ in $\graphds(D^{\alpha}_{\delta})$ is $p_{x} - \LF_{\delta}(p_{y})$. 
Hence, Lemma~\ref{lem:static_edge}-(i) holds. 

(ii) Lemma~\ref{lem:static_edge}-(ii) is symmetric to Lemma~\ref{lem:static_edge}-(i). 
We can prove Lemma~\ref{lem:static_edge}-(ii) by modifying the proof of Lemma~\ref{lem:static_edge}-(i). 
\end{proof}
Hence, the update operation of the LF-interval graph $\graphds(D^{\alpha}_{\delta})$ 
correctly outputs the LF-interval graph $\graphds(D^{2\alpha + 1}_{\delta+1})$. 

Next, Lemma~\ref{lem:the_number_of_directed_edges_to_u} indicates that 
the DBWT $D^{2\alpha + 1}_{\delta+1}$ contains at most one $\alpha$-heavy DBWT-repetition, 
and the $\alpha$-heavy DBWT-repetition covers at most $\alpha$ starting positions of F-intervals. 
Similarly, 
Lemma~\ref{lem:the_number_of_directed_edges_to_v} indicates that 
the DBWT $D^{2\alpha + 1}_{\delta+1}$ contains at most one $\alpha$-heavy F-interval, 
and the $\alpha$-heavy F-interval covers at most $\alpha$ starting positions of DBWT-repetitions.
Hence, the DBWT $D^{2\alpha + 1}_{\delta+1}$ is $(2\alpha + 1)$-balanced 
with at most two $\alpha$-heavy DBWT-repetitions and at most two $\alpha$-heavy F-intervals. 
Finally, Theorem~\ref{thm:insert_balance}-(ii) holds.


\oldsentence{
\section{Details for Section~\ref{sec:slow_update}}
\subsection{Supplementary examples and figures}\label{app:slow_update:examples}
\subsection{Proof of Lemma~\ref{lem:mod}}\label{app:mod}
\subsection{Proof of Lemma~\ref{lem:insert_and_delete_nodes_by_slow_update}}\label{app:insert_and_delete_nodes_by_slow_update}
}

\section{Details for Section~\ref{sec:fast_update}}
\oldsentence{
\subsection{Supplementary examples and figures}\label{app:fast_update:examples}
\subparagraph{The replacement of two nodes $v_{i}$ and $v_{j}$.}
\subparagraph{Merged nodes for three cases A, B, and C.}
\subsection{Proof of Lemma~\ref{lem:b-tree_update_lemma1}}\label{app:b-tree_update_lemma1}
\subsection{Proof of Lemma~\ref{lem:fast_insertion}}\label{app:fast_insertion}
}

\subsection{Details of the fast update operation}\label{app:details_of_fast_update_operation}
\begin{figure*}[t]
 \begin{center}
		\includegraphics[scale=0.6]{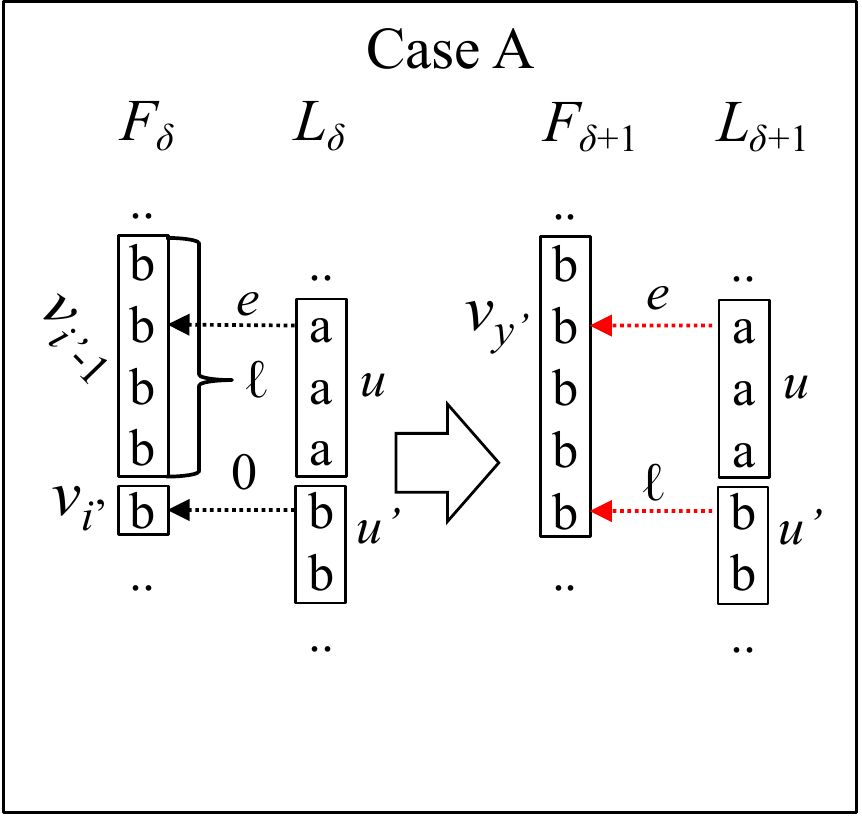}
		\includegraphics[scale=0.6]{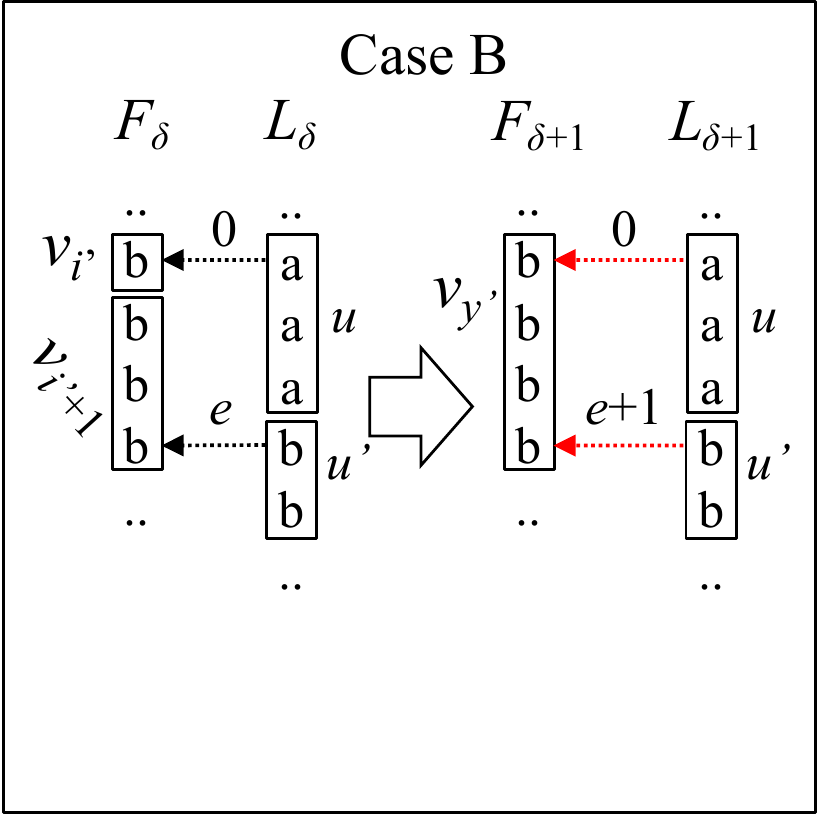}
		\includegraphics[scale=0.6]{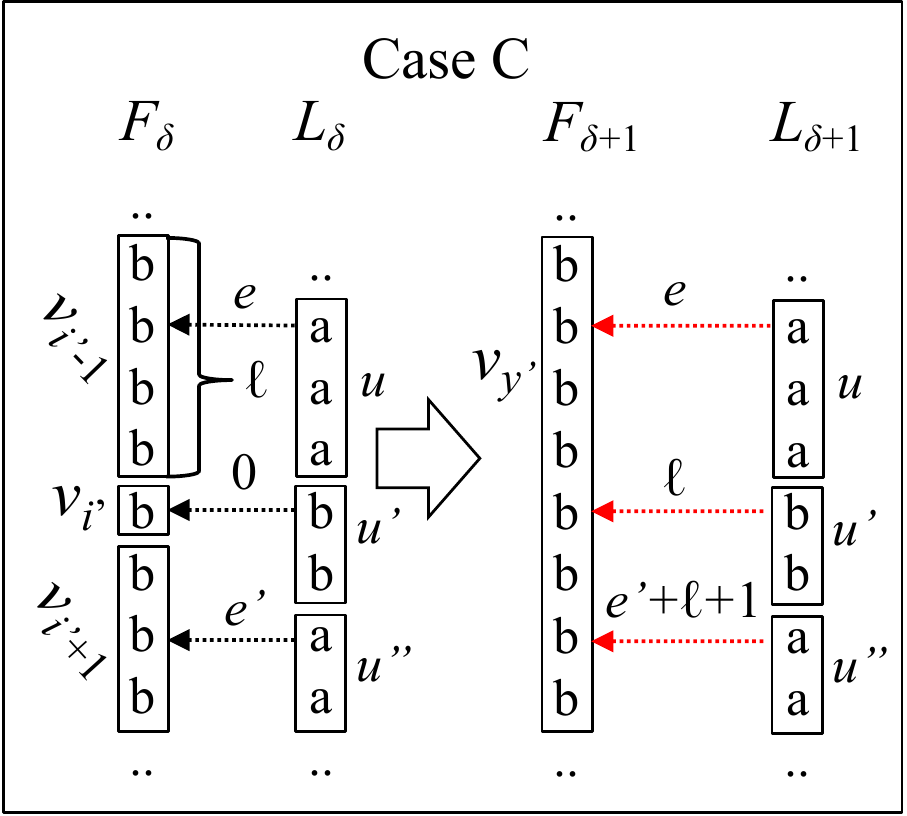}
	  \caption{
	  Three cases A, B, and C for replacing directed edges in $E_{L}$. 
	  Red arrows are new directed edges in $E_{L}$. 
	  The integers on the directed edges are their labels. 
	  }
 \label{fig:replace_edges_in_merge}
 \end{center}
\end{figure*}

In this subsection, we explain the details of the fast update operation 
and show that the operation can be performed in $O(\alpha)$ time. 
The fast update operation consists of the following five phases.  
In the first phase, 
the fast update operation updates the given LF-interval graph $\graphds(D^{\alpha}_{\delta})$ by the update operation presented in Section~\ref{sec:update_graph}. 
Here, the replace-node step is executed in $O(1)$ time using the technique presented in Section~\ref{sec:fast_update}, 
and updating the B-tree of $V$ is skipped. This phase takes $O(\alpha)$ time. 

In the second phase, 
the fast update operation appropriately merges nodes into new nodes $u_{y^{\prime}} \in U$ and $u_{y^{\prime}} \in V$ 
using the procedure for case A, B, or C, presented in Section~\ref{sec:fast_update}. 
The doubly linked lists of $U$ and $V$ are updated according to the merging of nodes. 
This phase takes $O(\alpha)$ time. 

We show that one of the three cases A, B, and C always holds.
The three cases do not hold only if neither $v_{i^{\prime}-1}$ nor $v_{i^{\prime}+1}$ has a label including character $c$, 
but we can show that either of the two nodes always has a label including character $c$ using a property of LF function. 
Hence, one of the three cases always holds. 
Formally, the following lemma holds. 
\begin{lemma}\label{lem:not_case_D}
Either $v_{i^{\prime}-1}$ or $v_{i^{\prime}+1}$ has a label including character $c$ 
for fast update operation $\fastUpdate(\graphds(D^{\alpha}_{\delta}), c)$.
\end{lemma}
\begin{proof}
We use proof by contradiction to prove the lemma. 
LF function has the property that $\LF_{\delta}(x) \neq x$ for all $\delta \geq 2$ and $x \in \{ 1, 2, \ldots, \delta \}$. 
By contrast, if neither $v_{i^{\prime}-1}$ nor $v_{i^{\prime}+1}$ has a label including character $c$, 
then there exists an integer $y \in \{ 1, 2, \ldots, \delta \}$ such that $\LF_{\delta}(y) = y$, 
which contradicts the property of LF function. 

We show that $\LF_{\delta}(x) \neq x$ for all $\delta \geq 2$ and $x \in \{ 1, 2, \ldots, \delta \}$. 
Recall that for the suffix $T_{\delta}$ of length $\delta$,  
$x_{1}, x_{2}, \ldots, x_{\delta}$ are a permutation of sequence $1, 2, \ldots, \delta$ 
such that $T_{\delta}[x_{1}..\delta] \prec T_{\delta}[x_{2}..\delta] \prec \cdots \prec T_{\delta}[x_{n}..\delta]$), 
and this permutation is used to define the BWT of $T_{\delta}$. 
For an integer $t \in \{ 1, 2, \ldots, \delta \}$, 
let $t^{\prime} \in \{ 1, 2, \ldots, \delta \}$ be the integer satisfying either of the following two conditions: 
(i) $x_{t} - 1 = x_{t^{\prime}}$ or (ii) $x_{t} = 1$ and $x_{t^{\prime}} = \delta$. 
Then, $\LF_{\delta}(t) = t^{\prime}$ from the definition of LF function~(see also Section~\ref{sec:preliminary}). 
Because $x_{1}, x_{2}, \ldots, x_{\delta}$ are a permutation of sequence $1, 2, \ldots, \delta$, 
$t \neq t^{\prime}$ always holds for $\delta \geq 2$. 
Hence, $\LF_{\delta}(x) \neq x$ for all $\delta \geq 2$ and $x \in \{ 1, 2, \ldots, \delta \}$. 

Next, the fast update operation ensures that 
either or both $u_{i-1}$ and $u_{i+1}$ have a label including character $c$. 
This fact indicates that 
neither $v_{i^{\prime}-1}$ nor $v_{i^{\prime}+1}$ has a label including character $c$ only if 
either of the following two conditions.  
(i) Node $u_{i-1}$ has a label including character $c$, $u_{i+1}$ does not have a label including character $c$, 
and the new node $u_{x^{\prime}}$ representing special character $\$$ is inserted into the doubly linked list of $U$ 
between $u_{i-1}$ and $u_{i^{\prime}}$~(i.e., $u_{i^{\prime}-1} = u_{x^{\prime}}$ and $u_{i^{\prime}+1} = u_{i+1}$).  
(ii) Node $u_{i-1}$ does not have a label including character $c$, $u_{i+1}$ has a label including character $c$, 
and $u_{x^{\prime}}$ is inserted into the doubly linked list of $U$ 
between $u_{i^{\prime}}$ and $u_{i+1}$~(i.e., $u_{i^{\prime}-1} = u_{i-1}$ and $u_{i^{\prime}+1} = u_{x^{\prime}}$). 

We show that $\LF_{\delta}(\inspos - 1) = \inspos - 1$ 
for the position $\inspos$ of special character $\$$ in BWT $L_{\delta+1}$ if the fast update operation satisfies the former condition. 
In this case, for the node $v_{i-1} \in V$ connected to $u_{i-1}$ by edge $(u_{i-1}, v_{i-1}) \in E_{LF}$ 
and the node $v_{i^{\prime}}$ representing the input character $c$ of this update operation, 
$v_{i^{\prime}}$ is inserted into the doubly linked list of $V$ at the position next to $v_{i-1}$ by the LF formula. 
Analogous to the extension of BWT described in Section~\ref{lab:ExBWT}, 
the node $v_{i^{\prime}}$ represents the $\inspos$-th character of $F_{\delta+1}$. 
Hence, the F-interval represented by $v_{i-1}$ ends at position $\inspos - 1$ on $F_{\delta}$. 

Similarly, 
the DBWT-repetition represented by $u_{i-1}$ ends at position $\inspos - 1$ on $L_{\delta}$ 
because the node $u_{x^{\prime}}$ represents the $\inspos$-th character of $L_{\delta+1}$, 
and $u_{x^{\prime}}$ is inserted into the doubly linked list of $U$ at the position next to $u_{i-1}$. 
Because $u_{i-1}$ is connected to $v_{i-1}$ by undirected edge $(u_{i-1}, v_{i-1}) \in E_{LF}$, 
we obtain $\LF_{\delta}(\inspos-1) = \inspos - 1$ by the LF formula. 

Next, 
we show that $\LF_{\delta}(\inspos) = \inspos$ if the fast update operation satisfies the latter condition. 
This proof is symmetric to the proof for the former condition. 
In this case, 
for the node $v_{i+1} \in V$ is connected to $u_{i+1}$ by edge $(u_{i+1}, v_{i+1}) \in E_{LF}$, 
$v_{i^{\prime}}$ is inserted into the doubly linked list of $V$ at the position previous to $v_{i-1}$ by the LF formula. 
Similarly, 
the latter condition ensures that $u_{x^{\prime}}$ is inserted into the doubly linked list of $U$ at the position previous to $u_{i-1}$. 
Hence, the F-interval represented by $v_{i+1}$ and the DBWT-repetition of $u_{i+1}$ 
start at position $\inspos$ on $F_{\delta}$ and $L_{\delta}$, respectively. 
This fact indicates that $\LF_{\delta}(\inspos) = \inspos$ by the LF formula. 

Therefore, either $v_{i^{\prime}-1}$ or $v_{i^{\prime}+1}$ must have a label including character $c$, 
i.e., we obtain Lemma~\ref{lem:not_case_D}. 
\end{proof}

In the third phase, 
the fast update operation updates edges and their labels~(i.e., five sets $E_{LF}, E_{L}, E_{F}, B_{L}$, and $B_{F}$) 
according to the second phase~(i.e., the merging of nodes). 
This phase takes takes $O(\alpha)$ time, which is explained below. 
\subparagraph{Updating set $E_{LF}$.} 
We remove the edges connected to merged nodes from $E_{LF}$. 
For case A, 
two edges $(u_{i^{\prime}-1}, v_{i^{\prime}-1})$ and $(u_{i^{\prime}}, v_{i^{\prime}})$ are removed from $E_{LF}$. 
For case B, two edges $(u_{i^{\prime}}, v_{i^{\prime}})$ and $(u_{i^{\prime}+1}, v_{i^{\prime}+1})$ are removed from $E_{LF}$. 
For case C, three edges $(u_{i^{\prime}-1}, v_{i^{\prime}-1})$, $(u_{i^{\prime}}, v_{i^{\prime}})$, and $(u_{i^{\prime}+1}, v_{i^{\prime}+1})$ are removed from $E_{LF}$. 
Subsequently, new edge $(u_{y^{\prime}}, v_{y^{\prime}})$ is inserted into $E_{LF}$.

\subparagraph{Updating two sets $E_{L}$ and $B_{L}$.}
First, directed edges pointing to merged nodes are replaced with new edges pointing to the new node $v_{y^{\prime}}$. 
Figure~\ref{fig:replace_edges_in_merge} illustrates the replaced edges and new edges in $E_{L}$ for the three cases A, B, and C.  
Formally, each directed edge $(u, v) \in E_{L}$ pointing to merged nodes is replaced with new edge $(u, v_{y^{\prime}})$. 
For case A, $v = v_{i^{\prime}-1}$ or $v = v_{i^{\prime}}$. 
The label $B_{L}(u, v_{y^{\prime}})$ of the new edge $(u, v_{y^{\prime}})$ is the same as 
the label of the replaced edge $(u, v)$ if $v = v_{i^{\prime}-1}$; 
otherwise, $B_{L}(u, v_{y^{\prime}}) = \ell$ for the label $(c, \ell)$ of the node $v_{i^{\prime}-1}$. 
For case B, $v = v_{i^{\prime}}$ or $v = v_{i^{\prime}+1}$. 
The label of the new edge $(u, v_{y^{\prime}})$ is the same as the label of the replaced edge if $v = v_{i^{\prime}}$; 
otherwise, the label of the new edge is $(1 + B_{L}(u, v_{i^{\prime}+1}))$ 
for the label $B_{L}(u, v_{i^{\prime}+1})$ of the replaced edge $(u, v_{i^{\prime}+1})$. 
For case C, $v = v_{i^{\prime}-1}$, $v = v_{i^{\prime}}$, or $v = v_{i^{\prime}+1}$. 
For $v = v_{i^{\prime}-1}$, 
the label of the new edge $(u, v_{y^{\prime}})$ is the same as the label of the replaced edge $(u, v_{i^{\prime}-1})$. 
For $v = v_{i^{\prime}}$, 
the label of the new edge is $\ell$. 
For $v = v_{i^{\prime}+1}$, the label of the new edge is $(\ell + 1 + B_{L}(u, v_{i^{\prime}+1}) )$.

Second, directed edges starting from merged nodes are removed from $E_{L}$. 
For case A, the directed edges starting from $u_{i^{\prime}-1}$ and $u_{i^{\prime}}$ are removed from $E_{L}$. 
For case B, the directed edges starting from $u_{i^{\prime}}$ and $u_{i^{\prime}+1}$ are removed from $E_{L}$.
For case C, the directed edges starting from $u_{i^{\prime}-1}$, $u_{i^{\prime}}$, and $u_{i^{\prime}+1}$ are removed from $E_{L}$.

Third, the new directed edge $(u_{y^{\prime}}, v)$ starting at the new node $u_{y^{\prime}}$ is inserted into $E_{L}$ for a node $v \in V$. 
The node $v$ and label $B_{L}(u_{y^{\prime}}, v)$ of the new directed edge are determined by the three cases A, B, and C. 
For case A, 
let $v^{\prime} \in V$ be the node connected to $u_{i^{\prime}-1}$ by the directed edge $(u_{i^{\prime}-1}, v^{\prime})$ in $E_{L}$. 
If $v^{\prime}$ is not merged into new node $v_{y^{\prime}} \in V$~(i.e., $v^{\prime} \not \in \{ v_{i^{\prime} - 1}, v_{i^{\prime}} \} $), 
then $v = v^{\prime}$, 
and the label $B_{L}(u_{y^{\prime}}, v)$ is set to the label $B_{L}(u_{i^{\prime}-1}, v^{\prime})$ of 
the directed edge from $u_{i^{\prime}-1}$ to $v^{\prime}$ in $E_{L}$. 
Otherwise~(i.e., $v^{\prime}$ is merged into new node $v_{y^{\prime}}$), 
we appropriately determine $v$ and $B_{L}(u_{y^{\prime}}, v)$ according to the merging of nodes, 
which is similar to the first step for updating two sets $E_{L}$ and $B_{L}$. 
That is, if $v^{\prime} = v_{i^{\prime} - 1}$, 
then, $v = v_{y^{\prime}}$ and 
the label $B_{L}(u_{y^{\prime}}, v)$ is set to the label $B_{L}(u_{i^{\prime}-1}, v^{\prime})$ of 
the directed edge from $u_{i^{\prime}-1}$ to $v^{\prime}$, which was removed from $E_{L}$ 
by the first step for updating two sets $E_{L}$ and $B_{L}$. 
Otherwise~(i.e., $v^{\prime} = v_{i^{\prime}}$), 
$v = v_{y^{\prime}}$ and 
the label $B_{L}(u_{y^{\prime}}, v)$ is set to $\ell$ for the label $(c, \ell)$ of the node $v_{i^{\prime}-1}$. 
Similarly, $v$ and $B_{L}(u_{y^{\prime}}, v)$ are determined according to the merging of nodes for cases B and C.

The number of removed edges and new edges for updating $E_{L}$ can be bounded by $O(\alpha)$ because 
the first phase creates the LF-interval graph for a $(2\alpha+1)$-balanced DBWT. 
Hence, updating both sets $E_{L}$ and $B_{L}$ takes $O(\alpha)$ time.

\subparagraph{Updating sets $E_{F}$ and $B_{F}$.}
Updating sets $E_{F}$ and $B_{F}$ is symmetric to updating sets $E_{L}$ and $B_{L}$. 
Hence, updating both sets $E_{F}$ and $B_{F}$ takes $O(\alpha)$ time.

In the fourth phase, 
the fast update operation updates the order maintenance data structure and the arrays for $\alpha$-heavy DBWT-repetitions and F-intervals according to the merging of nodes. 
This phase takes $O(\alpha)$ time, which is explained below.

\subparagraph{Updating the order maintenance data structure.}
Merged nodes are removed from the order maintenance data structure. 
In addition, new node $u_{y^{\prime}}$ is inserted into the data structure. 
Updating the order maintenance data structure takes $O(1)$ time.

\subparagraph{Updating the array for $\alpha$-heavy DBWT-repetitions.}
We delete each node from the array for $\alpha$-heavy DBWT-repetitions 
if the node does not represent an $\alpha$-heavy DBWT-repetition. 
This array stores at most two nodes for $\alpha$-heavy DBWT-repetitions before the nodes are merged 
because the update operation presented in Section~\ref{sec:update_graph} was executed in the first phase. 
Hence, the deletion takes $O(\alpha)$ time by using Lemma~\ref{lem:verify_heavy}. 

Next, if merging nodes creates nodes representing new $\alpha$-heavy DBWT-repetitions, 
we insert the nodes representing new $\alpha$-heavy DBWT-repetitions into the array for $\alpha$-heavy DBWT-repetitions. 
The target nodes inserted into the array are limited to new node $u_{y^{\prime}}$ and 
the node $u \in U$ connected to new node $v_{y^{\prime}}$ by directed edge $(v_{y^{\prime}}, u) \in E_{F}$ 
because the number of directed edges pointing to the other nodes in $U$ 
is not increased by the merging of nodes. 
The insertion takes $O(\alpha)$ time using Lemma~\ref{lem:verify_heavy}, 
and hence updating the array takes $O(\alpha)$ time.

\subparagraph{Updating the array for $\alpha$-heavy F-intervals.}
Updating the array for $\alpha$-heavy F-intervals is symmetric to updating the array for $\alpha$-heavy DBWT-repetitions. 
Hence, the array for $\alpha$-heavy F-intervals can be updated in $O(\alpha)$ time. 

In the fifth phase, we update the B-tree of $V$. 
This phase takes in $O(\alpha)$ time using the technique explained in Section~\ref{sec:fast_update}. 
Therefore, the fast update operation runs in $O(\alpha)$ time in total.

\subsection{Proof of Lemma~\ref{lem:fast_update_result}}\label{app:fast_update_result}
\begin{proof}
(i) The bottleneck of the update operation presented in Section~\ref{sec:update_graph} is the update of the B-tree of $V$ and replace-node step execution, which take $O(\log k)$ time. 
We already showed that the fast update operation can update the B-tree of $V$ in $O(1)$ time and execute the replace-node step in $O(1)$ time. 
Hence, the fast update operation runs in $O(\alpha)$ time~(see also Appendix~\ref{app:details_of_fast_update_operation}).

(ii) 
We show that the fast update operation $\fastUpdate(\graphds(D^{\alpha}_{\delta}), c)$ correctly outputs the LF-interval graph $\graphds(D^{2\alpha + 1}_{\delta+1})$ for DBWT $D^{2\alpha + 1}_{\delta+1}$. 
The fast update operation creates the LF-interval graph $\graphds(D^{2\alpha + 1}_{\delta+1})$ via 
the LF-interval graph $\graphds(\hat{D}^{2\alpha + 1}_{\delta+1})$ outputted by the update operation in Section~\ref{sec:update_graph}~(see also Appendix~\ref{app:details_of_fast_update_operation}). 
Here, the DBWT $\hat{D}^{2\alpha + 1}_{\delta+1}$ is a DBWT of BWT $L_{\delta+1}$. 
The LF-interval graph $\graphds(D^{2\alpha + 1}_{\delta+1})$ is created by merging 
at most six nodes in the LF-interval graph $\graphds(\hat{D}^{2\alpha + 1}_{\delta+1})$, 
i.e., at most three nodes $u_{i^{\prime}-1}$, $u_{i^{\prime}}$, and $u_{i^{\prime}+1} \in U$ are merged into 
a new node $u_{y^{\prime}}$, 
and at most three nodes $v_{i^{\prime}-1}$, $v_{i^{\prime}}$, and $v_{i^{\prime}+1} \in V$ are merged into 
a new node $v_{y^{\prime}}$. 
$D^{2\alpha + 1}_{\delta+1}$ is also a DBWT of $L_{\delta+1}$ because 
the three nodes $u_{i^{\prime}-1}$, $u_{i^{\prime}}$, and $u_{i^{\prime}+1}$ represent three consecutive DBWT-repetitions in DBWT $\hat{D}^{2\alpha + 1}_{\delta+1}$. 
Hence, the fast update operation correctly outputs the LF-interval graph for DBWT $D^{2\alpha + 1}_{\delta+1}$. 

Next, we show that DBWT $D^{2\alpha + 1}_{\delta+1}$ is $(2\alpha + 1)$-balanced 
with at most two $\alpha$-heavy DBWT-repetitions and at most two $\alpha$-heavy F-intervals. 
Recall that $u_{h} \in U$~(respectively, $v_{h^{\prime}} \in V$) is the node connected to the head of the directed edge starting at new node $v_{j^{\prime} + 1} \in V$~(respectively, $u_{j^{\prime}+1} \in U$) created in the split-node step. 
For a node $u \in U$~(respectively, $v \in V$) in LF-interval graph $\graphds(\hat{D}^{2\alpha + 1}_{\delta+1})$, 
let $m_{U}(u)$~(respectively, $m_{V}(v)$) be the number of directed edges pointing to node $u$~(respectively, $v$). 
By Lemma~\ref{lem:the_number_of_directed_edges_to_u}, 
$m_{U}(u_{h}) \leq \alpha$ and 
$m_{U}(u_{x}) \leq \alpha - 1$ for all node $u_{x} \in U \setminus \{ u_{h} \}$ in $\graphds(\hat{D}^{2\alpha + 1}_{\delta+1})$. 
Similarly, 
by Lemma~\ref{lem:the_number_of_directed_edges_to_v}, 
$m_{V}(v_{h^{\prime}}) \leq \alpha$ and 
$m_{V}(v_{x}) \leq \alpha - 1$ for all nodes $v_{x} \in V \setminus \{ v_{h^{\prime}} \}$ in $\graphds(\hat{D}^{2\alpha + 1}_{\delta+1})$.

For case A, the fast update operation merges two nodes $u_{i^{\prime}-1}$ and 
$u_{i^{\prime}}$ into new node $u_{y^{\prime}}$. At the same time, 
two nodes $v_{i^{\prime}-1}$ and 
$v_{i^{\prime}}$ are merged into new node $v_{y^{\prime}}$. 
The number of directed edges pointing to new node $u_{y^{\prime}}$ is never larger than $m_{U}(u_{i^{\prime}-1}) + m_{U}(u_{i^{\prime}})$. 
We have $m_{U}(u_{i^{\prime}-1}) + m_{U}(u_{i^{\prime}}) \leq 2\alpha -1$ 
by $m_{U}(u_{h}) \leq \alpha$ and 
$m_{U}(u_{x}) \leq \alpha - 1$ for all nodes $u_{x} \in U \setminus \{ u_{h} \}$. 
The number of directed edges pointing to each node in $U \setminus \{ u_{i^{\prime}-1}, u_{i^{\prime}}\}$ is not increased by merging nodes. 
Similarly, 
the number of directed edges pointing to the new node $v_{y^{\prime}}$ is never larger than $m_{V}(v_{i^{\prime}-1}) + m_{V}(v_{i^{\prime}}) \leq 2\alpha -1$. 
The number of directed edges pointing to each node in $V \setminus \{ v_{i^{\prime}-1}, v_{i^{\prime}}\}$ is not increased by merging nodes. 
Hence, $D^{2\alpha + 1}_{\delta+1}$ is $(2\alpha + 1)$-balanced 
with at most two $\alpha$-heavy DBWT-repetitions and at most two $\alpha$-heavy F-intervals for case A.

For case B, the fast update operation merges two nodes $u_{i^{\prime}}$ and 
$u_{i^{\prime}+1}$ into new node $u_{y^{\prime}}$. At the same time, 
two nodes $v_{i^{\prime}}$ and $v_{i^{\prime}+1}$ are merged into new node $v_{y^{\prime}}$. 
This case is symmetric to case A, and hence, $D^{2\alpha + 1}_{\delta+1}$ is $(2\alpha + 1)$-balanced 
with at most two $\alpha$-heavy DBWT-repetitions and at most two $\alpha$-heavy F-intervals for case B.

For case C, the fast update operation merges three nodes $u_{i^{\prime}-1}$, $u_{i^{\prime}}$,   
and $u_{i^{\prime}+1}$ into new node $u_{y^{\prime}}$. At the same time, 
three nodes $v_{i^{\prime}-1}$, $v_{i^{\prime}}$, and $v_{i^{\prime}+1}$ are merged into new node $v_{y^{\prime}}$.
The number of directed edges pointing to new node $u_{y^{\prime}}$ is never larger than $m_{U}(u_{i^{\prime}-1}) +m_{U}(u_{i^{\prime}}) + m_{U}(u_{i^{\prime}+1})$. 
Because $m_{U}(u_{h}) \leq \alpha$ and 
$m_{U}(u_{x}) \leq \alpha - 1$ for all node $u_{x} \in U \setminus \{ u_{h} \}$,  
We have $m_{U}(u_{i^{\prime}-1}) + m_{U}(u_{i^{\prime}+1}) \leq 2 \alpha - 1$. 
We have $m_{U}(u_{i^{\prime}}) \leq 1$ because the length of the DBWT-repetition of $u_{i^{\prime}}$ is $1$. 
Hence, $m_{U}(u_{i^{\prime}-1}) +m_{U}(u_{i^{\prime}}) + m_{U}(u_{i^{\prime}+1}) \leq 2 \alpha$. 
The number of directed edges pointing to each node in $U \setminus \{ u_{i^{\prime}-1}, u_{i^{\prime}}, u_{i^{\prime}+1}\}$ is not increased by merging nodes. 
Similarly, 
the number of directed edges pointing to new node $v_{y^{\prime}}$ is never larger than $m_{V}(v_{i^{\prime}-1}) + m_{V}(v_{i^{\prime}}) + m_{V}(v_{i^{\prime}+1})$, and $m_{V}(v_{i^{\prime}-1}) + m_{V}(v_{i^{\prime}}) + m_{V}(v_{i^{\prime}+1}) \leq 2 \alpha$. 
The number of directed edges pointing to each node in $V \setminus \{ v_{i^{\prime}-1}, v_{i^{\prime}}, v_{i^{\prime}+1}\}$ is not increased by merging nodes. 
Hence, $D^{2\alpha + 1}_{\delta+1}$ is $(2\alpha + 1)$-balanced 
with at most two $\alpha$-heavy DBWT-repetitions and at most two $\alpha$-heavy F-intervals for case C.

(iii)
The procedure of the fast update operation ensures that 
the targets inserted into the B-tree of $V$ can be limited to only three nodes $v_{x^{\prime}}$, $v_{j^{\prime}+1}$, and $v_{y^{\prime}}$, 
which is similar to Lemma~\ref{lem:insert_and_delete_nodes_by_slow_update}-(i). 
The fast update operation inserts each of the three nodes into the B-tree of $V$ if 
the node satisfies one of the three conditions of Lemma~\ref{lem:mod}, 
and hence 
the B-tree of $V$ in the outputted LF-interval graph contains all the nodes satisfying one of the three conditions of Lemma~\ref{lem:mod}.

Similarly, 
the targets deleted from the B-tree of $V$ can be limited to only four nodes $v_{i}$, $v_{j}$, $v_{i^{\prime}-1}$, and $v_{i^{\prime}+1}$,  
which is similar to Lemma~\ref{lem:insert_and_delete_nodes_by_slow_update}-(ii). 
The fast update operation deletes each of the four nodes from the B-tree of $V$ if 
the node satisfies none of the three conditions of Lemma~\ref{lem:mod}, 
and hence 
each node in the B-tree of $V$ satisfies one of the three conditions of Lemma~\ref{lem:mod}. 
Therefore, we obtain Lemma~\ref{lem:fast_update_result}-(iii).
\end{proof}

\section{Details for Section~\ref{sec:balancing}}
\oldsentence{
\subsection{Details of the balancing operation}\label{app:details_of_balancing_opeartion}
}

\subsection{Proof of Lemma~\ref{lem:split_time}}\label{app:split_time}
We already showed that 
each iteration of the balancing operation takes $O(\alpha)$ time if 
the balancing operation processes the LF-interval graph for an $O(\alpha)$-balanced DBWT at the iteration. 
The next lemma ensures that 
each iteration of the balancing operation outputs 
an LF-interval graph representing an $O(\alpha)$-balanced DBWT 
with $O(\alpha)$ $\alpha$-heavy DBWT-repetitions and $O(\alpha)$ $\alpha$-heavy F-intervals if 
an input of the balancing operation is the LF-interval graph for a $(2\alpha+1)$-balanced DBWT 
with $O(1)$ $\alpha$-heavy DBWT-repetitions and $\alpha$-heavy F-intervals. 

\begin{lemma}\label{lem:balance_run}
For all $\alpha \geq 4$, 
each iteration of the balancing operation outputs the LF-interval graph for an $O(\alpha)$-balanced DBWT 
with $O(\alpha)$ $\alpha$-heavy DBWT-repetitions and $O(\alpha)$ $\alpha$-heavy F-intervals if the input of the balancing operation is the LF-interval graph for a $(2\alpha+1)$-balanced DBWT with at most two $\alpha$-heavy DBWT-repetitions and at most two $\alpha$-heavy F-intervals.
\end{lemma}
Lemma~\ref{lem:balance_run} will be proved later. 
By Lemma~\ref{lem:balance_run}, 
each iteration of the balancing operation takes an LF-interval graph representing an $O(\alpha)$-balanced DBWT with $O(\alpha)$ $\alpha$-heavy DBWT-repetitions and $O(\alpha)$ $\alpha$-heavy F-intervals under the assumption that 
the input of the balancing operation is 
an LF-interval graph representing an $O(\alpha)$-balanced DBWT with $O(\alpha)$ $\alpha$-heavy DBWT-repetitions and $O(\alpha)$ $\alpha$-heavy F-intervals. 
This assumption always holds 
because the input of the balancing operation is the LF-interval graph $\graphds(D^{2\alpha+1}_{\delta})$ for a $(2\alpha+1)$-balanced DBWT including at most two $\alpha$-heavy DBWT-repetitions and at most two $\alpha$-heavy F-intervals. 
Hence, the balancing operation takes $O(\alpha)$ time per iteration for $\alpha \geq 4$. 

Similarly, each iteration of the balancing operation outputs the B-tree of $V$ that includes only nodes satisfying one of the three conditions of Lemma~\ref{lem:mod} if 
the iteration processes a B-tree of $V$ that includes only nodes satisfying one of the three conditions of Lemma~\ref{lem:mod}. 
Lemma~\ref{lem:split_time} assumes that the B-tree of $V$ includes only nodes satisfying one of the three conditions of Lemma~\ref{lem:mod} 
for the given LF-interval graph $\graphds(D^{2\alpha+1}_{\delta})$. 
Hence, the B-tree of $V$ in the outputted LF-interval graph contains only nodes satisfying one of the three conditions of Lemma~\ref{lem:mod}. 
Therefore, Lemma~\ref{lem:split_time} holds if Lemma~\ref{lem:balance_run} holds. 

In the next subsection, we prove Lemma~\ref{lem:balance_run}. 

\subsubsection{Proof of Lemma~\ref{lem:balance_run}}\label{app:balance_run}
We prove Lemma~\ref{lem:balance_run} using the \emph{excess} of a node and 
a \emph{weakly} $\alpha$-balanced LF-interval graph. 
The excess of a node $u \in U$ for an LF-interval graph 
is defined as the difference between (i) the number $m$ of directed edges pointing to $u$ in the LF-interval graph and (ii) $\alpha-1$ if $m \geq \alpha$. 
Otherwise, the excess of $u$ is $0$. 
In other words, the excess of $u$ is defined as $\max \{ m - (\alpha-1), 0 \}$. 
Similarly, the excess of a node $v \in V$ for an LF-interval graph is 
defined as $\max \{ m^{\prime} - (\alpha-1), 0 \}$ for the number $m^{\prime}$ of directed edges pointing to $v$ in the LF-interval graph.

We say that an LF-interval graph is weakly $\alpha$-balanced if the following two conditions hold: 
(i) the sum of the excesses of all the nodes in $U$ is at most $2\alpha+2$ and 
(ii) the sum of the excesses of all the nodes in $V$ is at most $2\alpha+2$. 
Every weakly $\alpha$-balanced LF-interval graph is $O(\alpha)$-balanced, 
and the input LF-interval graph $\graphds(D_{\delta}^{2\alpha+1})$ of the balancing operation $\balance(\graphds(D_{\delta}^{2\alpha+1}))$ is weakly $\alpha$-balanced. 
Formally, the following lemma holds. 
\begin{lemma}\label{lem:weakly_balance}
The following two statements hold: 
(i) The DBWT of every weakly $\alpha$-balanced LF-interval graph is $O(\alpha)$-balanced 
with $O(\alpha)$ $\alpha$-heavy DBWT-repetitions and $\alpha$-heavy F-intervals and 
(ii) the LF-interval graph $\graphds(D_{\delta}^{2\alpha+1})$ is weakly $\alpha$-balanced. 
\end{lemma}
\begin{proof}
(i) 
The DBWT of every weakly $\alpha$-balanced LF-interval graph is $O(\alpha)$-balanced 
because the number of directed edges pointing to a node in the LF-interval graph is at most $3\alpha + 3$.
If the number of nodes representing $\alpha$-heavy DBWT-repetitions is at least $2\alpha + 3$ in $U$, 
then the sum of the excesses of all the nodes in $U$ is larger than $2\alpha + 2$. 
Similarly, 
if the number of nodes representing $\alpha$-heavy F-intervals is at least $2\alpha + 3$ in $V$, 
then the sum of the excesses of all the nodes in $V$ is larger than $2\alpha + 2$. 
Hence, every weakly $\alpha$-balanced LF-interval graph has $O(\alpha)$ nodes for $\alpha$-heavy DBWT-repetitions and $\alpha$-heavy F-intervals. 
Therefore, we obtain Lemma~\ref{lem:weakly_balance}-(i).

(ii) 
The DBWT $D_{\delta}^{2\alpha+1}$ is $(2\alpha + 1)$-balanced, 
and it has at most two $\alpha$-heavy DBWT-repetitions and 
at most two $\alpha$-heavy F-intervals. 
The sum of the excesses of all the nodes in $U$ is at most $2\alpha+2$. 
Similarly, the sum of the excesses of all the nodes in $V$ is at most $2\alpha+2$. 
Hence, the LF-interval graph $\graphds(D_{\delta}^{2\alpha+1})$ is weakly $\alpha$-balanced. 
\end{proof}

The following lemma ensures that each iteration of balancing operation outputs a weakly $\alpha$-balanced LF-interval graph if the iteration is given a weakly $\alpha$-balanced LF-interval graph. 
\begin{lemma}\label{lem:weak2}
If an iteration of balancing operation is given a weakly $\alpha$-balanced LF-interval graph $\graphds(D_{\delta}^{O(\alpha)})$, 
then the iteration of the balancing operation outputs a weakly $\alpha$-balanced LF-interval graph $\graphds(\hat{D}^{O(\alpha)}_{\delta})$ for $\alpha \geq 4$. 
\end{lemma}

We obtain Lemma~\ref{lem:balance_run} by Lemmas~\ref{lem:weakly_balance} and \ref{lem:weak2} without any assumptions. 
The proof of Lemma~\ref{lem:weak2} is given below. 

\subparagraph{Proof of Lemma~\ref{lem:weak2} for case 1~(i.e., $u_{i} \in U$ represents an $\alpha$-heavy DBWT-repetition).}
In this case, node $u_{i} \in U$ represents an $\alpha$-heavy DBWT-repetition, 
and the node is split into nodes $u_{i^{\prime}}$ and $u_{i^{\prime}+1}$, 
where $U = \{ u_{1}, u_{2}, \ldots, u_{k} \}$. 
At the same time, $v_{i} \in V$ is split into two nodes $v_{i^{\prime}}$ and $v_{i^{\prime}+1}$, 
where $V = \{ v_{1}, v_{2}, \ldots, v_{k} \}$. 

We show that the sum of the excesses of all the nodes in $U$ is at most $2\alpha+2$ for LF-interval graph $\graphds(\hat{D}^{O(\alpha)}_{\delta})$. 
Let $\Lambda_{U}(u)$ be the excess of a given node $u \in U$ for LF-interval graph $\graphds(D_{\delta}^{O(\alpha)})$. 
Similarly, let $\Lambda^{\prime}_{U}(u^{\prime})$ be the excess of a given node $u^{\prime} \in U$ for the LF-interval graph $\graphds(\hat{D}^{O(\alpha)}_{\delta})$. 
The following lemma ensures that we can bound the excesses of nodes in $U$ for $\graphds(\hat{D}^{O(\alpha)}_{\delta})$ using 
the excesses of nodes in $U$ for $\graphds(D_{\delta}^{O(\alpha)})$. 
\begin{lemma}\label{lem:W_split_L_lemma1}
Let $u_{h} \in U$ be the node connected to $v_{i^{\prime}+1}$ by a directed edge in $E_{F}$. 
Then, the following two statements hold: 
(i) $\Lambda^{\prime}_{U}(u) \leq \Lambda_{U}(u) + \eqfn(u, u_{h})$ for all nodes $u \in \{ u_{1}, u_{2}, \ldots, u_{k} \} \setminus \{ u_{i} \}$. 
(ii) $\Lambda^{\prime}_{U}(u_{i^{\prime}}) + \Lambda^{\prime}_{U}(u_{i^{\prime}+1}) \leq \Lambda_{U}(u_{i}) - 1 + \eqfn(u_{i^{\prime}}, u_{h}) + \eqfn(u_{i^{\prime}+1}, u_{h})$ for any $\alpha \geq 2$. Here, $\eqfn$ is the function introduced in Appendix~\ref{app:update_edge_count}. 
\end{lemma}
\begin{proof}
(i) 
Let $m$ and $m^{\prime}$ be the number of directed edges pointing to node $u$ in LF-interval graphs $\graphds(D_{\delta}^{O(\alpha)})$ and $\graphds(\hat{D}^{O(\alpha)}_{\delta})$, respectively. 
$m^{\prime} = m$ if $u \neq u_{h}$. 
Otherwise, $m^{\prime} = m + 1$ because new node $v_{i^{\prime}+1}$ is connected to $u$ by a directed edge in $E_{F}$. Hence, $m^{\prime} = m + \eqfn(u, u_{h})$. 

Next, if $m + \eqfn(u, u_{h}) \geq \alpha$, 
then $\Lambda^{\prime}_{U}(u) = m + \eqfn(u, u_{h}) - (\alpha - 1)$ 
and $\Lambda_{U}(u) = m - (\alpha - 1)$. 
Otherwise, $\Lambda^{\prime}_{U}(u) = 0$ and $\Lambda_{U}(u) = 0$. 
Hence, $\Lambda^{\prime}_{U}(u) \leq \Lambda_{U}(u) + \eqfn(u, u_{h})$. 

(ii) 
Let $m \geq \alpha$ be the number of directed edges pointing to $u_{i}$ in $\graphds(D_{\delta}^{O(\alpha)})$. 
Similarly, let $m_{1}$ and $m_{2}$ be the numbers of directed edges pointing to $u_{i^{\prime}}$ and $u_{i^{\prime}+1}$, respectively, in $\graphds(\hat{D}^{O(\alpha)}_{\delta})$. 
For case 1, $m_{1} \geq 1$ and $m_{2} \geq 1$. 
Moreover, $m_{1} + m_{2} = m$ if $u_{i^{\prime}} \neq u_{h}$ and $u_{i^{\prime}+1} \neq u_{h}$. 
Otherwise, $m_{1} + m_{2} = m + 1$ because new node $v_{i^{\prime}+1}$ 
is connected to $u_{i^{\prime}}$ or $u_{i^{\prime}+1}$ by a directed edge in $E_{F}$. 
Hence, $m_{1} + m_{2} = m + \eqfn(u_{i^{\prime}}, u_{h}) + \eqfn(u_{i^{\prime}+1}, u_{h})$. 

We have $\Lambda^{\prime}_{U}(u_{i^{\prime}}) + \Lambda^{\prime}_{U}(u_{i^{\prime}+1}) = \max \{ 0, m_{1} - (\alpha - 1) \} + \max \{ 0, m_{2} - (\alpha - 1) \}$. 
The maximal value of $\Lambda^{\prime}_{U}(u_{i^{\prime}}) + \Lambda^{\prime}_{U}(u_{i^{\prime}+1})$ is $(m - 1 + \eqfn(u_{i^{\prime}}, u_{h}) + \eqfn(u_{i^{\prime}+1}, u_{h})) - (\alpha -1)$ by $m_{1} \geq 1$, $m_{2} \geq 1$, and $m_{1} + m_{2} = m + \eqfn(u_{i^{\prime}}, u_{h}) + \eqfn(u_{i^{\prime}+1}, u_{h})$. 
By contrast, $\Lambda_{U}(u_{i}) = m - (\alpha-1)$ by $m \geq \alpha$. 
Hence, $\Lambda^{\prime}_{U}(u_{i^{\prime}}) + \Lambda^{\prime}_{U}(u_{i^{\prime}+1}) \leq \Lambda_{U}(u_{i}) - 1 + \eqfn(u_{i^{\prime}}, u_{h}) + \eqfn(u_{i^{\prime}+1}, u_{h})$ by $\Lambda_{U}(u_{i}) = m - (\alpha-1)$. 
\end{proof}

The sum of the excesses of all nodes in $U$ for $\graphds(\hat{D}^{O(\alpha)}_{\delta})$ is at most $(\sum_{x = 1}^{k} \Lambda_{U}(u_{x})) + (\sum_{x = 1}^{k} \eqfn(u_{x}, u_{h})) + \eqfn(u_{i^{\prime}}, u_{h}) + \eqfn(u_{i^{\prime}+1}, u_{h}) - 1$ by Lemma~\ref{lem:W_split_L_lemma1}. 
We have $\sum_{x = 1}^{k} \Lambda_{U}(u_{x}) \leq 2 \alpha + 2$ because $\graphds(D_{\delta}^{O(\alpha)})$ is weakly $\alpha$-balanced. 
Furthermore, $(\sum_{x = 1}^{k} \eqfn(u_{x}, u_{h})) + \eqfn(u_{i^{\prime}}, u_{h}) + \eqfn(u_{i^{\prime}+1}, u_{h}) = 1$. 
Hence, the sum of the excesses of all nodes in $U$ for $\graphds(\hat{D}^{O(\alpha)}_{\delta})$ is at most $2 \alpha + 2$.

Similarly, we can show that the sum of the excesses of all nodes in $V$ for $\graphds(\hat{D}^{O(\alpha)}_{\delta})$ is at most $2 \alpha + 2$. 
Let $\Lambda_{V}(v)$ be the excess of a given node $v \in V$ for LF-interval graph $\graphds(D_{\delta}^{O(\alpha)})$. 
Similarly, let $\Lambda^{\prime}_{V}(v^{\prime})$ be the excess of a given node $v^{\prime} \in V$ for LF-interval graph $\graphds(\hat{D}^{O(\alpha)}_{\delta})$. 
The following lemma ensures that we can bound the excesses of nodes in $V$ for $\graphds(\hat{D}^{O(\alpha)}_{\delta})$ using 
the excesses of nodes in $V$ for $\graphds(D_{\delta}^{O(\alpha)})$.
\begin{lemma}\label{lem:W_split_L_lemma2}
Let $v_{h} \in V$ be the node connected to $u_{i^{\prime}+1}$ by a directed edge in $E_{L}$. 
Then, the following three statements hold for any $\alpha \geq 4$: 
(i) $\Lambda^{\prime}_{V}(v) \leq \Lambda_{V}(v)$ for all nodes $v \in \{ v_{1}, v_{2}, \ldots, v_{k} \} \setminus \{ v_{i}, v_{h} \}$; 
(ii) $\Lambda^{\prime}_{V}(v_{h}) = 0$;
(iii) $\Lambda^{\prime}_{V}(v_{i^{\prime}}) + \Lambda^{\prime}_{V}(v_{i^{\prime}+1}) \leq \Lambda_{V}(v_{i})$. 
\end{lemma}
\begin{proof}
(i) The number of directed edges pointing to $v$ is not changed by splitting node $u_{i}$, 
and hence $\Lambda^{\prime}_{V}(v) \leq \Lambda_{V}(v)$. 

(ii) The DBWT-repetition of $u_{i}$ is $L_{\delta}[p_{i}..p_{i+1}-1]$, 
and the F-interval of $v_{h}$ is $[\LF_{\delta}(p_{h})$, $\LF_{\delta}(p_{h+1}-1)]$. 
Interval $[\LF_{\delta}(p_{h}), \LF_{\delta}(p_{h+1}-1)]$ is contained in the interval $[p_{i}, p_{i+1}-1]$ of the DBWT-repetition in case 1 
for $\alpha \geq 4$. 
Nodes $u_{i^{\prime}}$ and $u_{i^{\prime}+1}$ represent the two DBWT-repetitions $L_{\delta}[p_{i}..x-1]$ and $L_{\delta}[x..p_{i+1}-1]$, 
respectively, where $p_{i} < x < p_{i+1}$ is an integer. 
Hence, each directed edge pointing to $v_{h}$ starts at $u_{i^{\prime}}$ or $u_{i^{\prime}+1}$ in LF-interval graph $\graphds(\hat{D}^{O(\alpha)}_{\delta})$. 
In other words, the number of directed edges pointing to $v_{h}$ is at most two in the LF-interval graph $\graphds(\hat{D}^{O(\alpha)}_{\delta})$. 
Therefore, $\Lambda^{\prime}_{V}(v_{h}) = 0$.

(iii) 
We use proof by contradiction to show that $v_{h}$ is neither $v_{i^{\prime}}$ nor $v_{i^{\prime}+1}$ for $\alpha \geq 4$. 
We assume that $v_{h} = v_{i^{\prime}}$ or $v_{h} = v_{i^{\prime}+1}$. 
Then the F-interval $[\LF_{\delta}(p_{i}), \LF_{\delta}(p_{i+1}-1)]$ of $v_{i}$ is properly contained in the interval $[p_{i}, p_{i+1}-1]$ of the DBWT-repetition represented by $u_{i}$~(i.e., $[\LF_{\delta}(p_{i}), \LF_{\delta}(p_{i+1}-1)] \subset [p_{i}, p_{i+1}-1]$). 
The length of the F-interval is shorter than that of the DBWT-repetition by $[\LF_{\delta}(p_{i}), \LF_{\delta}(p_{i+1}-1)] \subset [p_{i}, p_{i+1}-1]$, 
but the length of the F-interval must be equal to the length of the DBWT-repetition. 
Hence, $v_{h}$ is neither $v_{i^{\prime}}$ nor $v_{i^{\prime}+1}$. 

Because $v_{h} \not \in \{ v_{i^{\prime}}, v_{i^{\prime}+1} \}$, 
the following two statements hold for a directed edge $(u, v_{i^{\prime}})$ or $(u, v_{i^{\prime}+1})$ in $E_{L}$ stored in $\graphds(\hat{D}^{O(\alpha)}_{\delta})$: 
(i) $u \neq u_{i^{\prime} + 1}$, i.e., $u \in (\{ u_{1}, u_{2}, \ldots u_{k} \} \cup \{ u_{i^{\prime}} \}) \setminus \{ u_{i} \}$ 
and (ii) $\graphds(D^{O(\alpha)}_{\delta})$ has a directed edge $(u, v_{i})$ in $E_{L}$ if $u \neq u_{i^{\prime}}$; 
otherwise, it has a directed edge $(u_{i}, v_{i})$ in $E_{L}$. 
The two statements indicate that 
the number of directed edges pointing to $v_{i^{\prime}}$ or $v_{i^{\prime}+1}$ is no more than that of directed edges pointing to $v_{i}$ 
because we can map directed edges pointing to $v_{i^{\prime}}$ or $v_{i^{\prime}+1}$ in $\graphds(\hat{D}^{O(\alpha)}_{\delta})$, 
into distinct directed edges pointing to $v_{i}$ in $\graphds(D^{O(\alpha)}_{\delta})$. 
Hence, we obtain $\Lambda^{\prime}_{V}(v_{i^{\prime}}) + \Lambda^{\prime}_{V}(v_{i^{\prime}+1}) \leq \Lambda_{V}(v_{i})$. 

\end{proof}
The sum of the excesses of all nodes in $V$ for $\graphds(\hat{D}^{O(\alpha)}_{\delta})$ is at most $\sum_{x = 1}^{k} \Lambda_{V}(v_{x})$ by Lemma~\ref{lem:W_split_L_lemma2}. 
We have $\sum_{x = 1}^{k} \Lambda_{V}(v_{x}) \leq 2 \alpha + 2$ because $\graphds(D_{\delta}^{O(\alpha)})$ is weakly $\alpha$-balanced. 
Hence, the sum of the excesses of all nodes in $V$ for $\graphds(\hat{D}^{O(\alpha)}_{\delta})$ is at most $2 \alpha + 2$. 
Therefore, $\graphds(\hat{D}^{O(\alpha)}_{\delta})$ is weakly $\alpha$-balanced for all $\alpha \geq 4$. 

\subparagraph{Proof of Lemma~\ref{lem:weak2} for case 2~(i.e., $u_{i}$ does not represent an $\alpha$-heavy DBWT-repetition, and $v_{i}$ represents an $\alpha$-heavy F-interval).}
Case 2 is symmetric to case 1. 
Hence, we can prove Lemma~\ref{lem:weak2} for this case by modifying the proof of Lemma~\ref{lem:weak2} for case 1.

\section{Details for Section~\ref{sec:rcomp_algorithm}}\label{app:rcomp}
\oldsentence{
\subsection{Pseudo-code of r-comp}\label{app:pseudocode}
\subsection{Proof of Theorem~\ref{theo:rcomp_time}}\label{app:rcomp_time}
}
\subsection{Proof of Lemma~\ref{lem:modified_k_bound_lemma}}\label{app:modified_k_bound_lemma}
\subsubsection{Proof of Lemma~\ref{lem:modified_k_bound_lemma}-(i)}
Fast update operation $\fastUpdate(\graphds(D_{\delta}^{\alpha}),c)$, which is introduced in Section~\ref{sec:fast_update}, 
may merge nodes created by balancing operations for the first character $c$ of $T_{\delta+1}$. 
This merging makes it difficult to bound $k_{\splitop}$ by $O(r)$. 
The fast update operation has three cases, A, B, and C for merging nodes, 
and the operation has the possibility of merging nodes created by balancing operations only if 
the previous update operation~(i.e., $\update(\graphds(D_{\delta-1}^{\alpha}), T[n-\delta+1])$ or $\fastUpdate(\graphds(D_{\delta-1}^{\alpha}), T[n-\delta+1])$) 
did not split nodes in the split-node step. 
We change the two cases A and C into two new cases A', and C', respectively, using the result of the previous update operation, 
and we modify the procedure for the two cases so that it does not merge the nodes created by balancing operations. 

\subparagraph{Case A': (i) case A holds or (ii) case C holds, and the previous update operation did not split nodes in the split-node step.} 
In this case, the fast update operation executes the procedure for case A. 
That is, 
for the node $u_{i^{\prime}} \in U$ representing the input character $c$ of $\fastUpdate(\graphds(D_{\delta}^{\alpha}),c)$ 
and the node $u_{i^\prime-1} \in U$ previous to node $u_{i^{\prime}}$ in the doubly linked list of $U$, 
the consecutive nodes $u_{i^{\prime}-1}$ and $u_{i^{\prime}}$ are merged into new node $u_{y^{\prime}}$. 
Set $V$ is updated according to the merge of these nodes in $U$.

\subparagraph{Case C': (i) case C holds and (ii) the previous update operation split nodes in the split-node step.} 
In this case, the fast update operation executes the procedure for case C. 
That is, 
for the node $u_{i^{\prime}+1} \in U$ next to node $u_{i^{\prime}}$ in the doubly linked list of $U$, 
the three consecutive nodes $u_{i^{\prime}-1}$, $u_{i^{\prime}}$, and $u_{i^{\prime}+1}$ are merged into a new node $u_{y^{\prime}}$. 
Set $V$ is updated according to the merge of these nodes in $U$. 

One of the three cases A', B, and C' always holds 
because one of the three cases A, B, or C always holds. 
Lemma~\ref{lem:fast_update_result} still holds although the LF-interval graph outputted by the fast update operation is changed. 
This is because the procedure for case A, B, or C is executed. 
Hence, we obtain Lemma~\ref{lem:modified_k_bound_lemma}-(i).

\subsubsection{Proof of Lemma~\ref{lem:modified_k_bound_lemma}-(ii)}\label{app:update_count_bound}
R-comp creates the LF-interval graph for DBWT $D_{n}^{\alpha}$ via the LF-interval graphs 
for $(2n-3)$ DBWTs $D_{1}^{\alpha}, D_{2}^{2\alpha + 1}, D_{2}^{\alpha}, D_{3}^{2\alpha + 1}, \ldots, D_{n-1}^{\alpha}, D_{n}^{2\alpha + 1}$. 
The number $k$ of DBWT-repetitions in $D_{n}^{\alpha}$ is the sum of the difference 
between the size of each DBWT $D$ and that of the DBWT obtained by updating the LF-interval graph for $D$. 
Formally, $k = |D_{1}^{\alpha}| + (\sum_{\delta = 1}^{n-1} (|D_{\delta+1}^{2\alpha + 1}| - |D_{\delta}^{\alpha}|)) + (\sum_{\delta = 2}^{n} (|D_{\delta}^{\alpha}| - |D_{\delta}^{2\alpha + 1}|))$.
The term $\sum_{\delta = 2}^{n} (|D_{\delta}^{\alpha}| - |D_{\delta}^{2\alpha + 1}|)$ represents the number of DBWT-repetitions increased by balancing operations in the r-comp algorithm. 
Each iteration of a balancing operation removes a node from $U$ and inserts two nodes into $U$. 
This fact indicates $\sum_{\delta = 2}^{n} (|D_{\delta}^{\alpha}| - |D_{\delta}^{2\alpha + 1}|) = k_{\splitop}$ for 
the total number $k_{\splitop}$ of iterations executed in the balancing operations in the r-comp algorithm. 

Next, the term $|D_{\delta+1}^{2\alpha + 1}| - |D_{\delta}^{\alpha}|$ represents the number of DBWT-repetitions increased by the $\delta$-th update operation in the r-comp algorithm. 
We give the upper bound on the term using three values $r_{\delta}$, $\deltabwt_{\delta}$, and $\deltagrp_{\delta}$, which are explained below. 
For $1 \leq \delta \leq n$, $r_{\delta}$ is the number of BWT-runs in BWT $L_{\delta}$ of $T_{\delta}$. 
Clearly, $r_{1} = 1$ and $r_{n} = r$ for the number $r$ of BWT-runs in BWT $L_{n}$. 

Value $\deltabwt_{\delta} \in \{ 0, 1 \}$ is defined as $1$ if 
(i) $1 < \inspos < \delta+1$ and 
(ii) $L_{\delta+1}[\inspos - 1]$ and $L_{\delta+1}[\inspos + 1]$ are the same character 
for the position $\inspos$ of special character $\$$ in $L_{\delta+1}$. 
Otherwise, $\deltabwt_{\delta} = 0$. 
The following lemma ensures that the difference between $r_{\delta+1}$ and $r_{\delta}$ can be bounded by at most $1 + \deltabwt_{\delta}$. 

\begin{lemma}\label{lem:r_diff_bound}
We consider the following five cases for the position $\reppos$ of special character $\$$ in $L_{\delta}$: (i) $L_{\delta}[\reppos-1] = L_{\delta}[\reppos+1] = c$; 
(ii) $L_{\delta}[\reppos-1] = L_{\delta}[\reppos+1] \neq c$; 
(iii) $L_{\delta}[\reppos-1] \neq c$, $L_{\delta}[\reppos+1] \neq c$, and $L_{\delta}[\reppos-1] \neq L_{\delta}[\reppos+1]$; 
(iv) $L_{\delta}[\reppos-1] = c$ and $L_{\delta}[\reppos+1] \neq c$; 
and (v) $L_{\delta}[\reppos-1] \neq c$ and $L_{\delta}[\reppos+1] = c$. 
Here, $L_{\delta}[0]$ is defined as a character $c^{\prime}$ that does not appear in BWT $L_{n}$ for simplicity. 
Similarly, $L_{\delta}[\delta+1] = c^{\prime}$. 
The following five statements hold: 
$r_{\delta+1} - r_{\delta} =  -1 + \deltabwt_{\delta}$ for case (i); 
$r_{\delta+1} - r_{\delta} = 1 + \deltabwt_{\delta}$ for case (ii); 
$r_{\delta+1} - r_{\delta} = 1 + \deltabwt_{\delta}$ for case (iii);
$r_{\delta+1} - r_{\delta} = \deltabwt_{\delta}$ for case (iv);
$r_{\delta+1} - r_{\delta} = \deltabwt_{\delta}$ for case (v).
\end{lemma}
\begin{proof}
Let $L^{\prime}_{\delta}$ be the string obtained by replacing special character $\$$ from BWT $L_{\delta}$ with the first character $c$ of $T_{\delta+1}$, 
i.e., $L^{\prime}_{\delta} = L_{\delta}[1..\reppos-1] c L_{\delta}[\reppos+1..\delta]$. 
BWT $L_{\delta+1}$ can be created by inserting special character $\$$ into $L^{\prime}_{\delta}$ at position $\inspos$, 
and hence $r_{\delta+1} - m_{\delta} = 1 + \deltabwt_{\delta}$ for the number $m_{\delta}$ of runs in $L^{\prime}_{\delta}$. 
This fact indicates that $r_{\delta+1} - r_{\delta} = (1 + \deltabwt_{\delta}) + (m_{\delta} - r_{\delta})$.
For case (i), $r_{\delta+1} - r_{\delta} = -1 + \deltabwt_{\delta}$ because $m_{\delta} - r_{\delta} = -2$. 
For cases (ii) and (iii), $r_{\delta+1} - r_{\delta} = 1 + \deltabwt_{\delta}$ because $m_{\delta} - r_{\delta} = 0$. 
For cases (iv) and (v), $r_{\delta+1} - r_{\delta} = \deltabwt_{\delta}$ because $m_{\delta} - r_{\delta} = -1$.
\end{proof}

Value $\deltagrp_{\delta} \in \{ 0, 1 \}$ represents the number of split nodes in $U$ 
in the split-node step of the $\delta$-th update operation. 
That is, $\deltagrp_{\delta} = 1$ if the $\delta$-th update operation split nodes in the split-node step. 
Otherwise, $\deltagrp_{\delta} = 0$. 
The following lemma ensures that the difference between $|D_{\delta+1}^{2\alpha+1}|$ and $|D_{\delta}^{\alpha}|$ can be bounded by at most $1 + \deltagrp_{\delta}$. 

\begin{lemma}\label{lem:delta_bound_1}
The following three statements hold for $1 \leq \delta \leq n-1$: 
\begin{enumerate}
    \item $|D_{\delta+1}^{2\alpha+1}| - |D_{\delta}^{\alpha}| = \deltagrp_{\delta} + 1$ 
for $\graphds(D_{\delta+1}^{2\alpha+1}) = \update(\graphds(D_{\delta}^{\alpha}), c)$; 
    \item $|D_{\delta+1}^{2\alpha+1}| - |D_{\delta}^{\alpha}| = \deltagrp_{\delta}$ 
for $\graphds(D_{\delta+1}^{2\alpha+1}) = \fastUpdate(\graphds(D_{\delta}^{\alpha}), c)$ with case A' or B;
    \item $|D_{\delta+1}^{2\alpha+1}| - |D_{\delta}^{\alpha}| = \deltagrp_{\delta} - 1$ 
    for $\graphds(D_{\delta+1}^{2\alpha+1}) = \fastUpdate(\graphds(D_{\delta}^{\alpha}), c)$ with case C'.
\end{enumerate}
\end{lemma}
\begin{proof}
(1) The update operation $\update(\graphds(D_{\delta}^{\alpha}), c)$ removes $(1 + \deltagrp_{\delta})$ nodes from the set $U$ of nodes, 
and it inserts $(2 + 2\deltagrp_{\delta})$ nodes into $U$. 
Hence, $|D_{\delta+1}^{2\alpha+1}| - |D_{\delta}^{\alpha}| = \deltagrp_{\delta} + 1$. 

(2) The fast update operation with cases A' and B removes $(3 + \deltagrp_{\delta})$ nodes from $U$, 
and it inserts $(3 + 2\deltagrp_{\delta})$ nodes into $U$. 
Hence, $|D_{\delta+1}^{2\alpha+1}| - |D_{\delta}^{\alpha}| = \deltagrp_{\delta}$. 

(3) The fast update operation with case C' removes $(4 + \deltagrp_{\delta})$ nodes from $U$, 
and it inserts $(3 + 2\deltagrp_{\delta})$ nodes into $U$. 
Hence, $\graphds(D_{\delta+1}^{2\alpha+1}) = \fastUpdate(\graphds(D_{\delta}^{\alpha}), c)$. 
\end{proof}

Lemmas~\ref{lem:r_diff_bound} and \ref{lem:delta_bound_1} indicate that 
the difference between $|D_{\delta+1}^{2\alpha+1}|$ and $|D_{\delta}^{\alpha}|$ can be bounded using the difference between $r_{\delta+1}$ and $r_{\delta}$. 
The following lemma ensures that the difference between $|D_{\delta+1}^{2\alpha+1}|$ and $|D_{\delta}^{\alpha}|$ 
can be bounded using six values $r_{\delta+1}$, $r_{\delta}$, $\deltagrp_{\delta}$, $\deltagrp_{\delta-1}$, $\deltabwt_{\delta}$, and $\deltabwt_{\delta-1}$. 

\begin{lemma}\label{lem:modified_bound}
$|D_{\delta+1}^{2\alpha + 1}| \leq |D_{\delta}^{\alpha}| + (r_{\delta+1} - r_{\delta}) + (\deltagrp_{\delta} - \deltabwt_{\delta}) - (\deltagrp_{\delta-1} - \deltabwt_{\delta-1})$ for $2 \leq \delta \leq n-1$.
\end{lemma}
\begin{proof}

We show that  Lemma~\ref{lem:modified_bound} holds for any case used in Lemma~\ref{lem:r_diff_bound}.  
\subparagraph{Case (i).}
We have $\deltabwt_{\delta-1} = 1$, $\deltagrp_{\delta-1} \in \{ 0, 1\}$, 
and $r_{\delta+1} - r_{\delta} = -1 + \deltabwt_{\delta}$. 
For $\deltagrp_{\delta-1} = 1$, 
the $\delta$-th update operation is the fast update operation with case C'. 
By Lemma~\ref{lem:delta_bound_1}-(3), 
$|D_{\delta+1}^{2\alpha + 1}| - |D_{\delta}^{\alpha}| = (-1 + \deltabwt_{\delta}) + (\deltagrp_{\delta} - \deltabwt_{\delta})$. 
Because $r_{\delta+1} - r_{\delta} = -1 + \deltabwt_{\delta}$ and $\deltagrp_{\delta-1} - \deltabwt_{\delta-1} = 0$, 
we obtain $|D_{\delta+1}^{2\alpha + 1}| = |D_{\delta}^{\alpha}| + (r_{\delta+1} - r_{\delta}) + (\deltagrp_{\delta} - \deltabwt_{\delta}) - (\deltagrp_{\delta-1} - \deltabwt_{\delta-1})$. 


For $\deltagrp_{\delta-1} = 0$, 
the $\delta$-th update operation is the fast update operation with case A'. 
By Lemma~\ref{lem:delta_bound_1}-(2), 
$|D_{\delta+1}^{2\alpha+1}| - |D_{\delta}^{\alpha}| = (-1 + \deltabwt_{\delta}) + (\deltagrp_{\delta} - \deltabwt_{\delta}) - (-1)$. 
Because $r_{\delta+1} - r_{\delta} =  -1 + \deltabwt_{\delta}$ and $\deltagrp_{\delta-1} - \deltabwt_{\delta-1} = -1$, 
we obtain $|D_{\delta+1}^{2\alpha + 1}| = |D_{\delta}^{\alpha}| + (r_{\delta+1} - r_{\delta}) + (\deltagrp_{\delta} - \deltabwt_{\delta}) - (\deltagrp_{\delta-1} - \deltabwt_{\delta-1})$. 

\subparagraph{Case (ii).}
We have $\deltabwt_{\delta-1} = 1$, $r_{\delta+1} - r_{\delta} = 1 + \deltabwt_{\delta}$, 
and $\graphds(D_{\delta+1}^{2\alpha+1}) = \update(\graphds(D_{\delta}^{\alpha}), c)$. 
By Lemma~\ref{lem:delta_bound_1}-(1), 
$|D_{\delta+1}^{2\alpha + 1}| - |D_{\delta}^{\alpha}| = (1 + \deltabwt_{\delta}) + (\deltagrp_{\delta} - \deltabwt_{\delta})$. 
Because $r_{\delta+1} - r_{\delta} = 1 + \deltabwt_{\delta}$ and 
$-1 \leq \deltagrp_{\delta-1} - \deltabwt_{\delta-1} \leq 0$, 
we obtain $|D_{\delta+1}^{2\alpha + 1}| \leq |D_{\delta}^{\alpha}| + (r_{\delta+1} - r_{\delta}) + (\deltagrp_{\delta} - \deltabwt_{\delta}) - (\deltagrp_{\delta-1} - \deltabwt_{\delta-1})$. 

\subparagraph{Case (iii).}
We have $\deltabwt_{\delta-1} = 0$, $r_{\delta+1} - r_{\delta} = 1 + \deltabwt_{\delta}$, 
$\graphds(D_{\delta+1}^{2\alpha+1}) = \update(\graphds(D_{\delta}^{\alpha}), c)$, and $\deltagrp_{\delta-1} = 0$. 
By Lemma~\ref{lem:delta_bound_1}-(1), 
$|D_{\delta+1}^{2\alpha + 1}| - |D_{\delta}^{\alpha}| = (1 + \deltabwt_{\delta}) + (\deltagrp_{\delta} - \deltabwt_{\delta})$. 
Because $r_{\delta+1} - r_{\delta} = 1 + \deltabwt_{\delta}$ and 
$\deltagrp_{\delta-1} - \deltabwt_{\delta-1} = 0$, 
we obtain $|D_{\delta+1}^{2\alpha + 1}| = |D_{\delta}^{\alpha}| + (r_{\delta+1} - r_{\delta}) + (\deltagrp_{\delta} - \deltabwt_{\delta}) - (\deltagrp_{\delta-1} - \deltabwt_{\delta-1})$. 

\subparagraph{Case (iv).}
We have $\deltabwt_{\delta-1} = 0$, $r_{\delta+1} - r_{\delta} = \deltabwt_{\delta}$, 
$\graphds(D_{\delta+1}^{2\alpha+1}) = \fastUpdate(\graphds(D_{\delta}^{\alpha}), c)$, and $\deltagrp_{\delta-1} = 0$ 
The $\delta$-th update operation is the fast update operation with case A'.  
By Lemma~\ref{lem:delta_bound_1}-(2), 
$|D_{\delta+1}^{2\alpha+1}| - |D_{\delta}^{\alpha}| = \deltabwt_{\delta} + (\deltagrp_{\delta} - \deltabwt_{\delta})$. 
Because $r_{\delta+1} - r_{\delta} = \deltabwt_{\delta}$ and $\deltagrp_{\delta-1} - \deltabwt_{\delta-1} = 0$, 
we obtain $|D_{\delta+1}^{2\alpha + 1}| = |D_{\delta}^{\alpha}| + (r_{\delta+1} - r_{\delta}) + (\deltagrp_{\delta} - \deltabwt_{\delta}) - (\deltagrp_{\delta-1} - \deltabwt_{\delta-1})$. 

\subparagraph{Case (v).}
This case is symmetric to the fourth case, 
and hence we obtain $|D_{\delta+1}^{2\alpha + 1}| = |D_{\delta}^{\alpha}| + (r_{\delta+1} - r_{\delta}) + (\deltagrp_{\delta} - \deltabwt_{\delta}) - (\deltagrp_{\delta-1} - \deltabwt_{\delta-1})$. 
\end{proof}

We prove $k \leq r + k_{\splitop}$ using Lemma~\ref{lem:modified_bound}. 
We have $k = |D_{1}^{\alpha}| + (\sum_{\delta = 1}^{n-1} (|D_{\delta+1}^{2\alpha + 1}| - |D_{\delta}^{\alpha}|)) + k_{\splitop}$. 
By Lemma~\ref{lem:modified_bound}, 
we obtain $\sum_{\delta = 1}^{n-1} (|D_{\delta+1}^{2\alpha + 1}| - |D_{\delta}^{\alpha}|) \leq (|D_{2}^{2\alpha + 1}| - |D_{1}^{\alpha}|)  + (r_{n} - r_{2}) + (\deltagrp_{n-1} - \deltabwt_{n-1}) - (\deltagrp_{1} - \deltabwt_{1})$. 
Because $|D_{2}^{2\alpha + 1}| = 2$, $|D_{1}^{\alpha}| = 1$, $r_{2} = 2$, $\deltagrp_{1} - \deltabwt_{1} = 0$, and $r = r_{n}$, 
we obtain $\sum_{\delta = 1}^{n-1} (|D_{\delta+1}^{2\alpha + 1}| - |D_{\delta}^{\alpha}|) \leq r - 1 + (\deltagrp_{n-1} - \deltabwt_{n-1})$. 
$-1 \leq \deltagrp_{n-1} - \deltabwt_{n-1} \leq 0$ always holds because 
$\deltagrp_{n-1} = 0$ if $\deltabwt_{n-1} = 0$. 
Finally, $k \leq r + k_{\splitop}$. 

\subsubsection{Proof of Lemma~\ref{lem:modified_k_bound_lemma}-(iii)}
The modification of the fast update operation affects the total number $k_{\splitop}$ of iterations executed in the balancing operations in the r-comp algorithm. 
We show that the new value of $k_{\splitop}$ can be bounded by $\frac{2r}{\lceil \alpha / 2 \rceil - 7}$ for any constant $\alpha \geq 16$. 

Each iteration of balancing operation executes the procedure for case 1 or 2, which were introduced in Section~\ref{sec:balancing}. 
Let $k_{L}$ and $k_{F}$ be the numbers of iterations for cases 1 and 2 executed in the r-comp algorithm, respectively. 
The following lemma bounds $k_{L}$ and $k_{F}$. 
\begin{lemma}\label{lem:kL_and_kF}
The following two statements hold for any $\alpha \geq 16$: 
(i) $k_{L}, k_{F} \leq \frac{r}{\lceil \alpha / 2 \rceil - 7}$ if $k_{L} \leq k_{F}$ 
and (ii) $k_{L}, k_{F} \leq \frac{r}{\lceil \alpha / 2 \rceil - 7}$ if $k_{L} > k_{F}$.
\end{lemma}

Because $k_{\splitop}$ is the total number of iterations executed in the balancing operation in the r-comp algorithm, 
$k_{\splitop} = k_{L} + k_{F}$. Hence, Lemma~\ref{lem:modified_k_bound_lemma}-(iii) holds by Lemma~\ref{lem:kL_and_kF} 
and $k_{\splitop} = k_{L} + k_{F}$. 
In the remaining part, we prove Lemma~\ref{lem:kL_and_kF}.


\paragraph{Proof of Lemma~\ref{lem:kL_and_kF}-(i).}
We prove Lemma~\ref{lem:kL_and_kF}-(i) using an $\alpha$-\emph{partition}. 
An $\alpha$-\emph{partition} for a DBWT $D_{\delta}$ of BWT $L_{\delta}$ is a partition of $\{1, 2, \ldots, \delta \}$ 
into disjoint intervals $[w_{1}, w_{2}-1], [w_{2}, w_{3}-1], \ldots, [w_{d}, w_{d+1}-1]$~($w_{1} = 1 < w_{2} < w_{3} < \cdots < w_{d+1} = \delta+1$) such that each interval is called $\alpha$-\emph{big-interval}. 
The size of the partition is denoted as $d$, 
and each $\alpha$-big interval $[w_{x}, w_{x+1}-1]$ ensures that 
interval $[w_{x}, w_{x+1}-1]$ covers 
at least $(\lceil \alpha / 2 \rceil - 5)$ starting positions of the DBWT-repetitions in $D_{\delta}$. 
The following lemma ensures that 
there exists an $\alpha$-partition of size $k_{F}$ for 
the LF-interval graph $\graphds(D_{n}^{\alpha})$ created by the $(n-1)$-th balancing operation of r-comp. 

\begin{lemma}\label{lem:prt_f}
There exists an $\alpha$-partition of size $k_{F}$ for DBWT $D_{n}^{\alpha}$. 
\end{lemma}
Lemma~\ref{lem:prt_f} will be proved later. 
This lemma indicates that $D_{n}^{\alpha}$ consists of 
at least $k_{F} (\lceil \alpha / 2 \rceil - 5)$ DBWT-repetitions. 
By contrast, $D_{n}^{\alpha}$ consists of at most $(r + k_{L} + k_{F})$ DBWT-repetitions by Lemma~\ref{lem:modified_k_bound_lemma}-(ii)~(i.e., $k \leq r + k_{\splitop}$) and $k_{\splitop} = k_{L} + k_{F}$. 
We have $k_{F} (\lceil \alpha / 2 \rceil - 5) \leq r + 2 k_{F}$ for $k_{L} \leq k_{F}$, 
and thus we obtain inequality $k_{F} \leq \frac{r}{\lceil \alpha / 2 \rceil - 7}$. 
Hence, $k_{L}, k_{F} \leq \frac{r}{\lceil \alpha / 2 \rceil - 7}$ by $k_{L} \leq k_{F}$, i.e., 
Lemma~\ref{lem:kL_and_kF}-(i) holds. 
In the next paragraphs, we prove Lemma~\ref{lem:prt_f}. 

\subparagraph{Proof of Lemma~\ref{lem:prt_f}.}
We formally define $\alpha$-big intervals for proving Lemma~\ref{lem:prt_f}. 
For a DBWT $D_{\delta} = L_{\delta}[p_{1}..(p_{2}-1)]$, $L_{\delta}[p_{2}..(p_{3}-1)]$, $\ldots$, $L_{\delta}[p_{k}..(p_{k+1}-1)]$ of BWT $L_{\delta}$, 
we classify each DBWT-repetition as type A or type B. 
A DBWT-repetition $L_{\delta}[p_{x}..p_{x+1}-1]$ is type A if 
either of the following two conditions holds: 
(i) The DBWT-repetition represents special character $\$$ in the BWT~(i.e., $p_{x} = \reppos$ for the position $\reppos$ of special character $\$$ in BWT $L_{\delta}$). 
(ii) The DBWT-repetition is next to special character $\$$ in $L_{\delta}$~(i.e., $p_{x} = \reppos + 1$), 
and the $(\delta-1)$-th update operation split nodes in the split-node step. 
If the DBWT-repetition is not type A, then it is type B. 
Clearly, the number of DBWT-repetitions of type A in a DBWT is at most two, 
and the other DBWT-repetitions are type B. 

We say that an interval $[w, w^{\prime}]$ on $F_{\delta}$ \emph{inner covers} the starting position $p_{x}$ of a DBWT-repetition $L_{\delta}[p_{x}..p_{x+1}-1]$ of type B if either of the following two conditions holds: 
(i) $p_{x} \neq \reppos + 1$ and $p_{x} \in [w + 1, w^{\prime}]$; 
(ii) $p_{x} = \reppos + 1$ and $p_{x} \in [w + 2, w^{\prime}]$. 
An interval $[w, w^{\prime}]$ on $F_{\delta}$ is called $\alpha$-big interval if 
the interval satisfies the following two conditions: 
(i) the interval inner covers at least $(\lceil \alpha / 2 \rceil - 5)$ starting positions of DBWT-repetitions of type B in DBWT $D_{\delta}$; 
(ii) $w = 1$ or $D_{\delta}$ has a DBWT-repetition of type B such that its F-interval starts at position $w$ on $F_{\delta}$. 


The basic idea behind the proof of Lemma~\ref{lem:prt_f} is to create an $\alpha$-partition for the LF-interval graph outputted by an update operation or a balancing operation 
using the $\alpha$-partition for the input LF-interval graph of the operation. 
Because r-comp consists of update and balancing operations, 
we can obtain an $\alpha$-partition for DBWT $D_{n}^{\alpha}$ by repeating this procedure. 

Recall that the input of the $\delta$-th iteration of r-comp is the LF-interval graph $\graphds(D_{\delta}^{\alpha})$ for DBWT $D_{\delta}^{\alpha}$ 
of BWT $L_{\delta}$ and the first character $c$ of suffix $T_{\delta+1}$. 
At the $\delta$-th iteration, update operation $\update(\graphds(D_{\delta}^{\alpha}), c)$ or fast update operation $\fastUpdate(D_{\delta}^{\alpha}), c)$ 
outputs the LF-interval graph $\graphds(D_{\delta+1}^{2\alpha+1})$ of DBWT $D_{\delta+1}^{2\alpha+1}$, 
and balancing operation $\balance(\graphds(D_{\delta+1}^{2\alpha+1}))$ outputs LF-interval graph $\graphds(D_{\delta+1}^{\alpha})$ of $D_{\delta+1}^{\alpha}$. 
Let $k_{F, \delta}$ be the number of iterations for case 2 in the $\delta$-th balancing operation $\balance(\graphds(D_{\delta+1}^{2\alpha+1}))$ executed in the r-comp algorithm. 
The following lemma ensures that we can create $\alpha$-partitions of large size for $D_{\delta}^{\alpha}$ and $D_{\delta+1}^{2\alpha+1}$. 

\begin{lemma}\label{lem:prt_f_sub}
The following two statements hold: 
(i) If there exists an $\alpha$-partition of size $d$ for $D_{\delta+1}^{2\alpha+1}$, 
then there exists an $\alpha$-partition of size $(d+k_{F, \delta})$ for $D_{\delta}^{\alpha+1}$. 
(ii) If there exists an $\alpha$-partition of size $d$ for $D_{\delta}^{\alpha}$, 
then there exists an $\alpha$-partition of size $d$ for $D_{\delta+1}^{2\alpha+1}$.
\end{lemma}
Lemma~\ref{lem:prt_f_sub} will be proved later. 
By this lemma, 
there exists an $\alpha$-partition of size $d + k_{F, \delta}$ for $D_{\delta}^{\alpha+1}$ 
if there exists an $\alpha$-partition of size $d$ for $D_{\delta}^{\alpha}$. 
Because DBWT $D_{1}^{\alpha}$ has an $\alpha$-partition of size $0$, 
there exists an $\alpha$-partition of size $\sum_{\delta = 1}^{n-1} k_{F, \delta}$ for $D_{n}^{\alpha}$. 
Because $\sum_{\delta = 1}^{n-1} k_{F, \delta} = k_{F}$, 
there exists an $\alpha$-partition of size $k_{F}$ for $D_{n}^{\alpha}$, i.e., Lemma~\ref{lem:prt_f} holds. 
In the next paragraphs, we prove Lemma~\ref{lem:prt_f_sub}.

\subparagraph{Proof of Lemma~\ref{lem:prt_f_sub}-(i).}
\begin{figure}[t]
     \begin{center}
		\includegraphics[scale=0.65]{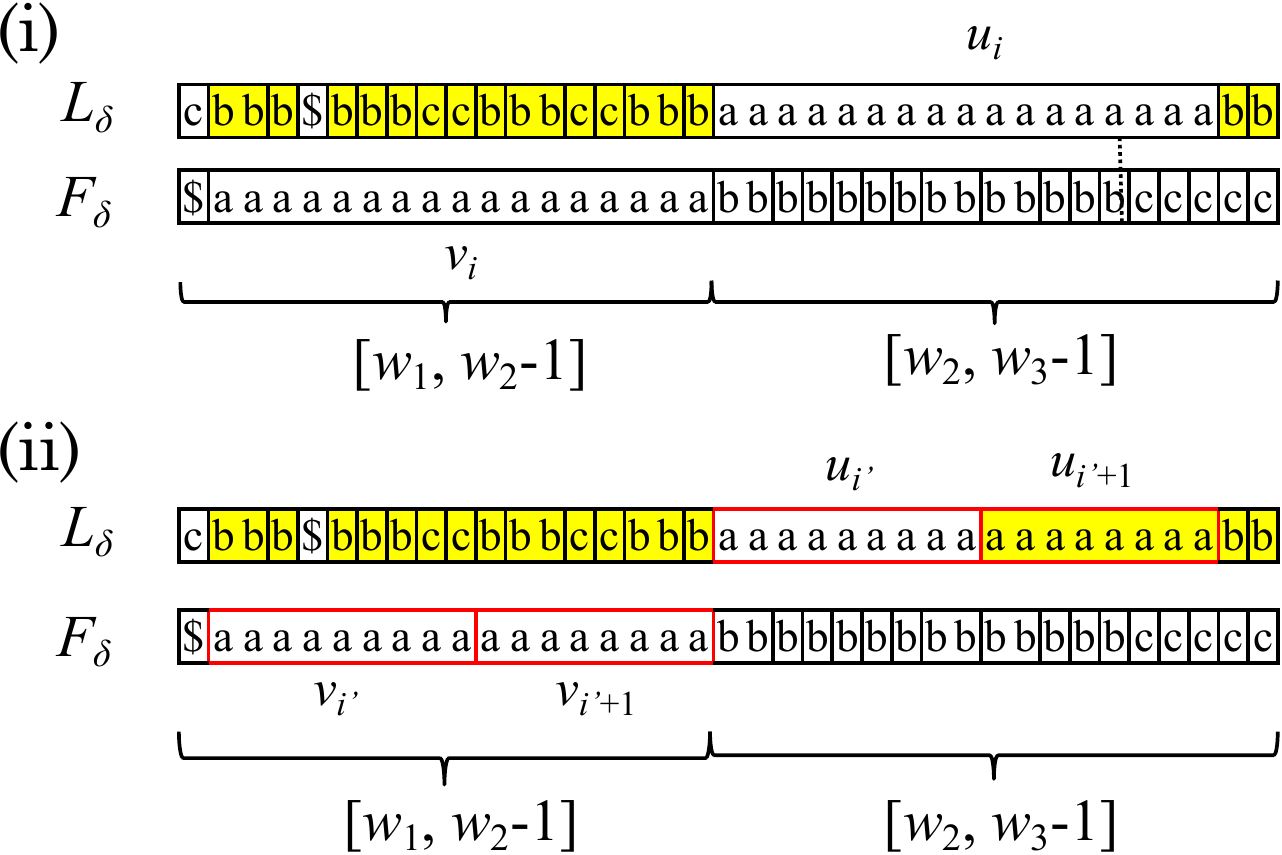}
		\includegraphics[scale=0.65]{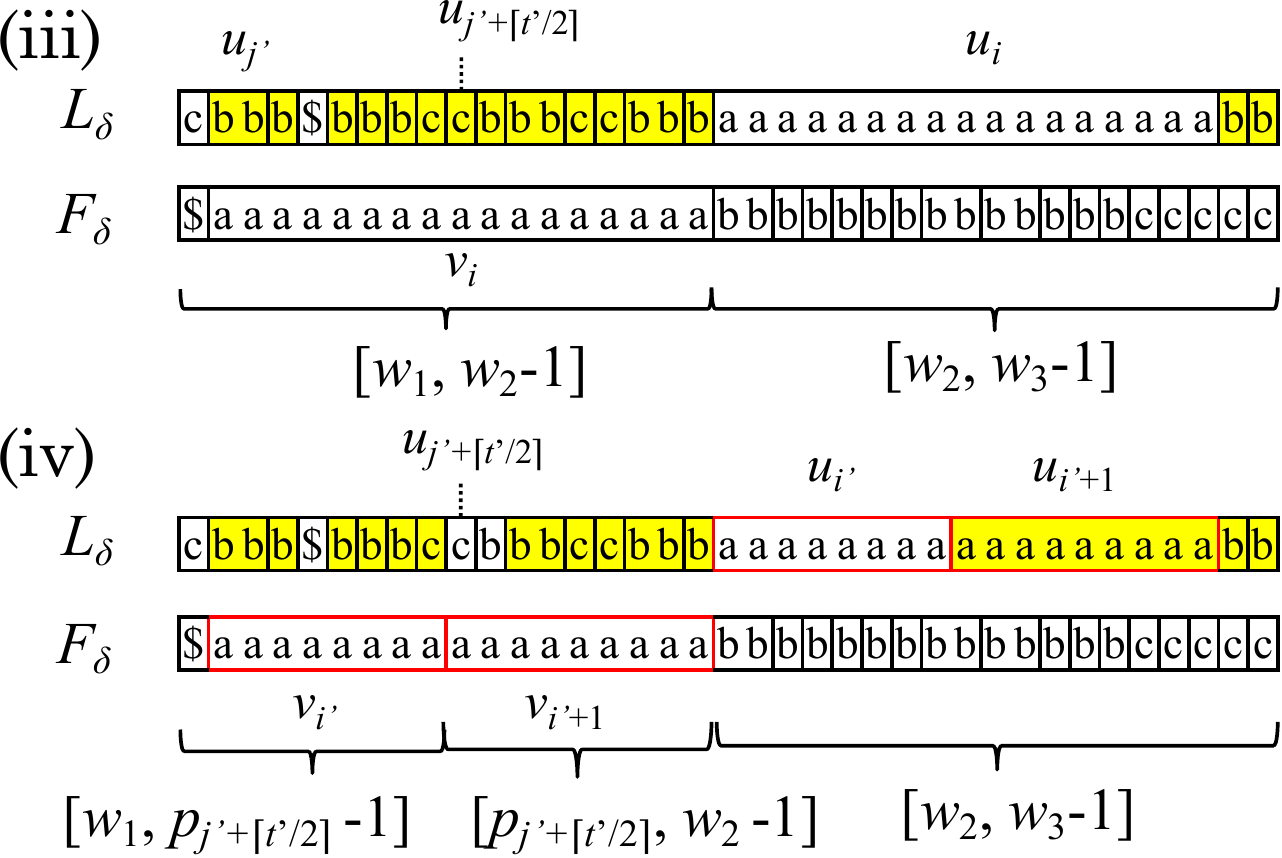}
     \end{center}
	  \caption{
	  $\alpha$-partitions 
	  for  DBWTs (i) $D_{\delta+1}^{O(\alpha)}$ and (ii) $\hat{D}^{O(\alpha)}_{\delta+1}$ in the proof of Lemma~\ref{lem:alpha_partition_for_case_1}.  
      $\alpha$-partitions 
	  for DBWTs (iii) $D_{\delta+1}^{O(\alpha)}$ and (iv) $\hat{D}^{O(\alpha)}_{\delta+1}$ in the proof of Lemma~\ref{lem:alpha_partition_for_case_2}. 
	  Yellow rectangles represent DBWT-repetitions of type B such that each DBWT-repetition is inner covered by an $\alpha$-big interval on $F_{\delta}$.
	  }
 \label{fig:partition_split_case}
\end{figure}

Each iteration of balancing operation $\balance(\graphds(D_{\delta+1}^{2\alpha+1}))$ takes an LF-interval graph $\graphds(D_{\delta+1}^{O(\alpha)})$ as input and 
outputs an LF-interval graph $\graphds(\hat{D}^{O(\alpha)}_{\delta+1})$. 
The following lemma ensures that 
any $\alpha$-big interval $[w, w^{\prime}]$ for $D_{\delta+1}^{O(\alpha)}$ is 
an $\alpha$-big interval for $\hat{D}^{O(\alpha)}_{\delta+1}$.

\begin{lemma}\label{lem:partition_lemma}
The following three statements hold: 
(i) If DBWT $D_{\delta+1}^{O(\alpha)}$ has a DBWT-repetition of type B starting at a position $x$ in $L_{\delta+1}$, 
then DBWT $\hat{D}^{O(\alpha)}_{\delta+1}$ has a DBWT-repetition of type B starting at position $x$ in $L_{\delta+1}$. 
(ii) If DBWT $D_{\delta+1}^{O(\alpha)}$ has a DBWT-repetition of type B such that its F-interval starts at a position $x$ in $F_{\delta+1}$, 
then DBWT $\hat{D}^{O(\alpha)}_{\delta+1}$ has a DBWT-repetition of type B such that its F-interval 
starts at position $x$ in $F_{\delta+1}$. 
(iii) Any $\alpha$-big interval $[w, w^{\prime}]$ for $D_{\delta+1}^{O(\alpha)}$ is 
an $\alpha$-big interval for $\hat{D}^{O(\alpha)}_{\delta+1}$.
\end{lemma}
\begin{proof}
(i) We can obtain DBWT $D_{\delta+1}^{O(\alpha)}$ by concatenating two DBWT-repetitions in DBWT $\hat{D}^{O(\alpha)}_{\delta+1}$. This fact indicates that Lemma~\ref{lem:partition_lemma}-(i) holds. 
(ii) Lemma~\ref{lem:partition_lemma}-(ii) follows from Lemma~\ref{lem:partition_lemma}-(i). 
(iii) 
The DBWT $\hat{D}^{O(\alpha)}_{\delta+1}$ has an DBWT-repetition of type B 
such that its F-interval starts at position $w$ on $F_{\delta+1}$ by Lemma~\ref{lem:partition_lemma}-(ii). 
Let $X$ be the set of starting positions of DBWT-repetitions of type B for $D_{\delta+1}^{O(\alpha)}$ such that 
the starting position of each DBWT-repetition is inner covered by the interval $[w, w^{\prime}]$. 
Then, the interval $[w, w^{\prime}]$ inner covers at least $|X|$ starting positions of DBWT-repetitions of type B 
for $\hat{D}^{O(\alpha)}_{\delta+1}$ by Lemma~\ref{lem:partition_lemma}-(i). 
Hence, the interval $[w, w^{\prime}]$ is an $\alpha$-big interval for $\hat{D}^{O(\alpha)}_{\delta+1}$ by $|X| \geq (\lceil \alpha / 2 \rceil - 5)$. 
\end{proof}

Let $d$ be the size of an $\alpha$-partition for $D_{\delta+1}^{O(\alpha)}$. 
The following lemma ensures that 
there exists an $\alpha$-partition of size $d$ for $\hat{D}^{O(\alpha)}_{\delta+1}$ 
if LF-interval graph $\graphds(\hat{D}^{O(\alpha)}_{\delta+1})$ is created by applying the procedure for case 1.

\begin{lemma}\label{lem:alpha_partition_for_case_1}
If LF-interval graph $\graphds(\hat{D}^{O(\alpha)}_{\delta+1})$ is created by 
applying the procedure for case 1 to LF-interval graph $\graphds(D_{\delta+1}^{O(\alpha)})$, 
then there exists an $\alpha$-partition of size $d$ for $\hat{D}^{O(\alpha)}_{\delta+1}$. 
\end{lemma}
\begin{proof}
Let $[w_{1}, w_{2}-1], [w_{2}, w_{3}-1], \ldots, [w_{d}, w_{d+1}-1]$ be an $\alpha$-partition of size $d$ for $D_{\delta+1}^{O(\alpha)}$. 
Then, each interval is an $\alpha$-big interval for $\hat{D}^{O(\alpha)}_{\delta+1}$ by Lemma~\ref{lem:partition_lemma}-(iii). 
Therefore, sequence $[w_{1}, w_{2}-1], [w_{2}, w_{3}-1], \ldots, [w_{d}, w_{d+1}-1]$ is an $\alpha$-partition for $\hat{D}^{O(\alpha)}_{\delta+1}$. 
Figures~\ref{fig:partition_split_case}-(i) and -(ii) illustrate the $\alpha$-partitions for $D_{\delta+1}^{O(\alpha)}$ and $\hat{D}^{O(\alpha)}_{\delta+1}$, respectively. 
\end{proof}

Similarly, the following lemma ensures that 
we can create an $\alpha$-partition of size $d+1$ for $\hat{D}^{O(\alpha)}_{\delta+1}$ 
if LF-interval graph $\graphds(\hat{D}^{O(\alpha)}_{\delta+1})$ is created by applying the procedure for case 2.

\begin{lemma}\label{lem:alpha_partition_for_case_2}
If LF-interval graph $\graphds(\hat{D}^{O(\alpha)}_{\delta+1})$ is created by 
applying the procedure for case 2 to LF-interval graph $\graphds(D_{\delta+1}^{O(\alpha)})$, 
then there exists an $\alpha$-partition of size $(d+1)$ for $\hat{D}^{O(\alpha)}_{\delta+1}$.
\end{lemma}
\begin{proof}
The balancing operation splits a node $v_{i} \in V$ into new nodes $v_{i^{\prime}}$ and $v_{i^{\prime}+1}$ in the LF-interval graph $\graphds(D_{\delta+1}^{O(\alpha)})$ 
for DBWT $D_{\delta+1}^{O(\alpha)} = L_{\delta}[p_{1}..(p_{2}-1)]$, $L_{\delta}[p_{2}..(p_{3}-1)]$, $\ldots$, $L_{\delta}[p_{k}..(p_{k+1}-1)]$. 
The node $v_{i}$ is connected to node $u_{i} \in U$ by an edge in $E_{LF}$, 
and node $u_{i}$ is split into new nodes $u_{i^{\prime}}$ and $u_{i^{\prime}+1}$. 
For case 2, node $u_{j^{\prime}} \in U$ is the most forward node in the doubly linked list of $U$ of the nodes connected to $v_{i}$ by directed edges in set $E_{L}$, 
and $t^{\prime} \geq \alpha$ is the number of directed edges pointing to $v_{i}$. 
Let $[w_{1}, w_{2}-1], [w_{2}, w_{3}-1], \ldots, [w_{d}, w_{d+1}-1]$ be an $\alpha$-partition of size $d$ for $D_{\delta+1}^{O(\alpha)}$. 
Then, there exists an $\alpha$-big interval $[w_{q}, w_{q+1}-1]$ containing the F-interval $[\LF_{\delta+1}(p_{i}), \LF_{\delta+1}(p_{i+1}-1)]$ of $v_{i}$~(i.e., $[\LF_{\delta+1}(p_{i}), \LF_{\delta+1}(p_{i+1}-1)] \subseteq [w_{q}, w_{q+1}-1]$). 
The $\alpha$-big interval $[w_{q}, w_{q+1}-1]$ covers the starting position $p_{j^{\prime} + \lceil t^{\prime} / 2 \rceil }$ of the DBWT-repetition represented as node $u_{j^{\prime} + \lceil t^{\prime} / 2 \rceil} \in U$.

We show that $(d+1)$ intervals $[w_{1}, w_{2}-1], [w_{2}, w_{3}-1], \ldots, [w_{q-1}, w_{q}-1], [w_{q}, p_{j^{\prime} + \lceil t^{\prime} / 2 \rceil } - 1], [p_{j^{\prime} + \lceil t^{\prime} / 2 \rceil }, w_{q + 1}-1], [w_{q+1}, w_{q+2}-1], \ldots, [w_{d}, w_{d+1}-1]$ are an $\alpha$-partition for $\hat{D}^{O(\alpha)}_{\delta+1}$. 
Here, intervals $[w_{q}, p_{j^{\prime} + \lceil t^{\prime} / 2 \rceil } - 1]$ and 
$[p_{j^{\prime} + \lceil t^{\prime} / 2 \rceil }, w_{q + 1}-1]$ can be obtained 
by splitting $\alpha$-big interval $[w_{q}, w_{q+1}-1]$ at position $p_{j^{\prime} + \lceil t^{\prime} / 2 \rceil }$. 
Figures~\ref{fig:partition_split_case}-(iii) and -(iv) illustrate two $\alpha$-partitions for $D_{\delta+1}^{O(\alpha)}$ and $\hat{D}^{O(\alpha)}_{\delta+1}$, respectively.

For all $x \in \{ 1, 2, \ldots, d \} \setminus \{ q \}$, 
interval $[w_{x}, w_{x}-1]$ is an $\alpha$-big interval 
for $\hat{D}^{O(\alpha)}_{\delta+1}$ by Lemma~\ref{lem:partition_lemma}-(iii). 
The interval $[w_{q}, p_{j^{\prime} + \lceil t^{\prime} / 2 \rceil } - 1]$ covers 
at least $\lceil t^{\prime} / 2 \rceil$ starting positions of DBWT-repetitions for $D_{\delta+1}^{O(\alpha)}$. 
In the DBWT-repetitions, at least $(\lceil t^{\prime} / 2 \rceil -4)$ DBWT-repetitions are 
type B and the starting positions of the DBWT-repetitions are inner covered 
by the interval $[w_{q}, p_{j^{\prime} + \lceil t^{\prime} / 2 \rceil } - 1]$. 
$\lceil t^{\prime} / 2 \rceil -4 \geq \lceil \alpha / 2 \rceil - 5$ 
because $\lceil t^{\prime} / 2 \rceil \geq \lceil \alpha / 2 \rceil$. 
Hence, $[w_{q}, p_{j^{\prime} + \lceil t^{\prime} / 2 \rceil } - 1]$ is an $\alpha$-big interval 
for $D_{\delta+1}^{O(\alpha)}$, 
and the interval is an $\alpha$-big interval for $\hat{D}^{O(\alpha)}_{\delta+1}$ by Lemma~\ref{lem:partition_lemma}-(iii).

Similarly, interval $[p_{j^{\prime} + \lceil t^{\prime} / 2 \rceil }, w_{q + 1}-1]$ inner covers 
at least $(t^{\prime} - \lceil t^{\prime} / 2 \rceil - 4)$ starting positions of DBWT-repetitions of type B for $D_{\delta+1}^{O(\alpha)}$. 
$t^{\prime} - \lceil t^{\prime} / 2 \rceil - 4 \geq \lceil \alpha / 2 \rceil - 5$ 
because $t^{\prime} - \lceil t^{\prime} / 2 \rceil \geq \lceil \alpha / 2 \rceil - 1$. 
Hence, the interval inner covers at least $(\lceil \alpha / 2 \rceil - 5)$ starting positions of DBWT-repetitions 
for $\hat{D}^{O(\alpha)}_{\delta+1}$ by Lemma~\ref{lem:partition_lemma}-(ii). 

Next, node $u_{i^{\prime}+1}$ represents a DBWT-repetition of type B, 
and its F-interval starts at position $p_{j^{\prime} + \lceil t^{\prime} / 2 \rceil }$. 
The position is the starting position of interval $[p_{j^{\prime} + \lceil t^{\prime} / 2 \rceil }, w_{q + 1}-1]$ on $F_{\delta+1}$, 
and hence the interval is an $\alpha$-big interval for $\hat{D}^{O(\alpha)}_{\delta+1}$. 
Therefore, intervals $[w_{1}, w_{2}-1], [w_{2}, w_{3}-1], \ldots, [w_{q-1}, w_{q}-1], [w_{q}, p_{j^{\prime} + \lceil t^{\prime} / 2 \rceil } - 1], [p_{j^{\prime} + \lceil t^{\prime} / 2 \rceil }, w_{q + 1}-1], [w_{q+1}, w_{q+2}-1], \ldots, [w_{d}, w_{d+1}-1]$ form an $\alpha$-partition for $\hat{D}^{O(\alpha)}_{\delta+1}$. 
\end{proof}

Finally, we obtain Lemma~\ref{lem:prt_f_sub}-(i) by Lemmas~\ref{lem:alpha_partition_for_case_1} and \ref{lem:alpha_partition_for_case_2}.  

\subparagraph{Proof of Lemma~\ref{lem:prt_f_sub}-(ii).}

The basic idea behind the proof of Lemma~\ref{lem:prt_f_sub}-(ii) is similar to the proof of Lemma~\ref{lem:prt_f_sub}-(i). 
We can map DBWT-repetitions of type B in DBWT $D_{\delta}^{\alpha}$ into 
distinct DBWT-repetitions of type B in DBWT $D_{\delta+1}^{2\alpha + 1}$. 
For an $\alpha$-partition for $D_{\delta+1}^{2\alpha + 1}$, 
we map a given $\alpha$-big interval for $D_{\delta}^{\alpha}$ into an $\alpha$-big interval for $D_{\delta+1}^{2\alpha + 1}$ 
using the mapping for DBWT-repetitions of type B. 

Function $\bmap_{L}$ maps the given starting position $p_{x}$ of a DBWT-repetition $L_{\delta}[p_{x}..p_{x+1}-1]$ of type B in DBWT $D_{\delta}^{\alpha} = L_{\delta}[p_{1}..(p_{2}-1)]$, $L_{\delta}[p_{2}..(p_{3}-1)]$, $\ldots$, $L_{\delta}[p_{k}..(p_{k+1}-1)]$ 
into the starting position of a DBWT-repetition of type B in DBWT $D_{\delta+1}^{2\alpha + 1}$. 
Formally, 
we consider two cases: $p_{x} = \reppos + 1$ or not. 
For $p_{x} = \reppos + 1$, 
if the $\delta$-th update operation is the fast update operation with case B, 
then $\bmap_{L}(\reppos + 1) = \reppos + \shfn(\reppos)$; otherwise, $\bmap_{L}(\reppos + 1) = (\reppos + 1) + \shfn(\reppos + 1)$. 
Here, $\shfn$ is the function introduced in Appendix~\ref{app:computingedge}~(i.e., $\shfn(x) = 0$ for $x \in \{ 1, 2, \ldots, \delta \}$ and for the position $\inspos$ of special character $\$$ in $L_{\delta+1}$ if $x < \inspos$; otherwise $\shfn(x) = 1$).
Similarly, for $p_{x} \neq \reppos + 1$, $\bmap_{L}(p_{x}) = p_{x} + \shfn(p_{x})$. 
The following two lemmas ensure that this function returns the starting position of a DBWT-repetition of type B in DBWT $D_{\delta+1}^{2\alpha + 1}$. 

\begin{lemma}\label{lem:bmap_L_special}
For a DBWT-repetition $L_{\delta}[p_{x}..p_{x+1}-1]$ of type B with $p_{x} = \reppos + 1$, 
$\bmap_{L}(p_{x})$ is the starting position of a DBWT-repetition of type B in DBWT $D_{\delta+1}^{2\alpha + 1}$. 
\end{lemma}
\begin{proof}

Before the insert-node step of the $\delta$-th update operation is executed, 
the given LF-interval graph has node $u_{i^{\prime}} \in U$, which represents the input character $c$ of the update operation, 
and node $u_{i^{\prime}+1} \in U$, which is next to $u_{i^{\prime}}$ in the doubly linked list of $U$. 
If this update operation splits a node $u_{j} \in U$ into two nodes $u_{j^{\prime}}$ and $u_{j^{\prime}+1}$ in the split-node step, 
then the new node $u_{x^{\prime}} \in U$ representing special character $\$$ is inserted into the doubly linked list of $U$ at the position previous to $u_{j^{\prime}+1}$. 
Otherwise, $u_{x^{\prime}} \in U$ is inserted into in the doubly linked list of $U$ at the position previous to a node $u \neq u_{j^{\prime}+1}$. 
Hence, for LF-interval graph $\graphds(D_{\delta+1}^{2\alpha + 1})$, 
a node in $U$ represents the DBWT-repetition of type A only if 
the node is $u_{x^{\prime}}$ or $u_{j^{\prime}+1}$. 

We consider two cases: (i) the $\delta$-th update operation is the fast update operation with case B or (ii) not. 
For case (i), 
two nodes $u_{i^{\prime}}$ and $u_{i^{\prime}+1}$ are merged into a new node $u_{y^{\prime}}$, 
which represents a DBWT-repetition starting at position $\reppos + \shfn(\reppos)$ in $L_{\delta+1}$. 
Because $u_{y^{\prime}} \not \in \{ u_{x^{\prime}}, u_{j^{\prime}+1} \}$, 
DBWT $D_{\delta+1}^{2\alpha + 1}$ has a DBWT-repetition of type B starting at position $\bmap_{L}(\reppos + 1) = \reppos + \shfn(\reppos)$. 

For case (ii), 
the $\delta$-th update operation is $\update(\graphds(D_{\delta}^{\alpha}), c)$ or the fast update operation with case A'. 
LF-interval graph $\graphds(D_{\delta}^{\alpha})$ has the node $u_{i+1} \in U$ representing the DBWT-repetition starting at position $\reppos + 1$ in $L_{\delta}$. 
Clearly, the DBWT-repetition represented by $u_{i+1}$ is not special character $\$$. 
In this case, $u_{i+1}$ satisfies one of the following three conditions: 
(1) $u_{i+1}$ is not removed from $U$ by the split-node step, 
and $u_{x^{\prime}}$ is not previous to $u_{i+1}$ in the doubly linked list of $U$; 
(2) $u_{i+1}$ is not removed from $U$ by the split-node step, 
and $u_{x^{\prime}}$ is previous to $u_{i+1}$ in the doubly linked list of $U$; 
(3) $u_{i+1}$ is split into two new nodes $u_{j^{\prime}}$ and $u_{j^{\prime}+1}$. 
The three conditions ensures that 
LF-interval graph $\graphds(D_{\delta+1}^{2\alpha + 1})$ has $u_{i+1}$ or $u_{j^{\prime}}$, 
and the two nodes represent the DBWT-repetition starting at position $(\reppos + 1) + \shfn(\reppos + 1)$ in $L_{\delta+1}$. 
Because $u_{i+1}, u_{j^{\prime}} \not \in \{ u_{x^{\prime}}, u_{j^{\prime}+1} \}$, 
$D_{\delta+1}^{2\alpha + 1}$ has a DBWT-repetition of type B starting at position $\bmap_{L}(\reppos + 1) = (\reppos + 1) + \shfn(\reppos + 1)$. 

\end{proof}

\begin{lemma}\label{lem:bmap_L}
For a DBWT-repetition $L_{\delta}[p_{x}..p_{x+1}-1]$ of type B with $p_{x} \neq \reppos + 1$, 
$\bmap_{L}(p_{x})$ is the starting position of a DBWT-repetition of type B in DBWT $D_{\delta+1}^{2\alpha + 1}$. 
\end{lemma}
\begin{proof}
For $p_{x} \neq \reppos + 1$, 
whether the $\delta$-th update operation is the fast update operation or not, 
DBWT $D_{\delta+1}^{2\alpha + 1}$ has the DBWT-repetition starting at position $\bmap_{L}(p_{x}) = p_{x} + \shfn(p_{x})$ in $L_{\delta+1}$ by Lemma~\ref{lem:shiftU}, 
and the DBWT-repetition is not special character $\$$ in $L_{\delta+1}$. 
The node representing the DBWT-repetition is not node $u_{j^{\prime}+1}$, which is created in the split-node step. 
Hence, the DBWT-repetition is type B. 

\end{proof}

The following lemma gives the upper and lower bounds on $\bmap_{L}(p_{x})$. 
\begin{lemma}\label{lem:bmap_L_range}
$(p_{x} - 1) + \shfn(p_{x} - 1) < \bmap_{L}(p_{x}) \leq p_{x} + \shfn(p_{x})$ if $p_{x} \neq \reppos + 1$; 
otherwise, $(p_{x} - 2) + \shfn(p_{x} - 2) < \bmap_{L}(p_{x}) \leq p_{x} + \shfn(p_{x})$. 
Here, $\shfn(- 1)$ is defined as $0$. 
\end{lemma}
\begin{proof}
We have $(p_{x} - 1) + \shfn(p_{x} - 1) < \bmap_{L}(p_{x}) \leq p_{x} + \shfn(p_{x})$ for $p_{x} \neq \reppos + 1$ 
because $\bmap_{L}(p_{x}) = p_{x} + \shfn(p_{x})$ and $0 \leq \shfn(p_{x} - 1) \leq \shfn(p_{x}) \leq 1$. 
For $p_{x} = \reppos + 1$, 
$\bmap_{L}(p_{x}) \in \{ \reppos + \shfn(\reppos), (\reppos + 1) + \shfn(\reppos + 1) \}$. 
We have $(p_{x}-1) + \shfn(p_{x}-1) \leq \bmap_{L}(p_{x}) \leq p_{x} + \shfn(p_{x})$ 
because $0 \leq \shfn(\reppos) \leq \shfn(\reppos + 1) \leq 1$. 
$(p_{x}-1) + \shfn(p_{x}-1) > (p_{x}-2) + \shfn(p_{x}-2)$ 
because $0 \leq \shfn(p_{x}-2) \leq \shfn(p_{x}-1) \leq 1$. 
Therefore, we obtain Lemma~\ref{lem:bmap_L_range}.
\end{proof}

Similarly, function $\bmap_{F}$ maps the given starting position $\LF_{\delta}(p_{x})$ of an F-interval $[\LF_{\delta}(p_{x}), \LF_{\delta}(p_{x+1}-1)]$ 
such that its DBWT-repetition is type B in DBWT $D_{\delta}^{\alpha}$, into 
the starting position of an F-interval such that its DBWT-repetition is type B in DBWT $D_{\delta+1}^{2\alpha + 1}$. 
Formally, $\bmap_{F}(1) = 1$, 
and for a DBWT-repetition $L_{\delta}[p_{x}..p_{x+1}-1]$ of type B in $D_{\delta}^{\alpha}$,  
$\bmap_{F}(\LF_{\delta}(p_{x}))$ is defined as $\LF_{\delta+1}(\bmap_{L}(p_{x}))$. 
Function $\bmap_{L}(p_{x})$ returns the starting position of a DBWT-repetition of type B in $D_{\delta+1}^{2\alpha + 1}$, 
and LF function $\LF_{\delta+1}$ returns the starting position of the F-interval 
such that its DBWT-repetition starts at a given position in $L_{\delta+1}$. 
Hence, $\bmap_{F}$ always returns $1$ or the starting position of the F-interval on $F_{\delta+1}$ such that its DBWT-repetition is type B in $D_{\delta+1}^{2\alpha + 1}$. 

The following lemma gives the upper and lower bounds on $\bmap_{F}(\LF_{\delta}(p_{x}))$. 
\begin{lemma}\label{lem:bmap_F_range}
Let $p^{\prime}$ be the starting position of an F-interval $[\LF_{\delta}(p_{x}), \LF_{\delta}(p_{x+1}-1)]$ 
such that its DBWT-repetition is type B in DBWT $D_{\delta}^{\alpha}$~(i.e., $p^{\prime} = \LF_{\delta}(p_{x})$). 
Then, $(p^{\prime} - 1) + \shfn(p^{\prime} - 1) < \bmap_{F}(p^{\prime}) \leq p^{\prime} + \shfn(p^{\prime})$. 
\end{lemma}
\begin{proof}
For $p_{x} \neq \reppos + 1$, 
we have $\bmap_{L}(p_{x}) = p_{x} + \shfn(p_{x})$. 
The $p_{x}$-th character of BWT $L_{\delta}$ is move to the $(p_{x} + \shfn(p_{x}))$-th character of BWT $L_{\delta+1}$ by the extension of BWT. 
Similarly, the $\LF_{\delta}(p_{x})$-th character of $F_{\delta}$ is move to the $(\LF_{\delta}(p_{x}) + \shfn(\LF_{\delta}(p_{x})))$-th character of $F_{\delta+1}$, 
and thus $\LF_{\delta+1}(p_{x} + \shfn(p_{x})) = \LF_{\delta}(p_{x}) + \shfn(\LF_{\delta}(p_{x}))$~(i.e., 
$\LF_{\delta+1}(\bmap_{L}(p_{x})) = p^{\prime} + \shfn(p^{\prime})$). 
We obtain $\bmap_{F}(p^{\prime}) = p^{\prime} + \shfn(p^{\prime})$ by $\LF_{\delta+1}(\bmap_{L}(p_{x})) = \bmap_{F}(p^{\prime})$.  
$0 \leq \shfn(p^{\prime}-1) \leq \shfn(p^{\prime}) \leq 1$ always holds, 
and hence $(p^{\prime} - 1) + \shfn(p^{\prime} - 1) < \bmap_{F}(p^{\prime}) \leq p^{\prime} + \shfn(p^{\prime})$. 

For $p_{x} = \reppos + 1$, 
the $\delta$-th update operation is the fast update operation with case B or not. 
If the $\delta$-th update operation is not the fast update operation with case B, 
we have $\bmap_{L}(p_{x}) = (\reppos + 1) + \shfn(\reppos + 1)$. 
In this case, the proof of Lemma~\ref{lem:bmap_F_range} is the same as the proof for $p_{x} \neq \reppos + 1$. 
That is, we have $\bmap_{F}(p^{\prime}) = p^{\prime} + \shfn(p^{\prime})$ by the extension of BWT, 
and we obtain $(p^{\prime} - 1) + \shfn(p^{\prime} - 1) < \bmap_{F}(p^{\prime}) \leq p^{\prime} + \shfn(p^{\prime})$ 
using $0 \leq \shfn(p^{\prime}-1) \leq \shfn(p^{\prime}) \leq 1$. 

Otherwise~(i.e., the $\delta$-th update operation is the fast update operation with case B), 
we have $\bmap_{L}(p_{x}) = \reppos + \shfn(\reppos)$, which represents the starting position of the input character $c$ of the update operation in BWT $L_{\delta+1}$. 
$\LF_{\delta+1}(\reppos + \shfn(\reppos)) = \inspos$~(i.e., $\LF_{\delta+1}(\bmap_{L}(p_{x})) = \inspos$) 
because the input character $c$ is inserted into $F_{\delta}$ at position $\inspos$ by the extension of BWT.
We obtain $\bmap_{F}(p^{\prime}) = \inspos$ by $\LF_{\delta+1}(\bmap_{L}(p_{x})) = \inspos$. 
Similarly, the input character $c$ is inserted into $F_{\delta}$ at the position previous to the $\LF_{\delta}(\reppos+1)$-th character, 
and thus $\inspos = \LF_{\delta}(\reppos+1)$~(i.e., $\inspos = p^{\prime}$). 
Because $\shfn(p^{\prime} - 1) = 0$, $\shfn(p^{\prime}) = 1$, and $\bmap_{F}(p^{\prime}) = p^{\prime}$, 
we obtain $(p^{\prime} - 1) + \shfn(p^{\prime} - 1) < \bmap_{F}(p^{\prime}) \leq p^{\prime} + \shfn(p^{\prime})$.


\end{proof}

Next, we show that DBWT $D_{\delta+1}^{2\alpha + 1}$ has an $\alpha$-partition of size $d$ for 
an $\alpha$-partition $[w_{1}$, $w_{2}-1]$, $[w_{2}$, $w_{3}-1]$, $\ldots$, $[w_{d}$, $w_{d+1}-1]$ for $D_{\delta}^{\alpha}$, 
where $w_{d+1} = \delta + 1$. 
We create $d$ intervals $[\bmap_{F}(w_{1}),  \bmap_{F}(w_{2})-1]$, $[\bmap_{F}(w_{2}), \bmap_{F}(w_{3})-1]$, $\ldots$, $[\bmap_{F}(w_{d}),  \bmap_{F}(w_{d+1})-1]$ 
using function $\bmap_{F}$, 
where $\bmap_{F}(\delta + 1)$ is defined as $\delta + 2$. 
Lemma~\ref{lem:bmap_F_range} ensures that 
the $d$ intervals are a partition of $\{1, 2, \ldots, \delta +1 \}$~(i.e., $\bmap_{F}(w_{1}) = 1 < \bmap_{F}(w_{2}) < \cdots < \bmap_{F}(w_{d+1}) = \delta + 1$). 
Hence, the $d$ intervals are an $\alpha$-partition for $D_{\delta+1}^{2\alpha + 1}$ if 
each interval $[\bmap_{F}(w_{x}), \bmap_{F}(w_{x+1})-1]$ is an $\alpha$-big interval for $D_{\delta+1}^{2\alpha + 1}$.

We show that each interval $[\bmap_{F}(w_{x}), \bmap_{F}(w_{x+1})-1]$ is an $\alpha$-big interval for $D_{\delta+1}^{2\alpha + 1}$. 
From the definition of function $\bmap_{F}$, 
interval $[\bmap_{F}(w_{x}), \bmap_{F}(w_{x+1})-1]$ starts at position $1$ in $F_{\delta+1}$ or 
$D_{\delta+1}^{2\alpha + 1}$ has a DBWT-repetition of type B such that 
its F-interval starts at position $\bmap_{F}(w_{x})$ on $F_{\delta+1}$. 
Let $\mathcal{L}_{B}$ be the set of the starting positions of DBWT-repetitions of type B in $D_{\delta}^{\alpha}$ 
such that each starting position is inner covered by interval $[w_{x}, w_{x+1}-1]$. 
Then, $|\mathcal{L}_{B}| \geq \lceil \alpha / 2 \rceil - 5$ because 
$[w_{x}, w_{x+1}-1]$ is an $\alpha$-big interval for $D_{\delta}^{\alpha}$. 
For $p_{y} \in \mathcal{L}_{B}$, 
$D_{\delta+1}^{2\alpha + 1}$ has the DBWT-repetition of type B starting at position $\bmap_{L}(p_{y})$ in $L_{\delta+1}$ by 
Lemmas~\ref{lem:bmap_L_special} and \ref{lem:bmap_L}. 
Hence, interval $[\bmap_{F}(w_{x}), \bmap_{F}(w_{x+1})-1]$ is an $\alpha$-big interval for $D_{\delta+1}^{2\alpha + 1}$ 
if the interval inner covers position $\bmap_{L}(p_{y})$. 

Interval $[\bmap_{F}(w_{x}), \bmap_{F}(w_{x+1})-1]$ inner covers position $\bmap_{L}(p_{y})$ if 
either of the following two conditions holds: 
(i) $\bmap_{L}(p_{y}) \neq \inspos + 1$ and $\bmap_{L}(p_{y}) \in [\bmap_{F}(w_{x}) + 1, \bmap_{F}(w_{x+1})-1]$; 
(ii) $\bmap_{L}(p_{y}) = \inspos + 1$ and $\bmap_{L}(p_{y}) \in [\bmap_{F}(w_{x}) + 2, \bmap_{F}(w_{x+1})-1]$. 
For $\bmap_{L}(p_{y}) \neq \inspos + 1$, 
the following lemma ensures that $\bmap_{L}(p_{y}) \in [\bmap_{F}(w_{x}) + 1, \bmap_{F}(w_{x+1})-1]$. 
\begin{lemma}\label{lem:include_bmap_1}
The following two statements hold for $\bmap_{L}(p_{y}) \neq \inspos + 1$: 
(i) $\bmap_{F}(w_{x}) + 1 \leq \bmap_{L}(p_{y})$, 
and (ii) $\bmap_{L}(p_{y}) \leq \bmap_{F}(w_{x+1})-1$. 
\end{lemma}
\begin{proof}
(i) 
For $p_{y} \neq \reppos + 1$, 
we have $w_{x} < p_{y}$ because $p_{y}$ is inner covered by $[w_{x}, w_{x+1}-1]$. 
Lemmas~\ref{lem:bmap_L_range} and \ref{lem:bmap_F_range} ensures that 
there exists an integer $m$ such that $\bmap_{F}(w_{x}) \leq m + \shfn(m) < \bmap_{L}(p_{y})$. 
Hence, $\bmap_{F}(w_{x}) < \bmap_{L}(p_{y})$, i.e., $\bmap_{F}(w_{x}) + 1 \leq \bmap_{L}(p_{y})$. 

Similarly, for $p_{y} = \reppos + 1$, we have $w_{x} + 1 < p_{y}$, 
and there exists an integer $m$ such that $\bmap_{F}(w_{x}) \leq m + \shfn(m) < \bmap_{L}(p_{y})$. 
Hence, $\bmap_{F}(w_{x}) < \bmap_{L}(p_{y})$, i.e., $\bmap_{F}(w_{x}) + 1 \leq \bmap_{L}(p_{y})$. 

(ii) 
$p_{y} < w_{x+1}$ because $p_{y}$ is inner covered by $[w_{x}, w_{x+1}-1]$. 
This fact indicates that 
there exists an integer $m$ such that $\bmap_{L}(p_{y}) \leq m + \shfn(m) < \bmap_{F}(w_{x+1})$. 
Hence, $\bmap_{L}(p_{y}) < \bmap_{F}(w_{x+1})$~(i.e., $\bmap_{L}(p_{y}) \leq \bmap_{F}(w_{x+1}) - 1$). 
\end{proof}
Similarly, for $\bmap_{L}(p_{y}) = \inspos + 1$, 
the following lemma ensures that $\bmap_{L}(p_{y}) \in [\bmap_{F}(w_{x}) + 1, \bmap_{F}(w_{x+1})-1]$. 
\begin{lemma}\label{lem:include_bmap_2}
The following two statements hold for $\bmap_{L}(p_{y}) = \inspos + 1$: 
(i) $\bmap_{F}(w_{x}) + 2 \leq \bmap_{L}(p_{y})$, 
and (ii) $\bmap_{L}(p_{y}) \leq \bmap_{F}(w_{x+1})-1$. 
\end{lemma}
\begin{proof}
(i) 
Similar to Lemma~\ref{lem:include_bmap_2}-(i), 
there exists an integer $m$ such that $\bmap_{F}(w_{x}) \leq m + \shfn(m) < \bmap_{L}(p_{y})$. 
We have $m + \shfn(m) \leq \inspos < \bmap_{L}(p_{y})$ for $\bmap_{L}(p_{y}) = \inspos + 1$. 
$m + \shfn(m) \neq \inspos$ always holds from the definition of $\shfn$, 
and hence we obtain $\bmap_{F}(w_{x}) \leq m + \shfn(m) < \inspos < \bmap_{L}(p_{y})$, 
which indicates that $\bmap_{F}(w_{x}) + 2 \leq \bmap_{L}(p_{y})$. 

(ii) We can use the proof of Lemma~\ref{lem:include_bmap_1} for proving this statement, 
and hence, we obtain $\bmap_{L}(p_{y}) \leq \bmap_{F}(w_{x+1})-1$. 

\end{proof}
Hence, interval $[\bmap_{F}(w_{x}), \bmap_{F}(w_{x+1})-1]$ is an $\alpha$-big interval for $D_{\delta+1}^{2\alpha + 1}$, 
which indicates that the $d$ intervals $[\bmap_{F}(w_{1}),  \bmap_{F}(w_{2})-1]$, $[\bmap_{F}(w_{2}), \bmap_{F}(w_{3})-1]$, $\ldots$, 
$[\bmap_{F}(w_{d}),  \bmap_{F}(w_{d+1})-1]$ are an $\alpha$-partition for $D_{\delta+1}^{2\alpha + 1}$. 
Finally, we obtain Lemma~\ref{lem:prt_f_sub}-(ii).

\paragraph{Proof of Lemma~\ref{lem:kL_and_kF}-(ii).}
The number $k_{L}$ is symmetric to the number $k_{F}$, 
and hence we can prove Lemma~\ref{lem:kL_and_kF}-(ii) by modifying the proof of Lemma~\ref{lem:kL_and_kF}-(i).

\section{Details for Section~\ref{sec:exp}}
\subsection{Space reduction technique}\label{app:space_reduction_technique}

\newcommand{\Ceil}[1]{\lceil{#1}\rceil}
\newcommand{\Floor}[1]{\lfloor{#1}\rfloor}

\begin{figure}[tb]
\begin{center}
\includegraphics[scale=0.7]{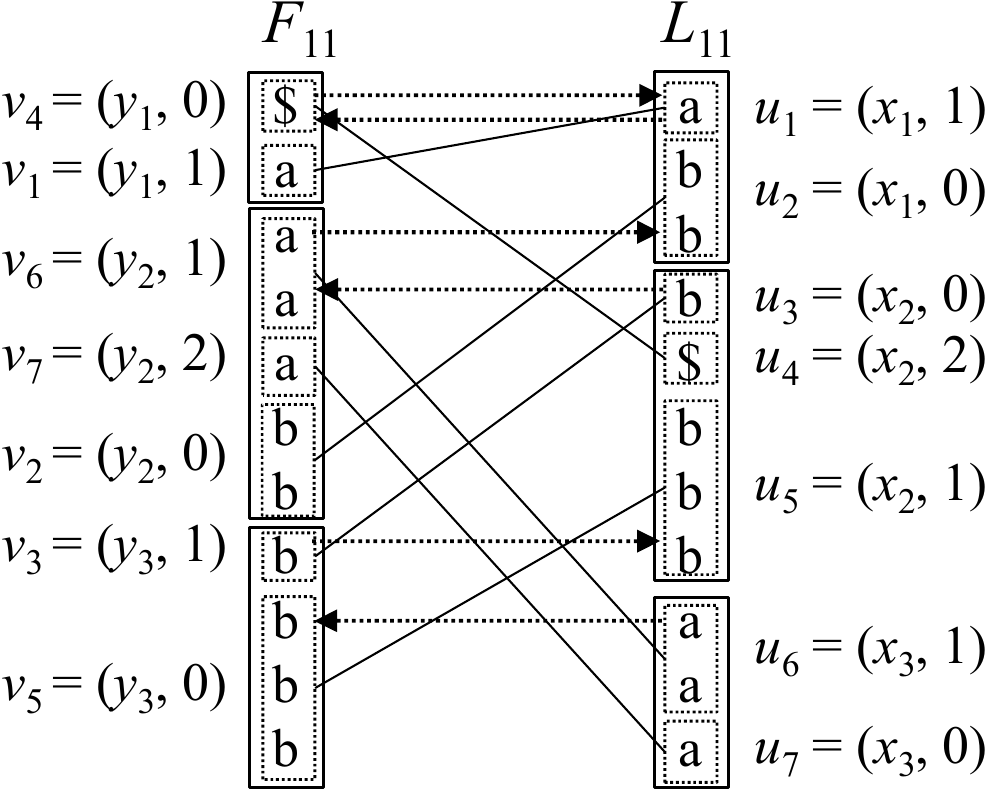}
\caption{
Grouped LF-interval graph for DBWT $D_{11}$ in Figure~\ref{fig:dbwt}. 
The solid-line rectangles depict groups for $U = \{ u_{1}, u_{2}, \ldots, u_{11} \}$ and $V = \{ v_{1}, v_{2}, \ldots, v_{11} \}$. 
The dashed-line rectangles depict nodes in $U$ and $V$. 
Pointer $x_{i}$~(respectively, $y_{i}$) is a pointer to the $i$-th group in the doubly linked list of groups for $U$~(respectively, $V$). 
Here, $u_{i} = (x_{j}, q)$ means that the identifier of the node $u_{i} \in U$ is $(x_{j}, q)$. 
Similarly, $v_{i} = (y_{j}, q)$ means that the identifier of the node $v_{i} \in V$ is $(y_{j}, q)$.
}
\label{fig:grouping}
\end{center}
\end{figure}

The main component of LF-interval graph $\graphds(D^{\alpha}_{\delta}) = (U \cup V$, $E_{LF} \cup E_{L} \cup {E}_{F}$, $B_{U} \cup B_{V}$,  $B_{L} \cup B_{F})$ is two doubly linked lists of two sets $U$ and $V$ of nodes.
In a straightforward manner, each node $u_i \in U$ stores the following four pointers, two integers, and one character:
(i) two pointers to previous node $u_{i-1} \in U$ and next node $u_{i+1} \in U$ in the doubly linked list of $U$~($2 \lceil \log n \rceil$ bits); 
(ii) a pointer to the node $v_{i}$ connected to $u_{i}$ by undirected edge $(u_i, v_i) \in E_{LF}$~($\lceil \log n \rceil$ bits); 
(iii) a pointer to the node $v_j$ connected to $u_{i}$ by directed edge $(u_i, v_j) \in E_{L}$~($\lceil \log n \rceil$ bits); 
(iv) the integer label $B_L(u_i, v_j)$ of directed edge $(u_i, v_j)$~($\lceil \log n \rceil$ bits); and 
(v) the label $B_U(u_i)$ of node $u_{i}$~($\lceil \log n \rceil + \lceil \log \sigma \rceil$ bits). 
Similarly, each node $v_i \in V$ has the equivalent members.
Thus, each node in the two doubly linked lists consumes $6 \Ceil{\log n} + \Ceil{ \log \sigma }$  bits of space, and the constant factor can be a space bottleneck in practice.
In this section, we present a space reduction technique called the \emph{grouping technique}.

The idea behind the grouping technique is to partition the $k$ nodes $u_1, u_2, \dots, u_k$ of $U$ 
into $m$ sequences $(u_{t_{1}}, u_{t_{1}+1}, \ldots, u_{t_{2}-1}), (u_{t_{2}}, u_{t_{2}+1}, \ldots, u_{t_{3}-1})$, 
$\ldots$, $(u_{t_{m}}, u_{t_{m}+1}, \ldots, u_{k})$~($t_{1} = 1 < t_{2} < \cdots < t_{m} \leq k$), 
and we reduce the number of pointers using the sequences of nodes. 
Similarly, we partition the $k$ nodes $v_{\pi_{1}}, v_{\pi_{2}}, \dots, v_{\pi_{k}}$ of $V$ 
into $m^{\prime}$ sequences $(v^{\prime}_{t^{\prime}_{1}}, v^{\prime}_{t^{\prime}_{1}+1}, \ldots, v^{\prime}_{t^{\prime}_{2}-1})$, $(v^{\prime}_{t^{\prime}_{2}}, v^{\prime}_{t^{\prime}_{2}+1}, \ldots, v^{\prime}_{t^{\prime}_{3}-1})$, $\ldots$, $(v^{\prime}_{t^{\prime}_{m^{\prime}}}, v^{\prime}_{t^{\prime}_{m^{\prime}}+1}, \ldots, v^{\prime}_{t^{\prime}_{k}})$~($t^{\prime}_{1} = 1 < t^{\prime}_{2} < \cdots < t^{\prime}_{m^{\prime}} \leq k$), 
where $v^{\prime}_{i}$ is defined as $v_{\pi_{i}}$ for all $i \in \{ 1, 2, \ldots, k \}$ 
and $\pi$ is the permutation introduced in Section~\ref{sec:dynamic}. 
Each sequence of nodes is called a \emph{group},  
and a group consists of at most $g$ nodes for a new constant parameter $g$.

We define the \emph{identifier} of a node $u_{i} \in U$ for finding the node in groups for $U$. 
The identifier of $u_{i} \in U$ is a pair of (i) a pointer to the group containing the node $u_{i}$
and (ii) an integer in $\{ 1, 2, \ldots, g \}$ such that the identifiers of all the nodes in the group are different. 
Similarly, the identifier of a node $v_{i} \in V$ is defined in a similar way.

Next, we present a new data structure called the \emph{grouped LF-interval graph} for storing LF-interval graph $\graphds(D^{\alpha}_\delta)$ in a small space. 
The grouped LF-interval graph consists of two doubly linked lists of groups for $U$ and $V$. 
The two doubly linked lists of groups for $U$ and $V$ store the $m$ groups for $U$ and the $m^{\prime}$ groups for $V$, respectively. 
Each group $(u_{t_{h}}, u_{t_{h}+1}, \ldots, u_{t_{h+1}-1})$ of nodes in $U$ is represented by a quadruple of the following four elements: 
(i) the identifier of the node $v_{i}$ connected to the first node $u_{h}$ of the group by 
directed edge $(u_{h}, v_{i}) \in E_{L}$; 
(ii) the label $B_{L}(u_{h}, v_{i})$ of directed edge $(u_{h}, v_{i})$; 
(iii) two pointers to the previous group $(u_{t_{h-1}}, u_{t_{h-1}+1}, \ldots, u_{t_{h}-1})$ and 
the next group $(u_{t_{h+1}}, u_{t_{h+1}+1}, \ldots, u_{t_{h+2}-1})$ in the doubly linked of groups for $U$; and
(iv) an array storing all the nodes in the group in the order of $u_{t_{h}}, u_{t_{h}+1}, \ldots, u_{t_{h+1}-1}$.  
Each node $u_{i}$ of the array is represented by a triplet of (1) the label $B_{U}(u_{i})$ of $u_{i}$, 
(2) the identifier of the node $v_{i}$ connected to $u_{i}$ by edge $(u_i, v_i) \in E_{LF}$, 
and (3) the integer included in the identifier of $u_{i}$.

Similarly, each group $(v^{\prime}_{t^{\prime}_{h}}, v^{\prime}_{t^{\prime}_{h}+1}, \ldots, v^{\prime}_{t^{\prime}_{h+1}-1})$ of nodes in $V$ is represented by a triplet of the following three elements: 
(i) the identifier of the node $u_{i}$ connected to the first node $v^{\prime}_{t^{\prime}_{h}}$ of the group by 
directed edge $(v^{\prime}_{t^{\prime}_{h}}, u_{i}) \in E_{V}$; 
(ii) two pointers to the previous group $(v^{\prime}_{t^{\prime}_{h-1}}, v^{\prime}_{t^{\prime}_{h-1}+1}, \ldots, v^{\prime}_{t^{\prime}_{h}-1})$ and 
the next group $(v^{\prime}_{t^{\prime}_{h+1}}, v^{\prime}_{t^{\prime}_{h+1}+1}, \ldots, v^{\prime}_{t^{\prime}_{h+2}-1})$ in the doubly linked of groups for $V$; and
(iii) an array storing all the nodes in the group in the order $v^{\prime}_{t^{\prime}_{h}}, v^{\prime}_{t^{\prime}_{h}+1}, \ldots, v^{\prime}_{t^{\prime}_{h+1}-1}$. 
Each node $v_{i}$ of the array is represented by a pair consisting of (i) the identifier of the node $u_{i}$ connected to $v_{i}$ by edge $(u_i, v_i) \in E_{LF}$, and (ii) the integer included in the identifier of $u_{i}$.

Each group without its array consumes at most $5 \Ceil{\log n} + \Ceil{\log g}$ bits, 
and each node consumes at most $2 \Ceil{\log n} + 2 \Ceil{\log g} + \Ceil{\log \sigma}$ bits. 
Hence, the grouped LF-interval graph consumes at most $(m + m^{\prime})(5 \Ceil{\log n} + \Ceil{\log g}) + 2k(2 \Ceil{\log n} + 2 \Ceil{\log g} + \Ceil{\log \sigma})$ bits.

This grouping technique improves the constant factors of the space usages of nodes in $U$ and $V$. 
It is not necessary for each node $u_{i} \in U$~(respectively, $v_{\pi_{i}} \in V$) 
to store the three pointers and one integer for computing 
the previous node $u_{i-1}$~(respectively, $v_{\pi_{i-1}}$), the next node $u_{i+1}$~(respectively, $v_{\pi_{i+1}}$), 
the directed edge starting at $u_{i}$~(respectively, $v_{\pi_{i}}$), and the label of the directed edge. 
Instead, each node consumes $2 \Ceil{ \log g }$ bits for the integers included in the two additional identifiers, 
but the overhead is slight when $g \ll n$. 

We represent LF-interval graph $\graphds(D^{\alpha}_{\delta})$ using the grouped LF-interval graph for $D^{\alpha}_{\delta}$. 
This modification reduces the working space of r-comp for appropriately chosen $g$. 

Figure~\ref{fig:grouping} shows an example of the grouped LF-interval graph with $g = 3$ for DBWT $D_{11}$ in Figure~\ref{fig:dbwt} 
(see also the LF-interval graph $\graphds(D_{11})$ illustrated in Figure~\ref{fig:graph}). 
The 14 nodes in $U$ and $V$ are partitioned into six groups, 
and the number of previous and next pointers for the doubly linked lists for $U$ and $V$ is reduced from 28 to 12. 
Similarly, the number of pointers for directed edges in $E_{L}$ and $E_{F}$ is reduced from 14 to six.

\subparagraph{Acceleration.}
The grouping technique increases the time to get nodes and edges in the LF-interval graph 
because we need to recover nodes and edges from the grouped LF-interval graph. 
Especially, searching for a node in a group takes $O(g)$ time for the given identifier of the node. 
The order of the computation time is constant, but 
the search can be a bottleneck in practice for a large constant $g$. 
For this reason, we introduce a \emph{lookup table}. 
The lookup table of a group is an array of size $g$, 
and the $i$-th element of the array stores a pointer to the node whose identifier is a pair consisting of a pointer to the group and $i$. 
The lookup table uses $g \Ceil{\log g} $ bits of space for each group, and the space overhead is not a bottleneck when $g \ll n$. 
The computation time of the search for a node in a group can be reduced to $O(1)$ time using the lookup table, 
and hence r-comp is expected to run in a practical time.

\end{document}